\documentclass[acmsmall, screen, nonacm]{acmart}

\AtBeginDocument{%
  }

\setcopyright{acmlicensed}
\copyrightyear{2018}
\acmYear{2018}
\acmDOI{XXXXXXX.XXXXXXX}

\usepackage{ltl}
\usepackage{stmaryrd}
\usepackage{cleveref}
\usepackage{bussproofs}
\usepackage{stackengine}

\usepackage[most]{tcolorbox}

\usepackage{thm-restate}

\newcommand{\best}[1]{\textbf{#1}}
\newcommand{\vbes}[1]{\textbf{#1}}

\newcommand{\Nat}{\ensuremath{\mathbb{N}}}

\newcommand{\power}[1]{\ensuremath{2^{#1}}}
\newcommand{\sema}[1]{\llbracket #1 \rrbracket}

\newcommand{\values}{\ensuremath{\mathcal{V}}}

\newcommand{\vars}{\ensuremath{\mathit{Vars}}}

\newcommand{\assignment}{\ensuremath{\nu}}
\newcommand{\assignments}[1]{\ensuremath{\mathit{Assignments}}(#1)}

\newcommand{\assmt}[1]{\mathbf{#1}}

\newcommand{\FOL}[1]{\ensuremath{\mathit{FOL}}(#1)}
\newcommand{\FOLX}{\ensuremath{\mathit{FOL}}}
\newcommand{\QF}[1]{\ensuremath{\mathit{QF}}(#1)}
\newcommand{\QFX}{\ensuremath{\mathit{QF}}}
\newcommand{\FOLentails}{\ensuremath{\models}}
\newcommand{\FOLentailsT}[1]{\ensuremath{\models_{#1}}}

\newcommand{\states}{\ensuremath{\mathcal{S}}}

\newcommand{\sys}{\ensuremath{\mathit{Sys}}}
\newcommand{\env}{\ensuremath{\mathit{Env}}}

\newcommand{\plays}{\mathit{Plays}}
\newcommand{\win}{\mathit{Win}}

\newcommand{\parity}{\mathit{Parity}}

\newcommand{\symgame}{\mathcal{G}}

\newcommand{\dom}{\ensuremath{\mathit{dom}}}

\newcommand{\Loc}{\mathit{Loc}}

\newcommand{\linit}{l_{\mathit{init}}}

\newcommand{\strat}{\sigma}

\newcommand{\templogic}{\ensuremath{\mathit{RP{-}LTL}}}

\newcommand{\progvars}{\ensuremath{\mathbb{X}}}
\newcommand{\inputs}{\ensuremath{\mathbb{I}}}
\newcommand{\specvars}{\ensuremath{\progvars \cup \inputs \cup \progvars'}}

\newcommand{\true}{\ensuremath{\top}}
\newcommand{\false}{\ensuremath{\bot}}

\newcommand{\lang}[1]{\mathcal{L}(#1)}

\newcommand{\tobool}[1]{\widehat{#1}}

\newcommand{\atoms}{\ensuremath{\mathit{Atoms}}}
\newcommand{\AP}{\ensuremath{\mathit{AP}}}

\newcommand{\tslmt}{\text{TSL-MT}}

\newcommand{\preds}{\ensuremath{\mathcal{P}}}

\newcommand{\qinit}{q_{\mathit{init}}}

\newcommand{\UNSAT}{\ensuremath{\mathsf{UNSAT}}}
\newcommand{\OPEN}{\ensuremath{\mathsf{OPEN}}}
\newcommand{\SAFETY}{\ensuremath{\mathsf{SAFETY}}}

\newcommand{\verdict}{\ensuremath{\mathit{Verdict}}}
\newcommand{\Verdicts}{\{\UNSAT,\SAFETY,\OPEN\}}

\newcommand{\as}{\mathsf{A}}
\newcommand{\ga}{\mathsf{G}}
\newcommand{\agset}{\{\as,\ga\}}

\newcommand{\formF}{F}
\newcommand{\formE}{E}
\newcommand{\Imp}{\ensuremath{\mathit{Imp}}}

\newcommand{\Curr}{\ensuremath{\mathit{Curr}}}
\newcommand{\DedInv}{\ensuremath{\mathit{ImpInv}}}

\newcommand{\toform}[1]{\llparenthesis #1 \rrparenthesis}

\newcommand{\subform}{\ensuremath{\mathit{Closure}}}
\newcommand{\boolcomb}{\ensuremath{\mathit{BoolComb}}}
\newcommand{\BC}{\ensuremath{\mathit{BC}}}

\newcommand{\stateformula}{\ensuremath{\mathit{Formula}}}

\newcommand{\transform}{\textsc{Transform}}

\newcommand{\ruleset}{\ensuremath{\templogic}$-$\mathit{Rules}}

\newcommand{\rulename}[1]{\textbf{\textsc{(#1)}}}

\newcommand{\nextextend}{\rulename{$\LTLnext$-Ext}}

\newcommand{\substunsat}{\rulename{UNSAT}}
\newcommand{\substunsatF}{\rulename{UNSAT-$\LTLfinally$}}

\newcommand{\substtrue}{\rulename{Subst-$\true$}}
\newcommand{\substfalse}{\rulename{Subst-$\false$}}

\newcommand{\simplifyimp}{\rulename{Simplify-$\rightarrow$}}
\newcommand{\simplifyand}{\rulename{Simplify-$\land$}}
\newcommand{\simplifynn}{\rulename{Simplify-Non-Nested}}

\newcommand{\propagateassump}{\rulename{Propagate-Assump}}
\newcommand{\propagateG}{\rulename{Propagate-$\LTLglobally$}}
\newcommand{\propagateW}{\rulename{Propagate-$\LTLweakuntil$}}

\newcommand{\geninv}{\rulename{Gen-Inv}}

\newcommand{\geninvp}{\rulename{Gen-Inv-P}}
\newcommand{\genreach}{\rulename{Gen-Reach}}

\newcommand{\chainimp}{\rulename{Chain-Imp}}
\newcommand{\chainimpG}{\rulename{Chain-Imp-$\LTLglobally$}}
\newcommand{\chainimpF}{\rulename{Chain-Imp-$\LTLfinally$}}
\newcommand{\chainimpN}{\rulename{Chain-Imp-$\LTLnext$}}

\newcommand{\joinimp}{\rulename{Join-Imp}}

\newcommand{\mapstonn}{\mapsto_{non-nested}}

\newcommand{\expand}{\ensuremath{\mathit{expand}}}
\newcommand{\expandSet}{\ensuremath{\mathit{exps}}}

\newcommand{\propagationPreds}{\ensuremath{\mathit{PropPreds}}}
\newcommand{\propagate}{\ensuremath{\mathit{propagate}}}
\newcommand{\generatedPreds}{\ensuremath{\mathit{GenPreds}}}
\newcommand{\applyRules}{\ensuremath{\mathit{applyRules}}}
\newcommand{\nextState}{\ensuremath{\mathit{nextState}}}

\newcommand{\toolname}{\texttt{tslmt2rpg}}
\newcommand{\upd}[2]{\ensuremath{\llbracket#1\leftarrowtail#2\rrbracket}}

\newcommand{\plang}[1]{\mathcal{L}_P(#1)}
\newcommand{\uniquerun}{\mathit{uniqueRun}}
\newcommand{\xiruns}{{(\assignments{\progvars \cup \inputs})}^\omega}
\newcommand{\apruns}{{(\power{\AP(\tobool{\varphi})})}^\omega}
\newcommand{\concretize}{\mathit{concretize}}

\begin{document}

\title{Translation of Temporal Logic for Efficient Infinite-State Reactive Synthesis (Full Version)}


\author{Philippe Heim}
\orcid{0000-0002-5433-8133}
\affiliation{%
  \institution{CISPA Helmholtz Center for Information Security}
  %\city{}
  \country{Germany}
}
\email{philippe.heim@cispa.de}

\author{Rayna Dimitrova}
\orcid{0009-0006-2494-8690}
\affiliation{%
  \institution{CISPA Helmholtz Center for Information Security}
  %\city{}
  \country{Germany}
}
\email{dimitrova@cispa.de}

%\renewcommand{\shortauthors}{Anonymous et al.}

\begin{abstract}
Infinite-state reactive synthesis has attracted significant attention in recent years,  which has led to the emergence of novel symbolic techniques for solving infinite-state games.
Temporal logics featuring variables over infinite domains offer an expressive high-level specification language for infinite-state reactive systems.  
Currently, the only way to translate these temporal logics into symbolic games is by naively encoding the specification to use techniques designed for the Boolean case.
An inherent limitation of this approach is that it results in games in which the semantic structure of the temporal and first-order constraints present in the formula is lost.
There is a clear need for techniques that leverage this information in the translation process to speed up solving the generated games. 

In this work, we propose the first approach that addresses this gap. 
Our technique constructs a monitor incorporating first-order and temporal reasoning at the formula level, enriching the constructed game with semantic information that leads to more efficient solving.  We demonstrate that thanks to this,  our method outperforms the state-of-the-art techniques across a range of benchmarks.
\end{abstract}

%\keywords{Do, Not, Us, This, Code, Put, the, Correct, Terms, for,Your, Paper}

\maketitle

\section{Introduction}\label{sec:intro}
Program synthesis,  that is,  the automatic construction of programs from high-level specifications, is a classical problem in computer science.
Reactive systems,  characterized by their ongoing, unbounded interaction with a potentially adversarial environment are typically specified  using linear temporal logic (LTL)~\cite{Pnueli77}.
Then, the task of \emph{reactive synthesis}~\cite{PnueliR89} is to synthesize a system implementation that reacts to all possible inputs from the environment in a way that ensures that all possible resulting executions satisfy the given specification.
Traditionally,  reactive synthesis methods have focused on systems with Boolean variables, such as finite-state controllers,  protocols, or hardware circuits.
However, the prevalence of programs that integrate complex control mechanisms and data transformations has led to the recent active development of specification languages and synthesis techniques for reactive systems with domains beyond Boolean.
Temporal Stream Logic (TSL)~\cite{FinkbeinerKPS19} is an extension of LTL with  updates and predicates over arbitrary function terms, capable of expressing requirements on data transformations.
An implementation is required to satisfy the TSL specification for all possible instantiations of the data processing functions.
Restricting the set of possible instantiations requires adding explicit assumptions to the specification,  which increases its complexity.
This has lead to the development of TSL modulo Theories (TSL-MT)~\cite{FinkbeinerHP22},   which extends TSL with first-order theories,  and thus enables the use of interpreted functions and predicate symbols.
The synthesis from TSL-MT specifications has been studied in~\cite{ChoiFPS22, MaderbacherB22}. These techniques extend the approach to TSL synthesis proposed in~\cite{FinkbeinerKPS19}, which is based on abstracting the TSL specification into a propositional LTL formula and invoking a method for finite-state reactive synthesis.
Thus, the control synthesis task is off-loaded to the finite-state reactive synthesis method, and Syntax-Guided Synthesis (SyGuS)~\cite{AlurBJMRSSSTU13} or SMT-based analysis are used for reasoning about data.
These methods implement the interaction between the reactive synthesis and the data layer as a refinement loop that adds assumptions to the propositional LTL formula.  
This can increase the size of the specification significantly.
Furthermore, it requires solving a finite-state reactive synthesis instance from scratch with a new specification every time assumptions are added.

To address the above limitation of the techniques in~\cite{ChoiFPS22, MaderbacherB22},  a direct approach to solving reactive synthesis games over non-Boolean domains was proposed recently in~\cite{HeimD24}.
The procedure in~\cite{HeimD24} circumvents the abstraction by lifting the game-solving  methods to a symbolic representation of the infinite-state game.  However,  the method and tool from~\cite{HeimD24} assume that the synthesis problem is directly specified as a so-called \emph{reactive program game}.
In principle,  an \tslmt\ formula $\varphi$ can be translated into a \emph{reactive program game} by first converting $\varphi$ to a propositional LTL formula $\tobool{\varphi}$ and then applying the standard procedure of constructing a deterministic $\omega$-automaton for $\tobool{\varphi}$ and the resulting synthesis game~\cite{Duret-LutzRCRAS22}.  Finally,  the first-order predicates from $\varphi$ are reintroduced to obtain a reactive program game.
A key drawback of this workflow is that once the specification is Booleanized into $\tobool{\varphi}$, the subsequent game construction does not take advantage of any theory reasoning.  One of the reasons is that the first-order formulas represented by the propositions are hidden from this construction,  but, more crucially,  we are currently lacking techniques that are able to make use of this semantic information on the temporal formula level.

In this paper,  we aim to fill this gap. 
We propose the first approach that augments with first-order reasoning the translation of first-order temporal logic specifications to two-player games encoding the reactive program synthesis problem. 
We introduce a new logic, called \emph{Reactive Program Linear Temporal Logic} (\templogic), which generalizes \tslmt\ by allowing constraints on ``next-step'' variables.  
Then,  we define the notion of a \emph{monitor for an \templogic\ formula} which provides a principled way of adding semantic information to the synthesis games constructed from \templogic\ formulas. 
This approach is \emph{generally applicable},  and independent of the game construction,  the monitor construction, and the method used subsequently for solving the resulting synthesis game.  
Our main contribution is a procedure for the  construction of a monitor for an \templogic\ specification.
Inspired by  constructions for the translation
of LTL to automata~\cite{EsparzaKS20, Vardi94},  the states of the monitor are collections of \templogic\ formulas, which enables semantics-based simplification operations.
A key distinguishing feature of our approach is that the monitor construction integrates temporal and first-order reasoning in order to perform extensive simplifications.  More concretely,  it establishes invariants and ranking arguments (by employing SMT and CHC solvers,  as well as fixpoint computation engines) that enable such simplifications for non-trivial specifications.

\paragraph{Contributions}
Our contributions  can be summarized as follows.
\begin{itemize}
\item We introduce the \emph{temporal logic \templogic\ } for the specification of reactive programs whose variables have domains beyond  Boolean.
\item  We define the notion of \emph{monitor for an \templogic\ formula $\varphi$}, which can be combined with any symbolic game for $\varphi$ to simplify the resulting synthesis game.
\item We describe a procedure for constructing a monitor for a  given \templogic\ formula. The monitor construction employs first-order and temporal logical reasoning on the formula level to derive information about the (un)-satisfiability of sub-formulas that enables the simplification of the resulting synthesis game.
\item We demonstrate that a prototype implementation of our method leads to performance improvement over the state of the art across a range of benchmarks.
\end{itemize}
The proofs of all formal claims can be found in \Cref{sec:approofs}.

In the next section, we give a high-level overview of our approach and describe several motivating examples demonstrating its use cases and the challenges we address.

\section{Overview and Motivating Examples}\label{sec:examples}
We study the problem of synthesizing reactive programs operating over variables with potentially unbounded domains.  
Since reactive systems execute in an ongoing interaction with their environment,  the requirements that the synthesized system must satisfy are specified in a temporal logic.  
We consider one such specification language,  termed  \emph{Reactive Program Linear Temporal Logic} (\templogic), which generalizes LTL.
To express properties of  the data transformations the program performs over time, we equip the logic with quantifier-free atomic formulas in a given first-order logical theory. 
These formulas express constraints over the program's \emph{input variables} $\inputs$ and the \emph{program variables} $\progvars$ (which include the program's outputs).  Furthermore,  formulas can refer to the variables $\progvars'$ (which are ``primed'' copies of the variables $\progvars$) representing the value of the program variables at the next step of the execution.  The latter feature allows us to express variable assignments and to relate the values of variables over multiple time steps. 

As a simple example,  consider the \templogic\ formula
\[
\LTLglobally (e > 0 \rightarrow x' \leq  x+2 \land \LTLnext\LTLnext(x + y > 10))
\land 
\LTLeventually (x \geq 42)\]
over input variable $e$ and program variables $x$ and $y$, all ranging over the integers.
Here, $\LTLglobally$ is the LTL \emph{``globally''} operator stating that its argument should hold forever on, from the current step of the execution. 
The LTL \emph{``eventually''} operator  $\LTLeventually$ states that its argument should hold eventually at some point in the future. 
The remaining temporal operator $\LTLnext$ (\emph{``next''}) refers to the next step of the execution.
The above specification requires that whenever the value of the input variable $e$ is positive,  the value of $x$ in the next time point should be at most the current value of $x$ plus $2$,  and that in two time steps, $x+y > 10$ should hold. The conjunct $\LTLeventually (x \geq 42)$ requires that $x \geq 42$ is eventually true.
The program that regardless of the input always increments $x$ by $2$ and assigns the value $10-x$ to $y$, is one possible program that satisfies this specification for any initial $x$ and $y$.

The problem of checking if an \templogic\ formula $\varphi$ is realizable and if so synthesizing an implementation, is undecidable.
We can obtain a sound but incomplete method by adapting the classical automata-theoretic approach for LTL synthesis.  
To this end,  we transform $\varphi$ into an LTL formula $\tobool{\varphi}$ by replacing all atomic formulas (such as $e >0$,  $x' \leq x+2$,  and so on)  with Boolean propositions. 
We then use  $\tobool{\varphi}$ to construct a two-player game between the system and the environment.  In the resulting game, we replace back the propositions with the original formulas, thus obtaining a game that encodes the synthesis problem for $\varphi$.  This is an infinite-state game represented symbolically.  A number of techniques and some tools exist for solving (classes of) such games.  
In \Cref{sec:experiments} we reference one such tool for solving a class of games called \emph{reactive program games}.

The outlined reduction of the synthesis problem for \templogic\ to the problem of solving an infinite-state game has the fundamental limitation that the process of constructing the game is oblivious to the first-order formulas that the propositions represent.  Due to that, any reductions and simplifications that could take advantage of this information are not possible. 
Next in this section,  we present several examples where such semantic reductions can have a significant impact on the resulting game, and hence,  on the capability of game-solving techniques to fulfil their task.

\subsection{Motivating Examples}
 
\subsubsection{Detecting Unsatisfiable Specifications and Vacuous Requirements}

Our first two examples are of formulas that contain requirements that are unsatisfiable, or that are implied by other requirements in the specification.
As identifying this fact requires reasoning about first-order constraints in a temporal context,  it remains undetected in the translation outlined above. 
This results in synthesis games where recovering this information is hard for the game-solving procedures. This is also evidenced in our experimental evaluation presented in \Cref{sec:experiments}.

\begin{example}\label{ex:unsat-motivating}
Consider the following specification of a reactive program with input variable $e$ and program variables $x$ and $y$,  all of which have integer type:
\[\begin{array}{ll}
\varphi_{\sf unsat} := e > 0  \land x = 0 \rightarrow & 
\big(
x' = 0 \land \LTLnext\LTLglobally(x' = x + 1 \lor x' = x+y)  \land
\LTLeventually(x \leq  -10000) \land 
\\&\phantom{\big(}
 y' = e \land \LTLnext\LTLglobally(y' = y) \land \LTLnext(y > 0) 
 \big).
\end{array}
\]
The formula $\varphi_{\sf unsat}$ contains the assumption that at the first step the input $e$ is positive and $x=0$.
It requires that initially $x$ is assigned value $0$ (conjunct $x' = 0$) and from the next step on,  either $x$ is incremented by $1$ or $y$ is added to $x$ (conjunct $\LTLnext\LTLglobally(x' = x + 1 \lor x' = x+y$).
Further, it contains the requirement that eventually $x \leq  -10000$ holds, which could only be satisfied if $y$ is at some point negative after the first step. This, however contradicts the requirement that $y$ is initially assigned the positive value of $e$ (that is, $y'=e$), and from the next step on,  $y$ should remain unchanged (conjunct $\LTLnext\LTLglobally(y' = y)$).  Thus,   $\varphi_{\sf unsat}$ is unsatisfiable,  and hence unrealizable.
This, however, remains undetected when constructing the synthesis game for the propositional LTL formula $\tobool{\varphi_{\sf unsat}}$.

Establishing the above observations  formally, can be achieved by showing that in any reactive program where the assumption $e > 0 \land x=0$ and the requirements 
\[
x' = 0 \land \LTLnext\LTLglobally(x' = x + 1 \lor x' = x+y)  \land
y' = e \land \LTLnext\LTLglobally(y' = y) \land \LTLnext(y > 0) 
\]
are satisfied,  the property $x \geq -9999 \land y \geq 1$ is an \emph{inductive invariant} from the second step on.
\end{example}

The reasoning in \Cref{ex:unsat-motivating} is also applicable  in cases where some sub-formula is unsatisfiable,  but the overall specification can be realized by an appropriate program.  Then, we would like to simplify the synthesis game by construction, pruning away ``losing'' choices for the system.

\begin{example}\label{ex:vacuous-motivating}
Consider the following specification of a reactive program with input variable $e$ and program variables $x$ and $y$,  all of which have integer type.
\[\begin{array}{ll}
\varphi_{\sf vac} := & 
\LTLglobally(y > 0 \rightarrow x' = y) \land
\LTLglobally\big(\neg(y > 0) \rightarrow x' = x+ 1 -y\big) \land
\LTLfinally(x > 10000)
\land \\&
\LTLglobally(e > 10000 \rightarrow y' = e) \land
\LTLglobally\big(\neg(e > 10000) \rightarrow y' = 0\big).
\end{array}
\]
The formula $\varphi_{\sf vac}$ requires that whenever the value of program variable $y$ is positive,  $x$ is assigned $y$,
and otherwise, $x$ is assigned $x + 1 -y$.
Further, it contains the requirement that eventually $\LTLfinally(x > 10000)$.
This latter requirement is \emph{vacuous} in the context of $\varphi_{\sf vac}$, that is,  it is implied by the other requirements.
To see this, note that the requirements on $y$ imply that for any value of the input variable $e$,  the only way $y$ could be assigned a positive value is if this value is greater than $10000$,  and otherwise $y$ is assigned $0$. 
This,  together with  the requirements on $x'$ implies that either $x$ is eventually assigned a value greater than $10000$, or it keeps being incremented. 
In both cases eventually $x > 10000$ must hold.

In the construction of the synthesis game for the formula $\varphi_{\sf vac}$,  we would like to simplify the game by making use of the information that some requirement is vacuous. This, however, remains undetected when considering the propositional formula $\tobool{\varphi_{\sf vac}}$.

The fact that the requirement $\LTLfinally(x > 10000)$ is vacuous in the context of $\varphi_{\sf vac}$ can be established by showing that  every execution in a reactive program satisfying 
\[
\begin{array}{l}
\LTLglobally(y > 0 \rightarrow x' = y) \land
\LTLglobally\big(\neg(y > 0) \rightarrow x' = x+ 1 -y\big) \land 
\\
\LTLglobally(e > 10000 \rightarrow y' = e) \land
\LTLglobally\big(\neg(e > 10000) \rightarrow y' = 0\big)
\end{array}
\]
eventually reaches a state satisfying $x > 10000$,  that is by
proving a \emph{reachability property}.
\end{example}

\subsubsection{Winning Condition Simplification}

How complex the synthesis game is,  depends not only on the game structure, but also on the winning condition, that is on the type of property that the system player must enforce.  
Invariant properties are preferable to more complex ones such as, for example, repeated reachability. 
In our next example,  a repeated reachability requirement is implied by the conjuncts of the formula that specify invariant properties. Detecting this enables a simplification of the resulting game.

\begin{example}\label{ex:GF-motivating}
Consider the following specification of a reactive program with input variable $e$ and program variables $x$ and $c$,  all of which have integer type.
\[\begin{array}{ll}
\varphi_{\sf simplify} := & 
\LTLglobally(c' = 0 \lor c' =1)  \land \LTLglobally(\LTLnext(c = 1)  \rightarrow e =1) \land \\&
\LTLglobally(c=0 \rightarrow x' = x+ 1) \land 
\LTLglobally\big(\neg(c=0) \rightarrow x' = 10000\big) \land 
\LTLglobally \LTLfinally(x \geq 10000).
\end{array}
\]
The formula $\varphi_{\sf simplify}$ requires that infinitely often $x \geq 10000$ must hold (conjunct $\LTLglobally \LTLfinally(x \geq 10000)$), which translates to a B\"uchi winning condition in the resulting infinite-state synthesis game. 
However, a closer look at $\varphi_{\sf simplify}$  reveals that the requirement $\LTLglobally \LTLfinally(x \geq 10000)$ is vacuous, that is, implied by the other conjuncts.
To see this,  observe that whenever $\neg(c = 0)$ holds, then $x$ is assigned the desired value $10000$,  and otherwise it is incremented by $1$.  This entails that from some point on, $x$ must always be greater than or equal to $10000$.  
Establishing this,   would allow us to simplify the synthesis game  to one with a safety winning condition resulting from the remaining requirements.

To show that the requirement $\LTLglobally \LTLfinally(x \geq 10000)$  is vacuous in $\varphi_{\sf simplify}$,  we can establish that from \emph{every state,  for all possible executions} in any reactive program satisfying the formula
\[
\LTLglobally(c' = 0 \lor c' =1)  \land \LTLglobally(\LTLnext(c = 1)  \rightarrow e =1) \land
\LTLglobally(c=0 \rightarrow x' = x+ 1) \land 
\LTLglobally\big(\neg(c=0) \rightarrow x' = 10000\big)
\]
a state satisfying $x \geq 10000$ is eventually reached. 
This implies that all executions of any such implementation satisfy $\LTLglobally \LTLfinally(x \geq 10000)$ and hence it is sufficient to consider the synthesis game induced by the rest of the requirements,  which results in a simpler winning condition.

\end{example}

\subsection{Challenges and Our Approach}
As the examples above demonstrated,  taking advantage of the concrete \templogic\ formula prior to applying the realizability check and synthesis procedure,  necessitates the application of first-order reasoning in a temporal context.  More concretely,  this involves deriving invariants and reachability properties that are implied by parts of the specification under certain executions. 

To this end,  we propose augmenting the constructed game with a \emph{monitor for the \templogic\ formula} that unrolls and tracks the formula along the system executions. 
The crux of the monitor construction is the derivation of implied invariants and reachability properties,  which enable global specification analysis on the formula level.
For instance, this allows us to derive the invariant in \Cref{ex:unsat-motivating} and the reachability properties in \Cref{ex:vacuous-motivating} and \Cref{ex:GF-motivating}.
Based on the formula unrolling and the computed additional properties, the monitor assigns verdicts such as \UNSAT, which enables pruning of the respective sub-game, or  \SAFETY,  which allows for simplifying the game's  winning condition.
It assigns a verdict \UNSAT\ in \Cref{ex:unsat-motivating},  and the verdict \SAFETY\ allows us to transform the synthesis games in \Cref{ex:vacuous-motivating} and \Cref{ex:GF-motivating} to safety games.

\section{Preliminaries}\label{sec:preliminaries}
In the following, we introduce the notation and concepts we use, including the logic LTL, which has the same temporal operators as our logic \templogic{}. 
Furthermore, we briefly introduce two-player turn-based games, as this is the semantic framework on which our synthesis method builds.

\subsection{Propositional Linear Temporal Logic and $\omega$-Automata}

For a set $\Sigma$, $\Sigma^\omega$ denotes the infinite sequences of elements of $\Sigma$.
For $\sigma =  a_0a_1\ldots \in \Sigma^\infty$ and $i, j \in \Nat$ with $i \leq j$, we define $\sigma[i]:=a_i$ and $\sigma[i,j]:=a_i\ldots a_j$.

Let  \AP\ be a set of Boolean propositions.  The logic LTL~\cite{Pnueli77} over \AP\ is defined by the  grammar
$\psi ::= 
p \mid 
\lnot  \psi \mid 
 \psi_1 \land  \psi_2 \mid 
\LTLnext  \psi \mid 
 \psi_1 \LTLuntil  \psi_2, 
$
where $p \in \AP$. 
LTL formulas are interpreted over infinite sequences in $(\power{\AP})^\omega$ and we denote the satisfaction of $\psi$ by $\sigma \in (\power{\AP})^\omega$ as $\sigma\models\psi$.  We refer the reader to \cite{BaierK08} for the definition.
The language represented by $\psi$ consists of the infinite words that satisfy $\psi$, formally $\lang{\psi} : =\{\sigma \in (\power{\AP})^\omega \mid \sigma \models \psi\}$.

A \emph{deterministic parity automaton} (DPA) is a tuple $\mathcal{A} = (Q,\Sigma, q_0,\delta,\lambda)$ where $Q$ is a finite set of states,  $\Sigma$ is a finite alphabet,  $q_0 \in Q$ is the initial state,  $\delta : Q \times \Sigma \to Q$ is a total transition function and $\lambda : Q \to \Nat$ is a coloring function that maps states to a finite set of colors.
An infinite word $a_0a_1a_2\ldots\in \Sigma^\omega$ is \emph{accepted} by $\mathcal A$ if the highest number occurring infinitely often in the sequence $\lambda(q_0)\lambda(q_1)\lambda(q_2) \ldots$ where $q_{i+1} = \delta(q_i,a_i)$ for every $i \in \Nat$, is even. We denote with $\lang{\mathcal A}\subseteq \Sigma^\omega$ the set of infinite words accepted by $\mathcal A$. 
For every LTL formula $\psi$ over $\AP$, there exists a DPA with alphabet $\power{\AP}$, and with $\power{\power{\mathcal{O}(|\psi|)}}$ states and $\power{\mathcal{O}(|\psi|)}$ colors such that $\lang{\mathcal A} = \lang{\psi}$\cite{Piterman07,EsparzaKRS22}.

\subsection{First-Order Logic}

Let $\values$ be the set of all values of arbitrary types and $\vars$ be the set of all variables.
For variables $X \subseteq \vars$, a function $\assignment: X \to \values$ is called an \emph{assignment to} $X$.
We denote the set of assignments to $X$ as $\assignments{X}$.
$\assignment_1\uplus \assignment_2$ denotes the combination of two assignments $\assignment_1, \assignment_2$ to disjoint variables.

We denote the set of all first-order formulas as $\FOLX$ and by $\QFX$ the set of all quantifier-free formulas in $\FOLX$.
Let $\alpha \in \FOLX$ be a formula and $X = \{x_1,\ldots,x_n\} \subseteq \vars$ be a set of variables.
We write $\alpha(X)$ to denote that the free variables of $\alpha$ are a subset of $X$.  
We also denote with $\FOL{X}$ and $\QF{X}$ the set of formulas (respectively quantifier-free formulas) whose free variables belong to $X$.
For a quantifier $Q \in \{\exists, \forall\}$, we write $Q X.\alpha$ as a shortcut for $Q x_1.\ldots Q x_n.\alpha$.  
For variables $y_1, \dots, y_n$, we write $\alpha[x_1 \mapsto y_1, \dots, x_n \mapsto y_n]$ for $\alpha$ with all $x_i$ replaced simultaneously by $y_i$.
For a formula $\alpha(X)$ and assignment $\assignment \in \assignments{X}$ we denote 
entailment of $\alpha$ by $\assignment$ in the \emph{first-order theory} $T$ by $\assignment \FOLentailsT{T} \alpha$.
For an exposition on first-order logic and first-order theories, we refer the reader to \cite{BradleyM07}.

For a set $Y \subseteq \vars$ of variables,  we denote with $Y'$ the  \emph{primed version of $Y$} such that 
$Y' : = \{y' \mid y \in Y\} \subseteq \vars$ and $Y \cap Y' = \emptyset$.
If $\assignment \in \assignments{Y}$ is an assignment to the variables in $Y$,  we define the assignment $\assignment' \in \assignments{Y'}$ over $Y'$ such that $\assignment'(y') = \assignment(y)$ for all $y \in Y$.
Given two assignments $\assignment_1,\assignment_2 \in \assignments{Y}$ to the variables in $Y$, 
we define $\langle \assignment_1, \assignment_2 \rangle := \assignment_1 \uplus \assignment_2'$.

\subsection{Two-Player Turn-Based Graph Games}

A \emph{two-player game graph} is a tuple $G = (V,V_\env,V_\sys,\tau)$ where $V = V_\env \uplus V_\sys$ are the vertices,  partitioned between the environment player (\emph{player $\env$}) and the system player (\emph{player $\sys$}), and $\tau \subseteq (V_\env \times V_\sys) \cup (V_\sys \times V_\env)$ is the \emph{transition relation}.
A \emph{play} in $G$ is a sequence $\xi \in V^\omega$ where $(\xi[i],\xi[i+1])\in\tau$ for all $i \in \Nat$.
A \emph{strategy for player~$p$} is a function 
$\sigma: V^*V_{p} \to V$ where $\sigma(\xi\cdot v) = v'$ implies $(v,v') \in \tau$.
A play $\xi$ is \emph{consistent with $\sigma$} if $\xi[i+1] = \sigma(\xi[0,i])$ for every $i \in \Nat$ where $\xi[i] \in V_p$.
$\plays_G(v,\sigma)$ is the set of all plays in $G$ starting in $v$ and  consistent with strategy $\sigma$.

A \emph{winning condition} in $G$ is a set $\Omega \subseteq V^\omega$. A \emph{two-player turn-based game} is a pair $(G,\Omega)$, where $G$ is a game graph and $\Omega$ is a wining condition for player~$\sys$.
A sequence $\xi \in V^\omega$ is \emph{winning for player~$\sys$} if and only if $\xi \in \Omega$,  and is \emph{winning for player~$\env$} otherwise.  
The \emph{winning region $W_p(G,\Omega)$ of player $p$ in $(G,\Omega)$} is the set of all vertices $v$ from which player $p$ has a strategy $\strat$ such that every play in $\plays_G(v,\sigma)$ is winning for player $p$. 
A strategy $\sigma$ of player $p$ is \emph{winning} if for every $v\in W_p(G,\Omega)$, every play in $\plays_G(v,\sigma)$ is winning for player~$p$.

\section{Linear Temporal Logic for Reactive Programs}\label{sec:logic_and_realizability}
In this section, we first introduce the temporal logic \templogic\ and define the realizability and  synthesis problems for this logic. 
We then present the general synthesis workflow, which transforms an \templogic\ formula into a symbolic synthesis game and then solves this game.
\subsection{\templogic\ as a Specification Language for Reactive Programs}
In this section we define Reactive Program Linear Temporal Logic (\templogic),  a temporal logic for the specification of reactive programs.
\templogic\ extends LTL by replacing the Boolean propositions used in LTL by quantifier-free first-order formulas.
Since \templogic\ is specifically targeted at specifying the behavior of reactive programs,  these first-order formulas are over variables representing the program's input,  current and next state.  More concretely,   we consider the set of formulas $\QF{\specvars}$, where $\progvars$ is the set of \emph{program variables},   $\inputs$ is the set of \emph{input variables} and $\progvars'$ is the set of variables representing the values of $\progvars$ at the \emph{next step} of the execution.
The use of $\progvars'$ allows us to specify relations between the current and next states of a program, such as, for example, modeling assignments to program variables. The next definition formalizes the syntax of \templogic.

\begin{definition}[\templogic\ Syntax]\label{def:rpltl-syntax}
Let $\progvars,\inputs$ and $\progvars'$ be mutually disjoint finite sets of variables such that $\progvars' = \{x' \mid x \in \progvars\}$.
\emph{Reactive Program Linear Temporal Logic (\templogic)} is defined by
\[ \varphi ::=
\alpha \mid
\lnot \varphi \mid
\varphi_1 \land \varphi_2 \mid
\LTLnext \varphi \mid
\varphi_1 \LTLuntil \varphi_2,
\]
where $\alpha \in \QF{\specvars}$.
We denote with $\templogic(\specvars)$ the set of \templogic\ formulas over the set of variables $ \specvars$.
For $\varphi \in \templogic(\specvars)$, we denote with
$\atoms(\varphi) \subseteq \QF{\specvars}$ the set of \emph{atomic formulas} appearing in $\varphi$.
\end{definition}

\templogic\ formulas are interpreted over infinite sequences of valuations of the variables $\progvars \cup \inputs$, representing the values of the program variables and the input variables at a given point in time. The next-step variables $\progvars'$ are interpreted by the next element of the sequence (which always exists since we consider infinite sequences). Hence, \templogic\ allows us to describe properties where the values of the program variables $\progvars$ are related over time.

\begin{example}
	Consider again the first example of an \templogic\ formula given in \Cref{sec:intro}.  There,  the  variable $x'$ denotes the value of the program variable $x$ at the next step.  
	The formula requires that for every point in time in an infinite sequence of valuations to the variables $\{e,x,y\}$,   if $e > 0$ holds true at the current step,  then the value of $x$ at the next step must be less than or equal to the value of the expression $x + 2$ at the current time step.
\end{example}

The next definition provides the formal semantics.

For the rest of the section, we fix a first-order logic theory $T$.

\begin{definition}[\templogic\ Semantics]\label{def:rpltl-semantics}
Let $\varphi \in \templogic(\specvars)$ be an \templogic\ formula and let
$\rho \in {(\assignments{\progvars \cup \inputs})}^\omega$ be an infinite sequence.
We say that $\rho$ \emph{satisfies} $\varphi$,
denoted as $\rho \models \varphi$, iff
    \[\begin{array}{lll}
        & \rho \models \alpha &:\Leftrightarrow \langle \rho[0], \rho[1] \rangle \FOLentailsT{T} \alpha,  \text{ for } \alpha \in  \QF{\specvars} \\
        & \rho \models \lnot \varphi &:\Leftrightarrow \rho \not\models \varphi \\
        & \rho \models \varphi_1 \land \varphi_2 &:\Leftrightarrow \rho \models \varphi_1 \text{ and } \rho \models \varphi_2 \\
        & \rho \models \LTLnext \varphi &:\Leftrightarrow \rho_{+1} \models \varphi \\
        & \rho \models \varphi_1 \LTLuntil \varphi_2 &:\Leftrightarrow \exists j \in \Nat \text{ s.t. } \rho_{+j} \models \varphi_2 \text{ and } \forall i < j.~ \rho_{+i} \models \varphi_1\\
    \end{array}
    \]
    where for every $k \in \Nat$ we define
    $\rho_{+k} [n] := \rho[n + k,\infty)$.
\end{definition}

The constants $\false$~(false),  $\true$~(true) and the remaining Boolean operators are derived in the standard way.
The temporal operators
``eventually'' $\LTLeventually$,
``globally'' $\LTLglobally$ and
``weak until'' $\LTLweakuntil$,  can be derived as
$\LTLeventually \varphi := \true \LTLuntil \varphi$,
$\LTLglobally \varphi := \neg(\LTLeventually \neg \varphi)$,  and
$\varphi_1 \LTLweakuntil \varphi_2 := (\varphi_1 \LTLuntil \varphi_2) \lor \LTLglobally \varphi_1$.
If an \templogic\ formula does not contain $\LTLuntil$ and $\LTLfinally$ when in negation normal form, then it is \emph{(syntactic) safety formula}~\cite{Sistla94}.

We define the \emph{language of infinite sequences represented by a formula $\varphi \in \templogic(\specvars)$} as
$\lang{\varphi} := \{\rho \in {(\assignments{\progvars \cup \inputs})}^\omega \mid \rho \models \varphi\}$.
If $\lang{\varphi} = \emptyset$ we say that $\varphi$ is \emph{unsatisfiable}, if  $\lang{\varphi} = (\assignments{\progvars \cup \inputs})^\omega$ we say that $\varphi$ is \emph{valid}, and if $\lang{\varphi} \not= \emptyset$, we say that $\varphi$ is \emph{satisfiable}.

\paragraph{\templogic\ Expressivity.}

As \templogic\ formulas allow us to relate values over time by expressing relationships between $\progvars$ and $\progvars'$, it is possible to use the variables from $\progvars$ as registers. Hence, in a sufficiently expressive theory domain,  e.g. linear integer arithmetic,  it is straightforward to encode computation formalisms like e.g. WHILE-programs or counter machines. Furthermore, since \templogic\ formulas are interpreted over infinite sequences, it is easily possible to express properties like e.g. a program terminates (or reaches a given state) on all inputs.

If the atomic formulas are allowed to use  $<$ and $=$,  and the variables in $\progvars$ range over the integers,  there exist \templogic\ formulas $\varphi$ such that $\lang{\varphi}$ is not $\omega$-regular,  which follows from the result for constraint LTL established in~\cite{DemriD07}. 
On the other hand,  \templogic\ generalizes LTL on the assertion level, and LTL does not capture all $\omega$-regular languages.  
Thus,  to specify in \templogic\ an arbitrary $\omega$-regular property represented by an automaton on infinite words,  we need,  in general,  to use additional variables to encode the states of the automaton and its transition relation. 
The standard way to increase the expressive power of LTL (or LTL modulo theories) is via combination with regular expressions~\cite{EisnerF06} (or regular expressions modulo theories~\cite{VeanesBES23}), which capture better the higher-level declarative aspects of specifications.
While such an extension of \templogic\  is worth investigating in the future,  in this paper we focus on \templogic, as it lifts LTL,  which underlies the defacto standard specification language for reactive synthesis~\cite{TLSF}.

\paragraph{\templogic\ Realizability.}
The core notion in reactive synthesis is \emph{realizability}.
Intuitively, a temporal formula is realizable, if there exists a system that can react to any input from the environment over infinitely many rounds,  such that the resulting infinite execution satisfies the formula.
We can formalize this as the statement that there \emph{exists a function} mapping the initial assignment to $\progvars$ and the sequence of inputs so far to the next assignment to the program variables $\progvars$, such that all possible sequences induced by this function satisfy the formula.
Formally,  we define the notion of \emph{realizability} of an \templogic\ formula $\varphi$ as follows.

\begin{definition}[\templogic\ Realizability]\label{def:realizability}
A formula $\varphi \in \templogic(\specvars)$ is called \emph{realizable} if and only if there exists a function
$\sigma: \assignments{\progvars} \times \assignments{\inputs}^+ \to \assignments{\progvars}$ such that for every infinite sequence $\mathit{input} \in \assignments{\inputs}^\omega$ of assignments to $\inputs$ and initial assignment $\mathit{init} \in \assignments{\progvars}$, for the infinite word $\rho \in {(\assignments{\progvars \cup \inputs})}^\omega$ defined as
$\rho[0] := \mathit{init} \uplus \mathit{input}[0] $ and $\rho[n] := \sigma(\mathit{init}, \mathit{input}[0,n-1]) \uplus \mathit{input}[n]$ for $n > 0$,  $\rho \models \varphi$ holds.
\end{definition}

\begin{tcolorbox}[
  colback=gray!0!white,
  colframe=gray!50!black,
  title={Realizability and Synthesis Problem from \templogic{}}]
The \emph{realizability problem} is checking whether a formula $\Phi \in \templogic(\specvars)$ is realizable.
The \emph{synthesis problem} is finding a witness for the realizing function (i.e. the system).
\end{tcolorbox}

One way of representing this witness is as \emph{reactive program}~\cite{HeimD24}.
Intuitively, this is a program that reads the inputs $\inputs$ and updates its program variables (a superset of $\progvars$),  repeating this indefinitely.

\label{sec:logic} 
\subsection{\templogic\  Realizability Checking via Symbolic Game Solving}
In the standard automata-theoretic approach to the synthesis of finite-state reactive systems from LTL specifications~\cite{LuttenbergerMS20},  the specification is translated to a deterministic parity automaton that is interpreted as a two-player turn-based game between the synthesized system and its environment.
A winning strategy for the system player from a dedicated initial vertex corresponds to a finite-state implementation that realizes the specification. If, on the other hand,  the environment wins from the initial vertex,  the specification is unrealizable.

To perform realizability checking and synthesis from \templogic\ specifications,  we can adopt a similar flow.
The key difference is that instead of Boolean atomic propositions,
\templogic\ formulas are over variables with possibly infinite domains.
To address this,  we translate \templogic\  specifications to deterministic parity automata with \emph{symbolically represented transition functions},  which are then interpreted as \emph{symbolic game structures}.
In that way,  we reduce the realizability checking and synthesis problems for \templogic\  to the problem of solving a symbolically represented infinite-state game.
Next in this section, we describe this reduction in detail, after which discuss its fundamental limitation, which our proposed approach addresses as presented in the rest of the paper.

\subsubsection{Constructing Automata from \templogic\  Formulas}\label{sec:formula-to-automaton}

We now explain the first step, which is the translation of \templogic\ formulas to automata. 
We remark that this translation is separate from and independent of the monitor construction that we present in subsequent sections.  
Here,  we utilize the respective construction for LTL specifications. To this end, we interpret the given \templogic\ formula $\varphi$ as an LTL formula $\tobool{\varphi}$ by replacing the atomic sub-formulas elements of $\QF{\specvars}$ by Boolean propositions.
More concretely,  we construct the LTL formula  $\tobool{\varphi}$ from $\varphi$ by replacing each atomic formula in $\atoms(\varphi)$ with a unique propositional variable.
Then,  we obtain a DPA $\mathcal{A}_{\tobool{\varphi}}$ with alphabet $\power{\AP(\tobool{\varphi})}$,  where $\AP(\tobool{\varphi})$ are the atomic propositions in $\tobool{\varphi}$.
The language of $\mathcal{A}_{\tobool{\varphi}}$ is $\lang{\tobool{\varphi}}$.
In practice,  this can be done using standard methods such as those implemented in~\cite{Duret-LutzRCRAS22, Duret-LutzLFMRX16}.

Intuitively,  by replacing back the propositions from $\AP(\tobool{\varphi})$ by the original atomic formulas $\atoms(\varphi)$ in the transition function of the DPA $\mathcal{A}_{\tobool{\varphi}}$, we can obtain a symbolic DPA representing the language  $\lang{\varphi}$.
In the next subsection we use this idea to directly interpret $\mathcal{A}_{\tobool{\varphi}}$ as a \emph{symbolically represented} two-player graph game encoding the realizability problem for  $\varphi$.

\subsubsection{Symbolic Infinite-State Games for \templogic\ Realizability and Synthesis}\label{sec:automaton-to-game}

As we explained above,  due to the possibly infinite domains of the variables in \templogic\ specifications,  the games encoding the synthesis problem have in general an infinite set of vertices.
To this end,  we represent two-player game graphs via \emph{symbolic game structures} defined next.

\begin{definition}[Symbolic Game Structure]\label{def:sym-games}
    A \emph{symbolic game structure} is a tuple $(L, \linit,\inputs, \progvars, \dom, \delta)$ where
    $L$ is a finite set of \emph{locations,}
    $\linit \in L$ is the \emph{initial location},
    $\inputs \subseteq \vars$ is a finite set of \emph{input variables},
    $\progvars \subseteq \vars$ is a finite set of \emph{program variables},
    $\mathit{dom}: L \mapsto \QF{\progvars}$ is the \emph{domain of the states},
    and $\delta: L \times L \mapsto \QF{\progvars \cup \inputs \cup \progvars'}$ is the \emph{transition relation}.
    $\delta$ is required to be \emph{deterministic} and \emph{non-blocking}, i.e., for all $l \in L$ and  $\assmt{x} \in \assignments{\progvars}$ where $\assmt{x} \FOLentailsT{T} \mathit{dom}(l)$ we have that
\begin{enumerate}
\item for every input assignment $\assmt{i} \in \assignments{\inputs}$ and program variable  assignment $\assmt{v} \in \assignments{\progvars}$,  there exists \emph{at most one} location $l' \in L$ such that $\assmt{x}\uplus \assmt{i} \uplus \assmt{v}' \FOLentailsT{T} \delta(l,l')$.
\item there exist $\assmt{i} \in \assignments{\inputs}$,  $\assmt{v}  \in \assignments{\progvars}$ and $l' \in L$   such that $\assmt{v} \FOLentailsT{T}  \dom(l')$ and $\assmt{x}\uplus \assmt{i} \uplus \assmt{v}' \FOLentailsT{T} \delta(l,l')$.
\end{enumerate}
\end{definition}

Intuitively,  a symbolic game structure $\symgame = (L, \linit,\inputs, \progvars, \dom, \delta)$ represents a two-player game graph encoding the interaction between an environment selecting values for the input variables $\inputs$ and a system selecting the next values of the program variables $\progvars$.
The function $\dom$ constrains the values of the program variables to a specific domain (represented symbolically via formulas in $\QF{\progvars}$), and the transition relation $\delta$ constrains the choices of the two players.
From a state $(l,\assmt{x}) \in L \times \assignments{\progvars}$  with $ \assmt{x} \FOLentailsT{T} \mathit{dom}(l)$,  the environment selects an input $\assmt{i} \in\assignments{\inputs}$ and the system selects a next assignment $\assmt{v} \in \assignments{\progvars}$ to the program variables such that for some $l' \in L$ it holds that $\assmt{v}  \FOLentailsT{T}  \dom(l')$ and $\assmt{x}\uplus \assmt{i} \uplus \assmt{v}' \FOLentailsT{T} \delta(l,l')$.  Then, the game proceeds to the new state $(l',\assmt{v})$ and the process repeats from there, resulting in  an infinite sequence of states.

The semantics of a symbolic game structure $\symgame$, which is a possibly infinite two-player turn-based game graph,   is formalized in the next definition.

\begin{definition}[Semantics of Symbolic Game Structures]
\label{def:sym-games-semantics}
The semantics of a symbolic game structure
$\symgame = (L, \linit,\inputs, \progvars, \dom, \delta)$ is the game graph
$\sema{\symgame} = (\states, \states_\env,\states_\sys,\tau)$ where
$\states := \states_\env\uplus\states_\sys$,
    \begin{itemize}
        \item  $\states_\env := \{ (l,\assmt{x}) \in L \times \assignments{\progvars} \mid \assmt{x} \FOLentailsT{T} \mathit{dom}(l)\}$;
        \item   $\states_\sys :=\{ ((l,\assmt{x}),\assmt{i}) \in \states_\env\times \assignments{\inputs}\mid\exists l' \in L \exists \assmt{v}\in \assignments{\progvars}.
		\assmt{v}\FOLentailsT{T}  \dom(l') \text{ and } \newline
		\phantom{((l,\assmt{x}),\assmt{i}) \in \states_\env\times \assignments{\inputs}\mid\exists l' \in L \exists \assignment \in \assignments{\progvars}\qquad}
		 \assmt{x} \uplus \assmt{i} \uplus \assmt{v}' \FOLentailsT{T} \delta(l,l')\}$;
        \item   $\tau \subseteq (\states_\env \times \states_\sys) \cup (\states_\sys \times \states_\env)$ is the smallest relation such that
	\begin{itemize}
	\item $(s,(s,\assmt{i})) \in \tau$ for every $s \in \states_\env$ and $\assmt{i} \in \assignments{\inputs}$ such that $(s,\assmt{i}) \in \states_\sys$,
	\item $(((l,\assmt{x}), \assmt{i}),(l',\assmt{v})) \in \tau$ if and only if $\assmt{x} \uplus\assmt{i} \uplus \assmt{v}' \FOLentailsT{T} \delta(l,l')$.
        \end{itemize}
    \end{itemize}
\end{definition}

For symbolic game structures,  we consider winning conditions defined in terms of the infinite sequence of locations visited in a play,  which we refer to as \emph{location-based winning conditions}.
Such winning conditions are independent of the valuations of the inputs and program variables, and,  as we will see later in this section, are sufficient for encoding the \templogic\ synthesis problem.

\begin{definition}[Location-Based Winning Condition]
\label{def:winning-condition}
A \emph{location-based wining condition} in a symbolic game structure $\symgame = (L, \linit,\inputs, \progvars, \dom, \delta)$ is a set of
infinite sequences of locations $\Lambda \subseteq L^\omega$.
A \emph{location-based parity condition} is defined by a coloring function $\lambda: L \to  \Nat$ such that
$\parity(\symgame,\lambda) = \{l_0l_1  \ldots \in L^\omega \mid \text{ the highest number occurring infinitely often in } \lambda(l_0)\lambda(l_1) \ldots \text{ is even}\}$.
\end{definition}

For a play $\xi$ of $\sema{\symgame}$ of the form
$ (l_0,\assmt{x}_0)((l_0,\assmt{x}_0),\assmt{i}_0)(l_1,\assmt{x}_1)((l_1,\assmt{x}_1),\assmt{i}_1),\ldots \in (\states_\env\cdot\states_\sys)^\omega$ or of the form
$((l_0,\assmt{x}_0),\assmt{i}_0)(l_1,\assmt{x}_1)((l_1,\assmt{x}_1),\assmt{i}_1),\ldots \in (\states_\sys\cdot\states_\env)^\omega$ we define $\Loc(\xi):=l_0l_1l_2\ldots \in L^\omega$ to be the corresponding sequence of game locations visited by $\xi$.
A location-based winning condition $\Lambda$ for $\symgame$ defines a winning condition $\sema{\Lambda} \subseteq \states^\omega$ where
$\sema{\Lambda} := \{\xi \in \states^\omega \mid \Loc(\xi) \in \Lambda\}$.
A \emph{symbolic game} is a pair $(\symgame,\Lambda)$  of
a symbolic game structure $\symgame$ and
a location-based winning condition $\Lambda$.

We are now ready to describe the construction of a symbolic game structure equipped with a location-based winning condition for a given \templogic\ formula  $\varphi \in \templogic(\specvars)$.

Let $\mathcal{A}_{\tobool{\varphi}} = (Q,\Sigma, q_0,\delta,\lambda)$ be a DPA with
$\lang{\mathcal{A}_{\tobool{\varphi}}} = \lang{\tobool{\varphi}}$  constructed from
$\varphi$ as explained in \Cref{sec:formula-to-automaton}.
From $\mathcal{A}_{\tobool{\varphi}}$ we  construct the symbolic game structure
$\symgame_{\varphi} = (Q, q_0,\inputs, \progvars, \dom, \delta_{\symgame})$ where
\begin{itemize}
\item the set of locations is the set of automaton states $Q$, and $\dom(q) = \true$ for all $q \in Q$, and
\item the transition relation $\delta_{\symgame}$ is defined based on $\delta$ such that for all $q,q' \in Q$, \\
$\delta_{\symgame}(q,q') := \bigvee_{\{a \in \Sigma \mid \delta(q,a) = q'\}} \Big(
\bigwedge_{\{\alpha \in \atoms(\varphi) \mid \tobool{\alpha} \in a\}}\alpha
\land
\bigwedge_{\{\alpha \in \atoms(\varphi) \mid \tobool{\alpha} \not\in a\}}
\neg\alpha
\Big).$

Intuitively,  the formula $\delta_{\symgame}(q,q') \in \QF{\specvars}$ is defined as the disjunction over all letters $a \in \power{\AP(\tobool{\varphi})}$ that label a transition for $q$ to $q'$.  Each such $a$ corresponds to the conjunction of the propositions that belong  to $a$ and the negation of those that do not.  Replacing each such proposition with the original formula in $\QF{\specvars}$ gives us the formula $\delta_{\symgame}(q,q') $ representing  the assignments to $\specvars$ with which the game transitions form $q$ to $q'$.
\end{itemize}

Since the automaton $\mathcal{A}_{\tobool{\varphi}}$ with alphabet $\Sigma=\power{\AP(\tobool{\varphi})}$ is deterministic, and its transition  function $\delta$ is total,  the definition of $\delta_{\symgame}$ guarantees that $\symgame_{\varphi}$ satisfies the conditions of \Cref{def:sym-games}.

The coloring function $\lambda$ of $\mathcal{A}_{\tobool{\varphi}}$ yields a location-based parity winning condition $\parity(\symgame_{\varphi},\lambda)$.

We say that $(\symgame,\Lambda)$ is a \emph{symbolic game for the \templogic\ specification $\varphi$} if and only if for some DPA
$\mathcal{A}_{\tobool{\varphi}} = (Q,\Sigma, q_0,\delta,\lambda)$  with $\lang{\mathcal{A}_{\tobool{\varphi}}} = \lang{\tobool{\varphi}}$,
$\symgame$ is the symbolic game structure obtained from $\mathcal{A}_{\tobool{\varphi}}$ as defined above, and
$\Lambda = \parity(\symgame,\lambda)$
The next theorem establishes that the symbolic game structure $\symgame_{\varphi}$ equipped with the winning condition $\parity(\symgame_{\varphi},\lambda)$ encodes the
realizability problem for $\varphi$.

\begin{restatable}[Symbolic Game Correctness]{theorem}{restateGameCorrectness}\label{thm:game-correctness}
Let $\varphi \in \templogic(\specvars)$ be a formula
and $(\symgame,\Lambda)$ be a symbolic game for $\varphi$ with
$\symgame = (L, \linit,\inputs, \progvars, \dom, \delta)$.
The formula $\varphi$ is realizable if and only if $(\linit,\assmt{x}) \in \win_\sys(\sema{\symgame},\sema{\Lambda})$ for every $\assmt{x} \in \assignments{\progvars}$ with $\assmt{x} \FOLentailsT{T} \dom(\linit)$.
\end{restatable}

\Cref{thm:game-correctness} gives us a method for reducing the realizability and synthesis problems for \templogic\ to the solving of an infinite-state game and the computation of a winning strategy for the system.
In contrast to the automata-theoretic approach to synthesis from LTL, where the construction of the deterministic automaton is often the bottleneck, here
the result of the translation is an infinite-state two-player game,  the problem of solving which is in general undecidable.

\label{sec:logic_games} 

\section{Augmenting the Synthesis Game with a Monitor for \templogic}\label{sec:monitor_definition}
The procedure outlined in \Cref{sec:logic_games} for translating an \templogic\ formula $\varphi$ into a symbolic game, does not account for the first-order concretization of the propositions in the propositional formula $\tobool{\varphi}$. In particular, any state-space reductions applied by the underlying automata construction are not able to take advantage of the semantics of the propositions.
Since techniques for solving infinite-state games face significant challenges,  applying possible simplifications to the game during its construction is essential. 
Following the automata construction, the \templogic\ formula is forgotten,  with the deterministic automaton being its low-level representation. 
Thus,  in order to take advantage of the structure of the \templogic\ formula it is desirable to \emph{apply semantic simplifications at the formula level and during the construction of the synthesis game}.

In this paper, we present a principled,  generally applicable approach to perform such semantic simplifications. The key idea is to construct a so-called \emph{monitor for the given \templogic\  formula} $\varphi$ that can be used to augment the constructed game with semantic information.  This augmentation,  which will be in the form of constructing the product of the monitor and the game,  will perform game-structure reductions based on this semantic information.
More concretely,  the role of the monitor is to identify states of the automaton  $\mathcal A_{\tobool{\varphi}}$ whose language is empty, or contains all possible executions,  thus identifying unsatisfiable or vacuous requirements as illustrated in \Cref{sec:examples}.

\paragraph*{Remark} The symbolic game is constructed from a DPA for the propositional formula $\tobool{\varphi}$. Currently, there exist no techniques and tools for \emph{directly} constructing \emph{deterministic symbolic automata} from temporal logic specifications over non-Boolean domains. 
We note  that the approach proposed here and the  monitor construction presented in the next subsection are independent from the procedure used to construct the DPA and the resulting symbolic game. 

\subsection{Monitors for \templogic\ Formulas}
We now proceed with the formal definition of our notion of \emph{monitor for  \templogic\ formulas}.

\begin{definition}[Monitor]\label{def:monitor}
A \emph{monitor} over $\progvars$,  $\inputs$, $\progvars'$ is a tuple $M = (Q,\qinit,\preds,\delta,\verdict)$ where 
\begin{itemize}
\item $Q$ is a finite set of \emph{monitor states} and $\qinit$ is the \emph{monitor's initial state},
\item $\preds \subseteq \QF{\specvars}$ is a finite set of \emph{monitored predicates}, 
\item $\delta : Q \times 2^{\preds} \to Q$ is the monitor's \emph{transition function},
\item $\verdict: Q \to \Verdicts$ is a \emph{verdict-labelling function} such that  for every $q,  q' \in Q$ where $\delta(q,a) = q'$ for some $a$,  the following two conditions are satisfied: 
\begin{itemize}
\item if $\verdict(q)  =\UNSAT$, then $\verdict(q') = \UNSAT$,  and
\item if $\verdict(q) = \SAFETY$, then $\verdict(q') \in \{\UNSAT,\SAFETY\}$.
\end{itemize}
\end{itemize}
\end{definition}

Thus, a monitor is a finite-state  automaton whose transitions are labelled with sets of predicates, that is, formulas in  $\QF{\specvars}$.  Each such set defines a subset of $\assignments{\specvars}$ which determines the variable assignments with which this transition can be taken. 
We define a function $\delta^*_M : (\assignments{\progvars\cup\inputs})^* \to Q$ that maps each finite sequence $\pi \in (\assignments{\progvars\cup\inputs})^*$ to the state that $M$ reaches after reading $\pi$.
Formally,  we let 
$\delta^*_M(\pi) = \qinit$ if $|\pi| = 0$ or $|\pi| = 1$,  and otherwise we let 
$\delta^*_M(\pi) = q$ if $\pi = \pi' \cdot \assignment_1 \cdot \assignment_2$,
$\delta^*_M(\pi'\cdot\assignment_1) = q'$ and 
$q := \delta(q',\{p \in \preds\mid \langle \assignment_1,\assignment_2\rangle \models p\})$. 

The verdict-labelling function $\verdict$ of a monitor $M$ maps each of the monitor's states to an element of the set $\Verdicts$.  
Intuitively,  in order for $M$ to be a monitor for a \templogic\ formula $\varphi$,  we require that for each finite prefix $\pi \in (\assignments{\progvars\cup\inputs})^*$,  the verdict assigned to the state $\delta^*_M(\pi)$ that $M$ reaches when reading $\pi$ should be consistent with the language $\lang{\varphi}$. This requirement is formalized in the next definition.

\begin{definition}[Monitor for an \templogic\ Formula]\label{def:monitor-formula}
Let $\varphi \in \templogic(\specvars)$ be a formula.
A monitor $M = (Q,\qinit,\preds,\delta,\verdict)$ over variables $\progvars$,  $\inputs$, $\progvars'$ is a \emph{monitor for $\varphi$} iff for every $\pi \in (\assignments{\progvars\cup\inputs})^*$ and $\assignment \in \assignments{\progvars\cup\inputs}$, the following properties hold for $q=\delta^*_M(\pi\cdot \assignment)$:
\begin{enumerate}
\item if $\verdict(q) = \UNSAT$,  then $\pi \cdot \assignment\cdot \rho \not \in \lang{\varphi}$ for all $\rho \in (\assignments{\progvars\cup\inputs})^\omega$;
\item if $\verdict(q) = \SAFETY$,  then, for every $\rho \in (\assignments{\progvars\cup\inputs})^\omega$, \\ if $\verdict(\delta^*_M(\pi \cdot \assignment\cdot \rho[0,i])) \neq \UNSAT$ for all $i \in \Nat$, then it holds that $\pi \cdot \assignment \cdot\rho \in \lang{\varphi}$.
\end{enumerate}
\end{definition}

Thus,  in particular, a monitor $M$ for an \templogic\ formula $\varphi$ ensures that the verdict assigned to a finite sequence $\pi \in (\assignments{\progvars\cup\inputs})^*$,  that is $\verdict(\delta^*_M(\pi))$,  is \UNSAT\ only if the language $\{ \rho \in (\assignments{\progvars\cup\inputs})^\omega \mid \pi\cdot \rho \in \lang{\varphi} \}$ is empty.

Intuitively,  \UNSAT\ verdicts allow us to identify unsatisfiable requirements,  while \SAFETY\ verdicts are useful for the simplification of the winning condition of the game resulting from the augmentation of the symbolic game structure $\symgame_{\varphi}$ with a monitor for $\varphi$.

\paragraph*{Remark} The verdicts assigned by a monitor for an \templogic\ formula $\varphi$ are determined based on the language $\lang{\varphi}$, and not based on the existence of strategies in the synthesis game.  Hence, we speak of (un)satisfiability instead of (un)realizability when referring to the monitor's role.  
The reason for that is that the game-based interaction between the system and the environment is handled when solving the resulting game, while the utilization of the monitor is during the game's construction.
The role of the monitor is to identify inconsistencies or implied sub-specifications by temporal and first-order reasoning on the formula level.
In \Cref{sec:monitor} we present our main technical contribution, which is a technique for the construction of monitors from \templogic\ formulas.
Before that, we describe how a monitor for an \templogic\ formula $\varphi$ can be used to \emph{prune a symbolic game for $\varphi$}.

\subsection{Product of a Symbolic Game and a Monitor}
We use a monitor for an \templogic\ formula $\varphi$ to prune a symbolic game for $\varphi$ by constructing  a product between the game and the monitor such that they both ``run in parallel''.
This gives us a general method for adding semantic information to the resulting synthesis game,  which is independent of how the symbolic game for $\varphi$ was constructed,  and of the concrete construction of the monitor,  as long as it is a monitor for $\varphi$ according to \Cref{def:monitor-formula}.
In this construction,  we use the verdict-labelling function of the monitor to modify the winning condition of the game. 
For example, if the monitor assigns verdict \UNSAT, the respective product state is losing for the system player. 
Furthermore, if the assigned verdict is \SAFETY,  the winning condition for the respective sub-game can be switched to a safety winning condition.
This pruning  has  two advantages:
First, we disregard obvious bad choices for one of the players if they lead to a verdict that can already be computed from the monitor.
Second, for \SAFETY\ verdicts the winning condition simplifies from e.g. a parity condition to a safety condition, which can potentially make solving the game easier.
Formally, we define this product as follows:

\begin{definition}[Game-Monitor Product]\label{def:monitor-product}
Let 
$\symgame = (L, \linit,\inputs, \progvars, \dom, \delta)$ be a symbolic game structure, 
$\Lambda \subseteq L^\omega$ be  a location-based winning condition, 
and let $M = (Q,\qinit,\preds,\delta_M,\verdict)$ be a monitor over
 $\progvars$,  $\inputs$, $\progvars'$. 
We define the \emph{game-monitor product} for $(\symgame,\Lambda)$ and $M$ as the  pair $(\symgame_\times,\Lambda_\times)$ consisting  of the symbolic game structure 
$\symgame_\times := (L \times Q, (\linit,\qinit),\inputs, \progvars, \dom_\times, \delta_\times)$ and the location-based winning condition $\Lambda_\times \subseteq {(L \times Q)}^\omega$ where
\begin{itemize}
    \item   $\dom_\times((l,q)) := \dom(l)$,
    \item   $\delta_\times((l,q), (l',q')) := \delta(l, l') \land \bigvee_{\{a \in \power\preds \mid q' = \delta_M(q, a)\}} \left(  \left( \bigwedge_{\alpha \in \preds \cap a} \alpha \right) \land \left( \bigwedge_{\alpha \in \preds \setminus a} \lnot \alpha \right) \right)$,
\end{itemize}
and, furthermore, $\rho = (l_0,q_0)(l_1,q_1)\ldots\in \Lambda_\times$ if and only if $\verdict(q_i) \neq \UNSAT$ for all $i \in \Nat$,  and
\begin{itemize}
    \item   $l_0 l_1 l_2 \dots \in \Lambda$, or
    \item   there exists $i \in \Nat$ such that $\verdict(q_i)  = \SAFETY$.
\end{itemize}
\end{definition}

Intuitively, the product transition function ensures that the monitor's execution is synchronized with the transitions of the game.
Note that the product game is also deterministic, as the monitor is deterministic and cannot transition to different states for the same valuation. 
Furthermore,  the monitor's transition function is total w.r.t.\ all possible valuations of the variables, meaning that the product is non-blocking.
Hence, the resulting game structure is well-defined.
The product winning condition ensures, together with the definition of monitor, that if the monitor assigns an $\UNSAT$ verdict, then the respective  play is losing for the system.  
Otherwise,  in order to win,  the system player has to either satisfy the original winning condition, or reach a product state where the monitor assigns a \SAFETY\ verdict and subsequently avoid reaching a state with an $\UNSAT$ verdict.

The next theorem establishes that when $(\symgame,\Lambda)$ and $M$ are a symbolic game and a monitor for an \templogic\ formula $\varphi$, then the game-monitor product encodes the realizability of $\varphi$.

\begin{restatable}[Product Correctness]{theorem}{restateProductCorrectness}\label{thm:product-correctness}
Let $\varphi \in \templogic(\specvars)$ be a formula,
 $(\symgame,\Lambda)$ be a symbolic game for $\varphi$,  with 
$\symgame = (L, \linit,\inputs, \progvars, \dom, \delta)$,  and let
 $M = (Q,\qinit,\preds,\delta_M,\verdict)$  be a monitor for $\varphi$. 
Let $(\symgame_\times,\Lambda_\times)$ 
be the game-monitor product for $(\symgame, \Lambda)$ and $M$.

The formula $\varphi$ is realizable if and only if 
$((\linit,\qinit),\assmt{x}) \in 
\win_\sys(\sema{\symgame_\times},\sema{\Lambda_\times})$ 
for every assignment $\assmt{x} \in \assignments{\progvars}$ with $\assmt{x} \FOLentailsT{T} \dom((\linit,\qinit))$.
\end{restatable}

\section{\templogic\ Monitor Construction}\label{sec:monitor}
In this section, we describe how we construct a monitor for a given $\Phi \in \templogic(\specvars)$.
Intuitively, the states of the monitor are formulas generated from $\Phi$ that track the current obligations.
More concretely,  each state corresponds to a formula that needs to hold on the suffix of an infinite sequence whose prefix  reaches this state in order for the sequence to satisfy $\Phi$.
Clearly,  if the formula associated with a state is $\false$,  we can assign verdict $\UNSAT$ to that state.
If this is not the case, we apply logical reasoning to the state in order to simplify it,  and potentially enable a simplification of the winning condition of the game or facilitate useful verdicts in  successor states.

In the following subsections, we present the construction of the monitor $M = (Q,\qinit,\preds,\delta,\verdict)$.
\begin{itemize}
\item   
    $Q$ is defined in \Cref{sec:monitor-states} such that each state corresponds to an $\templogic(\specvars)$ formula.
    The initial state $\qinit$ is the state corresponding to $\Phi$. 

\item   
    The transition function is defined by $\delta(q, a) := \applyRules(\nextState(q, a))$.
    \begin{itemize}
    \item    In \Cref{sec:expansion},  we give $\nextState$ which computes successors using temporal expansion.
    \item   In \Cref{sec:rules}, we present $\applyRules$ which applies semantic transformation rules to a state using logical reasoning.  
    Those rules are the key ingredient of the monitor construction as they enable a more refined verdict assignment.
    \end{itemize}
\item    $\preds$ is defined as $\atoms(\Phi) \cup \propagationPreds \cup \generatedPreds$ where $\propagationPreds$ are predicates introduced in \Cref{sec:expansion} and $\generatedPreds$ are predicates generated by a rule in \Cref{sec:rules}.

\item   
    $\verdict$ is defined in \Cref{sec:verdict-labelling} depending on  the formula represented by each state.
    \end{itemize}

\subsection{Monitor State Space}
\label{sec:monitor-states}
\begin{figure}
$
\begin{array}{ll}
    \subform(\alpha) &:= \{ \alpha, \top, \bot \} \\
    \subform(\lnot \varphi) &:= \subform(\varphi) \cup \{\lnot \psi \mid \psi \in \subform(\varphi)\} \\
    \subform(\varphi_1 \land \varphi_2) &:= \subform(\varphi_1) \cup \subform(\varphi_2) \cup \{ \psi_1 \land \psi_2 \mid \psi_i \in \subform(\varphi_i), i \in \{1,2\} \} \\
    \subform(\LTLnext \varphi) &:= \subform(\varphi) \cup \{\LTLnext \psi \mid \psi \in \subform(\varphi)\} \\
    \subform(\varphi_1 \LTLuntil \varphi_2) &:= \subform(\varphi_1) \cup \subform(\varphi_2) \cup \{ \psi_1 \LTLuntil \psi_2 \mid \psi_i \in \subform(\varphi_i), i \in \{1,2\} \} 
\end{array}
$
\caption{Definition of $\subform : \templogic(\specvars) \to \power{\templogic(\specvars)}$.}\label{fig:subforms}
\end{figure}

We assume w.l.o.g.\ that $\Phi$ is an implication of the form $\Phi_\as \rightarrow \Phi_\ga$, where 
$\Phi_\as$ is a conjunction of \templogic\ formulas representing the specification's assumptions and
$\Phi_\ga$ is a conjunction of \templogic\ formulas representing the guarantees.
Note that this is a common form for specifications.

We construct the states $Q$ using formulas obtained from $\Phi$ 
by applying temporal expansion laws.
To this end, we first define in \Cref{fig:subforms} the \emph{closure} of an \templogic\ formula,  which contains its subformulas and their combinations that can result from applying the expansion laws.

Semantically,  each state $q \in Q$ of the constructed monitor corresponds to an \templogic\ formula, that is, it tracks the temporal properties that have to be satisfied by the infinite sequences in the language associated with that state.
For the purpose of the monitor construction,  we represent $q$ as a tuple of sets of \templogic\ formulas, that represent different parts of the formula associated with $q$.

Formally, a state $q \in Q$ of the constructed monitor is a tuple  
\[q = \langle \formF_\as,\formE_\as,\formF_\ga,\formE_\ga,\Imp_\as,\Imp_\ga\rangle\]
with components which are sets of \templogic\  formulas such that
\begin{itemize}
\item $\formF_\as,\formE_\as,\formF_\ga,\formE_\ga \subseteq \boolcomb(\subform(\Phi) \cup \preds)$, and
\item $\Imp_\as,\Imp_\ga \subseteq \{ \LTLglobally (\gamma \to \varphi) \mid \varphi \in \{ \alpha, \LTLnext \alpha, \LTLeventually \alpha, \LTLglobally \alpha \}, \gamma, \alpha \in \boolcomb(\preds) \}$.
\end{itemize}

We now explain the role  of the individual components of a state $q\in Q$.

Intuitively,  the sets $\formF_\as \cup \formE_\as$ and $\formF_\ga \cup \formE_\ga$ represent the current assumptions and guarantees respectively. 
More precisely, 
we associate with $q$ the \templogic\  formula $\stateformula(q)$  defined as
\[\stateformula(q) := 
\toform{\formF_\as \cup \formE_\as} \rightarrow 
\toform{\formF_\ga \cup \formE_\ga},\]
where for a set $F$ of formulas, we denote by $\toform{F}:= \bigwedge_{\varphi \in F}\varphi$ the conjunction of all  elements of $F$.

The sets $\Imp_A$  and $\Imp_G$ are used to accumulate additional assumptions and guarantees that are ``\emph{implied}'' by any possible path in the monitor from the initial state to $q$.

The separation into the components $\formF_D$ and $\formE_D$ for each $D \in \agset$ facilitates the derivation and utilization of elements in $\Imp_D$, while avoiding circular reasoning.  In particular,  elements of $\Imp_D$ are deduced using the $\formE_D$ component  and used to simplify formulas in the $\formF_D$ component.
The formulas in $\Imp_D$ are all of the form $\LTLglobally \psi$,  meaning that $\psi$ always holds from that state onwards.
Each $\psi$ is an implication with a premise $\gamma$ that is a non-temporal formula,  which allows for checking locally if $\gamma$ is implied by the current state.
We will see more concretely in \Cref{sec:rules},  how the specific form of the states in $Q$ is utilized in the construction of the monitor's transition function.

\paragraph{Finiteness of $Q$}

By definition, $Q$ is required to be finite. 
The set $\subform(\Phi)$ is finite.  
To ensure that the set $\preds$ remains finite,   we
ensure that the sets $\propagationPreds$ and $\generatedPreds$ remain finite throughout the construction of the monitor, 
as explained in \Cref{sec:expansion} and \Cref{sec:rules}, respectively.
While,  there are infinitely many syntactically different Boolean combinations over $\subform(\Phi) \cup \preds$,  there are only finitely many of them up to propositional equivalence. 
For forming the sets $\formF_\as,\formE_\as,\formF_\ga$, and $\formE_\ga$ we use unique representatives of the respective equivalence classes, and hence, there are only finitely many instances of $\formF_\as,\formE_\as,\formF_\ga$, and $\formE_\ga$.
Similarly for $\Imp_\as$ and $\Imp_\ga$.

We remark that the state space of the monitor is constructed on-the-fly, without computing $\subform(\Phi)$ upfront and many elements of $\subform(\Phi)$ are usually not reachable in the monitor.

\paragraph{Initial State}
The formula associated with the monitor's initial state $\qinit$ is $\Phi$.
We define 
$\qinit := \langle \formF^0_\as,\formE^0_\as,\formF^0_\ga,\formE^0_\ga,\emptyset,\emptyset\rangle$ 
where for $D \in \agset$,  the sets $\formF^0_D$ and $\formE^0_D$ partition the set of conjuncts of $\Phi_D$. 
We discuss the concrete choice of elements of $\formF_D$ and $\formE_D$ in~\Cref{sec:rules}.

\begin{example}\label{ex:monitor-initial}
For $\Phi := \LTLglobally (i > 0) \land i = 0\to x = 0 \land \LTLglobally (x' - y = x) \land \LTLglobally \LTLnext x > 0 \land \LTLglobally \LTLeventually x = 10$ 
a possible initial state is
$\langle \emptyset, \{\LTLglobally (i > 0), i = 0 \}, \{\LTLglobally \LTLeventually x = 10\}, \{x = 0, \LTLglobally (x' - y = x), \LTLglobally \LTLnext x > 0 \},\emptyset,\emptyset\rangle$.
\end{example}

\paragraph{Correctness Property}

In the rest of this section, we present the construction of the transition function $\delta$ and the verdict-labelling function $\verdict$,  after which we will show that the constructed monitor satisfies the conditions in \Cref{def:monitor-formula}. 
To this end,  we will show that
\begin{center}
for every $q\in Q$,  $\assignment \in \assignments{\progvars\cup\inputs}$ and 
$\pi\cdot\assignment\cdot\rho \in \assignments{\progvars\cup\inputs}^\omega$ with
$\delta^*_M(\pi\cdot\assignment) = q$:
\end{center}
\begin{equation} \label{eq:monitor-correctness}
\pi\cdot\assignment\cdot\rho\models \Phi \text{ if and only if }\assignment\cdot\rho \models \stateformula(q).
\end{equation}

We refer to (\ref{eq:monitor-correctness}) as the \emph{monitor-state correctness property}.

\subsection{Next-State Construction Using Expansion for \templogic\ Formulas}
\label{sec:expansion}
The states of the monitor correspond to temporal formulas that track properties that have to hold for an infinite sequence from this step on. 
In order to compute the successors of such a state, we split the formula into a part that has to hold immediately in the current step,  and a part that has to hold later in the future.
We do so by applying the expansion laws of the temporal operators~\cite{BaierK08}. 

To construct $\nextState(q,a)$ for a state $q$ and a letter $a$,  we will apply expansion to each of the formulas in the sets in $q$.  Thus, we first define an expansion function for individual formulas,  similar to~\cite{EsparzaKS20, Vardi94}, and then we explain how we use this function to compute the successors of states in $Q$.

\paragraph{Expansion Function for \templogic\ Formulas}
Before we present the formal definition of the expansion operation, consider as a simple example the formula $\LTLglobally (i > 0 \to \LTLnext (x' > x))$. 
This formula, by the expansion laws, is equivalent to $(i > 0 \to \LTLnext (x' > x)) \land \LTLnext\LTLglobally (i > 0 \to \LTLnext (x' > x))$, which splits the formula into a ``current part'' and the part under a next operator.
The latter only matters in future steps.
If $i > 0$ holds in the current letter $a$,  then the formula simplifies to 
$\LTLnext (x' > x) \land \LTLnext\LTLglobally (i > 0 \to \LTLnext (x' > x))$.
This resulting formula has only requirements that have to be checked in the future.
Hence, in this case the successor obligation is $(x' > x) \land \LTLglobally (i > 0 \to \LTLnext (x' > x))$, that is,  the same obligation and the variable $x$ has to become larger in the following step.
Analogously, if $i \leq 0$ holds in the current letter,  then the successor formula is $\LTLglobally (i > 0 \to \LTLnext (x' > x))$. 

We now proceed with the formal definition.
Let us denote with $\BC(\Phi) := \boolcomb(\subform(\Phi)\cup \preds)$ the set of Boolean combinations of the elements of $\subform(\Phi) \cup \preds$.
We define the expansion operation 
$\expand: \BC(\Phi) \times \power{\preds} \to \BC(\Phi)$
to take a formula as used in $Q$ and a subset of predicates of $\preds$ and return the successor element.  The $\expand$ operation folds the application of the expansion laws, the assignment of the predicates in $\preds$, and the removal of  the $\LTLnext$ operator into one:
\[
\begin{array}{lll}
    \expand(\alpha, a) & := \top &\text{ if } \alpha \in a \\
    \expand(\alpha, a) & := \bot &\text{ if } \alpha \not\in a \\
    \expand(\lnot \varphi, a) & := \lnot \expand(\varphi, a) \\
    \expand(\varphi_1 \land \varphi_2, a) & := \expand(\varphi_1, a) \land \expand(\varphi_2, a) \\
    \expand(\LTLnext \varphi, a) & := \varphi \\
    \expand(\varphi_1 \LTLuntil \varphi_2, a) & := \expand(\varphi_2, a) \lor \expand(\varphi_1, a) \land (\varphi_1 \LTLuntil \varphi_2). \\
\end{array}
\]
Note that $\expand$ is well-typed,  as it only generates elements of $\BC(\Phi)$.

The function $\expand$ has the following property,  which is crucial for the correctness of the monitor: An infinite execution $\rho$ that is consistent with an assignment $a$ to the predicates,  satisfies a formula $\varphi$ if and only if the suffix $\rho[1,\infty)$ satisfies $\expand(\varphi, a)$. This is formally stated below.
\begin{restatable}{lemma}{restateLemaExpansion}\label{lem:expansion}
For all $\rho \in {(\assignments{\progvars \cup \inputs})}^\omega$ and $a \subseteq \preds$ where 
$\langle \rho[0], \rho[1] \rangle \FOLentailsT{T} \left( \bigwedge_{\alpha \in \preds \cap a} \alpha \right) \land \left( \bigwedge_{\alpha \in \preds \setminus a} \lnot \alpha \right)$, 
it holds that
$\rho \models \varphi$ if and only if $\rho_{+1} \models \expand(\varphi, a)$.
\end{restatable}

\textit{Remark on computing $\expand$:} 
To compute the expansion we do not enumerate $\power{\preds}$ but branch on occurrences inside the formula, as often not all predicates are relevant to the current step.
Also, we remove selections $a \in \power{\preds}$ that are inconsistent on the theory level, i.e. if the conjuction of predicates in $a$ is unsatisifable, as those selection will never be enabled anyways.
This is a simple, local application of theory-level reasoning.
For example, in $x = 0 \lor x > 0 \land \LTLnext y = 0$, the predicate $y = 0$ is irrelevant for the current step and the selection $\{x = 0, x > 0, \dots\}$ is inconsistent.

\paragraph{Predicate Propagation}

The application of $\expand(\varphi,a)$ checks the ``current'' part of the formula $\varphi$  against the letter $a$.  
However,  as the values of the variables $\progvars'$ are related to the values of $\progvars$ at the next step,   $\varphi$ and $a$ might imply some additional properties that are satisfied at the next step, but which are not propagated into $\expand(\varphi,a)$.
For example, consider the formula $\varphi := x = 0 \land x' = x \land \LTLnext (x = 1)$,  which is unsatisfiable.
However, for $a = \{ x = 0, x' = x, \dots \}$ the expansion will yield $x = 1$ which is satisfiable. 
While this is irrelevant for the correctness of the monitor, 
it is sometimes useful to \emph{propagate these additional properties to the next step, and strengthen $\expand(\varphi,a)$}.
For instance,  in the above example, we would like to propagate the information that $x = 0$ in the next step. 

To this end,  we fix at the start of the monitor construction a set $\propagationPreds \subseteq \QF{\progvars}$ of formulas whose value we want to propagate. This set depends on $\Phi$. One suitable choice is the set of elements of $\atoms(\Phi)$ that belong to $\QF{\progvars}$ plus their negations.
We discuss this further in \Cref{sec:experiments}.

For $a \in \power{\preds}$, we can compute the set of formulas in $\propagationPreds$ that hold in the next step as 
\begin{center}
$\propagate(a) := \left\{ \beta \in \propagationPreds \mid \forall \progvars, \inputs, \progvars'.~\left( \bigwedge_{\alpha \in \preds \cap a} \alpha \right) \land \left( \bigwedge_{\alpha \in \preds \setminus a} \lnot \alpha \right) \to \beta[\progvars \mapsto \progvars'] \text{ is valid}\right\}.$
\end{center}
Intuitively, we just check which elements of $\propagationPreds$ always hold in the next step after our current assignment $a \in \power{\preds}$.
In our example from above,  this would be the case for $x = 0$, since $x = 0 \land x' = x \to x' = 0$ is indeed valid,  and we can propagate $x = 0$ to the successor.

\paragraph{Successor-State Computation}

We have shown how to apply expansion to a single \templogic\ formula, and how to compute additional predicates that we can propagate.
We now put those together,  and define the function $\nextState$. 
Intuitively, we apply $\expand$ individually to all formulas in 
$\formE_\as$, $\formF_\as$, $\formE_\ga$, and $\formF_\ga$.
We add the propagated elements of $\propagationPreds$ to $\formE_\ga$, which maintains our form.
We do not apply $\expand$ to $\Imp_\as$ and $\Imp_\ga$,  since they are not part of $\stateformula(q)$, the formula associated with a state $q$.
In fact,  as we will see in \Cref{sec:rules},  the elements of $\Imp_\as$ and $\Imp_\ga$ are derived from the other components of $q$ and the paths in the monitor leading to $q$.

Formally, for a monitor state $q = \langle \formF_\as,\formE_\as,\formF_\ga,\formE_\ga,\Imp_\as,\Imp_\ga\rangle$ and $a \in \power{\preds}$ we define
$$\nextState(q, a) := \langle 
\expandSet(\formF_\as),
\expandSet(\formE_\as),
\expandSet(\formF_\ga),
\expandSet(\formE_\ga) \cup \propagate(a),
\Imp_\as,\Imp_\ga
\rangle$$ 
where $\expandSet(F) := \{ \expand(\varphi, a) \mid \varphi \in F\}$ for $F\subseteq \templogic(\specvars)$.

\subsection{Rules for Transformation of Monitor States}
\label{sec:rules}
In principle,  using the function $\expand$ defined in \Cref{sec:expansion}, we can construct a monitor for $\Phi$.
Such a monitor, however, will be of little use,  since the verdicts that it can assign will be based purely on the temporal expansion laws and the aforementioned pruned predicate selections.
The goal of our monitor construction,  however,  is to perform a combination of first-order and temporal logic reasoning that will enable the desired game simplifications envisioned in \Cref{sec:examples}.

To achieve this, we define a function $\applyRules$ which performs a sequence of transformations on state $\expand(q,a)$ to compute the successor monitor state $\delta(q,a)$.
These transformations modify the state by applying local theory reasoning, high-level reasoning over the temporal operators, and most importantly reasoning that connects the theory reasoning applied over time with the temporal aspect of the state.
The ultimate aim of these transformations is to transform, when possible a state $q_1$ into a state $q_2$, where 
$\stateformula(q_2)$ is $\false$ or a syntactic safety formula. 
If this is the case, then the monitor can assign the verdict \UNSAT\ or  \SAFETY\ respectively, thus enabling pruning of the game or simplification of the winning condition when building the product of the monitor with the game.
For states that cannot be reformulated to $\false$ or syntactic safety formula, the transformation might allow better pruning of the predicate selection between states or enable useful reasoning in a successor state by simplification or deriving new information.

We will now define a set of allowed transformations,  $\ruleset$, each of which is in the form of a rule whose \emph{premises} characterize the monitor states to which this rule is applicable,  and an \emph{effect} that describes the state modification. 
The transformation rules will make heavy use of the sets $\Imp_\as$ and $\Imp_\ga$, which did not play a role in the definition of $\expand$ which left them unmodified.

As stated earlier,  $\Imp_A$  and $\Imp_G$ are used to accumulate additional assumptions and guarantees that are ``implied'' by any possible path in the monitor from the initial state to a given state $q$. 
More formally,  the rules in $\ruleset$, which will extend and use the sets $\Imp_D$, will be such that the following \emph{implied-state correctness property} is satisfied:
for every $q = \langle \formF_\as,\formE_\as,\formF_\ga,\formE_\ga,\Imp_\as,\Imp_\ga\rangle$, and every
$\assignment \in \assignments{\progvars\cup\inputs}$ and 
$\pi\cdot\assignment\cdot\rho \in \assignments{\progvars\cup\inputs}^\omega$ with
$\delta^*_M(\pi\cdot\assignment) = q$:
\begin{equation} \label{eq:imp-correctness}
\begin{array} {lll}
\text{if }\assignment \cdot \rho \models  \toform{\formE_\as} & \text{ then } & 
\assignment \cdot \rho \models \toform{\Imp_\as}\\
\text{if }\assignment \cdot \rho \models  \toform{\formE_\as  \cup \formE_\ga} & \text{ then } & 
\assignment \cdot \rho \models \toform{\Imp_\ga}.
\end{array}
\end{equation}
Intuitively, this property states that the properties satisfied by any prefix reaching $q$,  together with the current assumptions in the set $\formE_\as$ (respectively assumptions plus guarantees in $\formE_\as  \cup \formE_\ga$) imply the formulas collected in the set $\Imp_\as$ (respectively the set $\Imp_\ga$).

\paragraph{Transformation Soundness.}
In order to guarantee that the resulting monitor satisfies the monitor-sate correctness property~(\ref{eq:monitor-correctness}),  
we will ensure that its states satisfy the implied-state correctness property~(\ref{eq:imp-correctness}).
To facilitate this, we provide a local criterion, termed \emph{sound monitor-state transformation}, for the state transformation rules that guarantees the global property~(\ref{eq:imp-correctness}).

\begin{definition}[Sound Monitor-State Transformation]\label{def:sound-transform}

A \emph{sound monitor-state transformation} is a partial function $\transform : Q \to Q$ such that for every 
$q_1 = \langle \formF_A^1,\formE_\as^1,\formF_\ga^1,\formE_\ga^1,\Imp_\as^1,\Imp_\ga^1\rangle$ and  
$q_2 = \langle \formF_\as^2,\formE_\as^2,\formF_\ga^2,\formE_\ga^2,\Imp_\as^2,\Imp_\ga^2\rangle$
with $\transform(q_1) = q_2$ the following conditions are satisfied:
\begin{itemize}
\item[(a)] $\toform{\formF_\as^1 \cup \formE_\as^1 \cup \Imp_\as^1} \rightarrow 
\toform{\formF_\ga^1 \cup \formE_\ga^1 \cup \Imp_\ga^1}\equiv 
\toform{\formF_\as^2 \cup \formE_\as^2 \cup \Imp_\as^2} \rightarrow 
\toform{\formF_\ga^2 \cup \formE_\ga^2 \cup \Imp_\ga^2}
$
\item[(b)] the following two implications are valid
$\begin{array}{lll}
\toform{\formE_\as^2 \cup  \Imp_\as^1} & \rightarrow & 
\toform{ \Imp_\as^2},\\ 
\toform{\formE_\as^2 \cup \formE_\ga^2 \cup \Imp_\ga^1 \cup \Imp_\as^1}  & \rightarrow & \toform{ \Imp_\ga^2}.
\end{array}$
\item[(c)] the implications 
$\toform{\formE_D^2} \rightarrow  \toform{ \formE_D^1}$  for $D \in \agset$ are valid.
\end{itemize}
\end{definition}
Condition~(a) requires equivalence  on the formula level,  including the implied formulas,  while condition (b) is the localized version of the global property~(\ref{eq:imp-correctness}).  Condition (c) states that $E_D$ can only be strengthened.
Any rule that defines a transformation satisfying the conditions in \Cref{def:sound-transform} can be added to our set of rules $\ruleset$ while maintaining correctness of the monitor.

We define several operations and useful sets of formulas that are used in the rules' formalization.

\paragraph{Substitution in Rules.}

The effect of a rule is often described by means of substitution.  
For a set $F\subseteq \templogic(\specvars)$ of formulas and formulas $\varphi,\psi \in \templogic(\specvars)$  we denote with 
$F[\varphi\mapsto\psi]$ the set of formulas obtained from $F$ by simultaneously substituting in each of them each occurrence of $\varphi$ by $\psi$.
Furthermore, we denote by 
$F[\varphi \mapstonn \psi]$ the set of formulas obtained  from  $F$ by simultaneously substituting in each of them each \emph{non-nested} occurrence of $\varphi$ by $\psi$.
Here, a \emph{non-nested} occurrence of $\varphi$ in a formula $\theta$ is one that is not in the scope of a temporal operator.

\paragraph{Useful State Components.}

For the premises of rules, two types of formulas extracted from a state's components  are of particular interest.  
These are the non-temporal assumptions and guarantees that must be satisfied in the \emph{current} step, and the state properties that must hold for all steps in the future.  
We refer to the latter as the \emph{implied invariants} associated with the given state.
Formally, for a state $q = \langle \formF_\as,\formE_\as,\formF_\ga,\formE_\ga,\Imp_\as,\Imp_\ga\rangle$, we define the following formulas in $\QF{\specvars}$:
$$
\begin{array}{lll}
\Curr_\as(q) & := & \toform{\{ \alpha \in \QF{\specvars} \mid \alpha\in \formE_\as \}},\\
\Curr_\ga(q) & := & \toform{\{ \alpha \in \QF{\specvars} \mid \alpha\in \formE_\ga \cup \formE_\as \}}
\end{array}
$$
representing the current-state assumptions and guarantees, and 
$$
\DedInv_D(q) := \toform{\{ \alpha \in \QF{\specvars} \mid \LTLglobally \alpha\in \Imp_D \}} \quad\quad \text{ for } D \in \{\as, \ga\}
$$
representing the implied invariants.
Note that we add the current assumptions $\Curr_\as(q)$ as part of the current guarantees, as the guarantees must be satisfied only if the assumptions are.

\begin{example}
    Reconsider the state $q$ from~\Cref{ex:monitor-initial} where additionally
    $\Imp_\as = \{ \LTLglobally (\top \to i > 0) \}$ and
    $\Imp_\ga = \{ \LTLglobally (\top \to x' - y = x), \LTLglobally (\top \to \LTLnext x > 0) \}$. 
    We have that 
    $\Curr_\as(q) = (i = 0)$, 
    $\Curr_\ga(q) = (i = 0 \land x = 0)$, 
    $\DedInv_\as(q) = i > 0$, and
    $\DedInv_\ga(q) = x' - y = x$.
\end{example}

We are now ready to present the state transformation rules of our monitor construction.
We group the rules in two categories: \emph{formula simplification rules} and \emph{rules for deriving implied formulas}.
Intuitively,  the formula simplification rules simplify the formulas represented by the sets $\formF_D$ and $\formE_D$ for $D \in \agset$, thus possibly bringing the state closer to one where the associated formula is such that a verdict can be assigned.
The rules for deriving implied formulas extend the sets $\Imp_D$ possibly enabling further applications of the simplification rules.

\subsubsection{Simplification Rules}
The formula simplification rules, shown in~\Cref{fig:simp-rules}, serve two purposes. 
First,  rewriting the formulas in a state $q$ into simpler ones, or into forms that facilitate the application of other rules.
Second,  replacing formulas or sub-formulas in $q$ with constants (i.e.,  with $\true$ or $\false$). We propose three types of simplification rules: rules for \emph{formula rewriting}, 
rules for \emph{detecting unsatisfiability} of $\toform{\formF_\as \cup \formE_\as \cup \Imp_\as}$ or $\toform{\formF_\ga \cup \formE_\ga \cup \Imp_\ga}$ and rules for \emph{sub-formula substitution}.

\paragraph{Formula Rewriting.}
As \emph{formula rewriting} rules we consider state transformations based on known equivalences between \templogic\  formulas (in fact, between LTL formulas).  As one example,  consider the replacement of every occurrence of $\varphi \LTLweakuntil \false$ in $q$ by $\LTLglobally\varphi$.
These rules trivially guarantee Conditions~(a) and (c) and $\toform{\Imp_D} \equiv \toform{\Imp_D'}$, which ensures soundness of the resulting transformation.

One noteworthy rule of this type is \nextextend, shown in \Cref{fig:simp-rules}.  It makes explicit the connection between the next-state variables $\progvars'$ and the next-state temporal operator $\LTLnext$. 
To this end,  for a formula $\alpha\in \QF{\progvars}$, the rule \nextextend\ simultaneously replaces in every formula in $q$ every occurrence of $\alpha[\progvars \mapsto \progvars']$ or $\LTLnext \alpha$ by the equivalent formula $\alpha[\progvars \mapsto \progvars'] \wedge \LTLnext\alpha$.

\begin{figure}
\input{rules_simplify}
\caption{Formula simplification rules.  Each rule is given by its set of premisses and its effect on the monitor state, where $D \in \agset$.  Components of the state not assigned in the effect are not modified. by the rule.}
\label{fig:simp-rules}
\end{figure}

\paragraph{Detecting Unsatisfiability.}
We propose several rules for detecting logical inconsistency,   which play a crucial role in identifying states $q$ where the assumptions or the guarantees are  equivalent to $\false$.

\smallskip

\noindent
\textbf{Rule \substunsat} is applicable when for some $D \in \agset$ the formula $\Curr_D(q)$, which should be satisfied at the current time-point and the formula $\DedInv_D(q)$ which should be satisfied for all future time-points, including the current,  contradict each other.  In such case,  we replace the assumption or guarantee (depending on $D$) component of $q$ by $\false$.  As a result, a subsequent application of the formula rewriting transformation rules will reduce the formula associated with $q$ to a constant. 

\smallskip

\noindent
\textbf{Rule \substunsatF} detects contradicting liveness and invariants assumptions or guarantees.  It is applicable when there is an element of some $\Imp_D$ of the form 
$\LTLglobally(\gamma\rightarrow\LTLeventually \beta)$ with $\beta \in \QF{\specvars}$ such that $\gamma$ is must hold in the current time-point and $\beta$,  which in this case must be satisfied eventually,  contradicts the invariants in $\Imp_D$.
In such case,  $\toform{\formF_D \cup \formE_D \cup \Imp_D}$ is unsatisfiable.
 
 \smallskip 
Note that since  $\Curr_D(q),  \DedInv_D(q)$  and $\beta$ are formulas in our background logical theory, the premise of rules 
\substunsat\  and \substunsatF\ can be checked using an SMT solver.

\paragraph{Sub-Formula Substitution.}
The rules above are applicable when the overall assumption or guarantee component of a state is unsatisfiable. 
The rest of the rules shown in \Cref{fig:simp-rules}, on the other hand,  are applicable to individual sub-formulas,  and allow for substituting them by the constants $\true$ or $\false$, or by some of their sub-formulas. 
Overall,  this leads to simplification of the state.

To avoid circular reasoning,  the substitution operations are applied only to the $\formF_\as$ and $\formF_\ga$ components of $q$,  which are not used for deriving the  $\Curr_D(q)$  and $\DedInv_D(q)$ formulas.

\smallskip

\noindent
\textbf{Rules \substtrue\ and \substfalse} consider formulas $\gamma \in \QF{\specvars}$
that appear as sub-formulas in some $\formF_D$.
If $\gamma$ (or its negation) is entailed by the implied invariants  $ \DedInv_D(q)$, 
we can replace $\gamma$ by $\true$ (or $\false$, respectively) in $\formF_D$,   
since the elements of $\formF_D$ are not used in the derivation of $\Imp_D(q)$.

\smallskip
\noindent
\textbf{Rules \simplifyimp\ and \simplifyand} directly match elements of the set $\Imp_D$ for some $D \in \agset$ with sub-formulas in $\formF_D$. 
Using the fact that $\gamma \rightarrow \varphi$ must hold at every point in time, 
we can apply the  respective substitutions to all occurrences of $\gamma \rightarrow \varphi$ (respectively $\gamma \land \varphi$) in $\formF_D$.

\smallskip
\noindent
\textbf{Rule \simplifynn.}
The previous sub-formula substitution rules replace formulas entailed by,  or elements of,  $\Imp_D$. 
Since each formula $\gamma \rightarrow \varphi$ with $\LTLglobally(\gamma \rightarrow \varphi) \in \Imp_D$ must be satisfied at all points in time,  these rules can apply replacement of any occurrences of the respective formulas in $\formF_D$.
The last simplification rule,  \simplifynn, on the other hand, considers elements 
$\LTLglobally(\gamma \rightarrow \varphi)$ of $\Imp_D$  where 
the formula $\gamma$ is entailed by the conjunction of the implied invariants and the formula $\Curr_D(q)$.
Since $\Curr_D(q)$ is only required in the current time-point,  the rule \simplifynn\ only replaces occurrences of $\varphi$ pertaining to the current point in time, that is, those not in the scope of a temporal operator.

\begin{example}
    Consider a state $\langle \emptyset, \emptyset, \emptyset, \formE_\ga, \emptyset, \Imp_\ga \rangle$ where $\formE_\ga = \{ x = 0 \}$ and $\Imp_\ga = \{ \LTLglobally (x \neq 0)\}$. 
    Now $\Curr_\ga$ is $x = 0$ and $\DedInv_\ga$ is $x \neq 0$.
    Hence, by \substunsat\ we get the state $\langle \emptyset, \emptyset, \{ \bot \}, \{ \bot \}, \emptyset, \Imp_\ga \rangle$.
    Similarly, if we had instead  $\Imp_\ga = \{ \LTLglobally (x = 0 \to \LTLeventually y = 0), \LTLglobally (y \neq 0)\}$, we would have that $\DedInv_\ga$ is $y \neq 0$ and therefore we get again $\langle \emptyset, \emptyset, \{ \bot \}, \{ \bot \}, \emptyset, \Imp_\ga \rangle$, this time by applying  \substunsatF\ .
\end{example}

\subsubsection{Rules for Deducing Formulas in $\Imp_D$}
The transformation rules presented so far do not modify the sets $\Imp_\as$ and $\Imp_\ga$, but rely on their elements for 
the satisfaction of their premises. 
Our second category of rules derives implied formulas and adds them to  $\Imp_\as$ and $\Imp_\ga$.
We distinguish three types of such rules.
First, we have rules that \emph{propagate formulas to $\Imp_\as$ and $\Imp_\ga$} from the sets $\formE_\as$ and $\formE_\ga$ respectively. 
The second type of rules are of crucial importance, as they  \emph{generate invariants and reachability properties} that are 
guaranteed to hold for the infinite sequences in the language of the respective state.
These rules integrate temporal and first-order logical reasoning to derive consequences which can enable the application of the rules in the first category.
Finally,  we have a set of rules that allow us to \emph{combine different elements of each of the sets $\Imp_D$} for $D\in\agset$.

All the formulas that the rules add to the set $\Imp_\as$ are implied by formulas in $\formE_\as \cup \Imp_\as$,
and all the formulas added to $\Imp_\ga$ are implied by $\formE_\as \cup \formE_\ga \cup \Imp_\ga \cup \Imp_\as$. 
Furthermore,  they do not modify the sets $\formE_\as$ and $\formE_\ga$.
Thus, they satisfy condition~(b) in \Cref{def:sound-transform}.
The rules modify $\Imp_D$  by adding new elements to them, and leave the sets $\formF_D$ and $\formE_D$ unchanged.  Therefore,  the resulting transformations satisfy conditions~(a) and (c) in \Cref{def:sound-transform}.

\paragraph{Propagating Formulas to $\Imp_D$.}
The rules \propagateG\ and \propagateW\  shown in \Cref{fig:ded-rules} add to the set $\Imp_D$  formulas of the appropriate form that are present in the set $\formE_D$ or entailed by formulas in $\formE_D$.  
The rule  \propagateassump\ ensures that the formulas implied by the assumptions,  i.e.,  those in $\Imp_\as$ are also present in the set $\Imp_\ga$.

Rule  \propagateG\ directly matches formulas in $\formE_D$ (modulo simple equivalence rewriting).
Rule \propagateW\, identifies pairs of formulas $\alpha_1 \LTLweakuntil \beta_1$ and $\alpha_2 \LTLweakuntil \beta_2$ where the ``end-condition'' formulas $\beta_1$ and $\beta_2$ cannot be satisfied simultaneously and neither $\beta_1$ can be satisfied while $\alpha_2$ is true, nor the other way around.  
In this case,  the two formulas together entail that $\alpha_1 \land \alpha_2$ should hold always meaning that we can add $\LTLglobally(\true \rightarrow \LTLglobally (\alpha_1 \land \alpha_2))$ to $\Imp_D$.

\paragraph{Combination of Implied Formulas.}
The transformation rules  shown in \Cref{fig:saturate-rules} enable the extension of the sets $\Imp_\as$ and $\Imp_\ga$ by combining existing elements of the respective set.
Rules \joinimp\ and \chainimp\ allow for weakening the left-hand side of the implications by combining elements of $\Imp_D$. 
Rule \chainimpG\ enables the derivation of $\LTLglobally \alpha$ from 
$\LTLglobally (\gamma \rightarrow\LTLglobally\alpha)$ by establishing that $\gamma$ is entailed.
Finally,  rules \chainimpF\ and \chainimpN\ allow for strengthening the right-hand side of implications when they are of the form $\LTLeventually \beta$ or $\LTLnext \beta$.

\begin{figure}
\input{rules_saturate}
\caption{
Rules for saturating the sets $\Imp_D$ for $D \in \agset$.
These rules modify only $\Imp_D$.}
\label{fig:saturate-rules}
\end{figure}

\begin{figure}
\input{rules_deduce}
\caption{Rules for deducing implied formulas. 
These rules modify only the set $\Imp_D$ for some $D \in \agset$.}
\label{fig:ded-rules}
\end{figure}

\paragraph{Generating Invariants and Reachability Properties}
We  now turn to the set of rules that generate new invariants, that is,  formulas of the form 
$\LTLglobally(\gamma \rightarrow \LTLglobally\alpha)$,  
and new reachability properties, i.e.,  formulas of the form   
$\LTLglobally(\gamma \rightarrow \LTLfinally\beta)$.
The  generated formulas are entailed by the current $\Imp_D$ and $\formE_D$ sets, 
and are added to $\Imp_D$.
The ability to generate new implied invariants and reachability properties is crucial for ``discharging'' temporal obligations or detecting ``temporal conflicts''.

\smallskip

\noindent
\textbf{Rule \geninv}, shown in \Cref{fig:ded-rules}, establishes that whenever  $\gamma$ is satisfied,  then all sequences of valuations that satisfy the invariants in $\DedInv_D(q)$ must also satisfy $\alpha$ at every point in time.  
This justifies the addition of the formula $\LTLglobally(\gamma \rightarrow \LTLglobally\alpha)$ to $\Imp_D$.
For example, if $\gamma := (x = 0)$ and $\LTLglobally (x' \geq x)$ holds, \geninv\ can prove that $\alpha := (x \geq 0)$ always holds, i.e. $\alpha$ is an invariant.
The formula $\theta$ is a strengthening of $\alpha$,  which plays the role of an inductive invariant.
To apply this rule we need to supply the formulas $\alpha$ and $\gamma$ and check the respective entailments.
A suitable choice is taking $\gamma :=\Curr_D(q) $ and $\alpha := \neg \beta$ for some $\beta$ where $\LTLeventually\beta$ or $\alpha_1 \LTLweakuntil \beta$ appears in $\formF_D$, which would subsequently enable the detection of liveness sub-formulas that cannot be satisfied.
As for $\theta$,  we can either take $\theta:=\alpha$, or if this does not succeed, check for existence of $\theta$ using for instance a solver for Constrained Horn Clauses (CHC).

\begin{example}\label{ex:geninv}
    Consider a state $\langle \emptyset, \emptyset, \formF_\ga, \formE_\ga, \emptyset, \Imp_\ga \rangle$ where $\formE_\ga = \{ x = 0, \LTLglobally x = y, \LTLglobally x' \geq x \}$, $\formF_\ga = \{ y = 0 \to \LTLeventually x = -5\}$, and $\Imp_\ga = \emptyset$.
    \propagateG\ adds from $\formE_\ga$ the formulas $\{ \LTLglobally (\top \to x = y), \LTLglobally (\top \to x' \geq x) \}$ to $\Imp_\ga$. 
    We try to apply \geninv\ for $\alpha := \lnot (x = -5)$ and $\gamma := (x = 0)$.
    As $x' \geq x$ always holds, \geninv\ is applicable with the strengthening $\theta := x \geq 0$.
    Hence, we add $\LTLglobally (x = 0 \to \LTLglobally \lnot (x = -5))$ to $\Imp_\ga$ and by 
    \chainimpG\ also $\LTLglobally (\top \to \lnot (x = -5))$.
    This allows \substfalse\ to replace $x = -5$ in $F_\ga$ by $\bot$ such that $\formF_\ga := \{ y \neq 0 \}$.
    Now $\Curr_\ga = y \neq 0 \land x = 0$ and $\DedInv_\ga$ contains $x = y$.
    Hence, by \substunsat\ we get the state $\langle \emptyset, \emptyset, \{ \bot \}, \{ \bot \}, \emptyset, \Imp_\ga \rangle$.
\end{example}

\smallskip

\noindent
\textbf{Rule \geninvp.}
The rule \geninv\ requires the candidate invariant $\alpha$. 
This is useful for identifying liveness sub-formulas present in $q$ that cannot be satisfied. 
However, we would also like to be able to derive the \emph{most precise} invariant, which might be useful to simplify successors of the state $q$.
This is the goal of \geninvp\ in \Cref{fig:ded-rules} which computes $\alpha$ as a least fixpoint.
This fixpoint is the smallest set of assignments to the program variables $\progvars$ consisting of the assignments where $\gamma$ holds (the first disjunct in the fixpoint), and any possible successor assignments restricted by $\DedInv_D(q)$ (the second disjunct). 
Note that in the fixpoint we reformulate  $\DedInv_D(q)$ as a ``predecessor-assignment current-assignment relation'' via variable renaming.
Intuitively, $\alpha$ characterizes the set of assignments that are reachable when the implied invariants hold.
A suitable application of \geninvp\  is with $\gamma :=\Curr_D(q)$, which subsequently allows us to conclude that $\LTLglobally \alpha$ must be satisfied and add it to $\Imp_D$ using the rule \chainimpG\ for expanding $\Imp_D$.

\geninvp\ might generate  new predicates that are not in the set $\preds$ yet. 
We add those to $\generatedPreds$ but need to take care to keep $\generatedPreds$ finite.
Hence, we limit how many times we apply \geninvp\ to the same $\formE_D$.  Since the number of possible $\formE_D$ sets is finite and independent of $\generatedPreds$,  this ensures that the set $\generatedPreds$ remains finite.

\begin{example}\label{ex:geninvp}
    Consider a state $\langle \emptyset, \emptyset, \formF_\ga, \formE_\ga, \emptyset, \Imp_\ga \rangle$ where $\formE_\ga = \{ x = 1 \}$, $\formF_\ga = \{ \LTLnext y = 1, \LTLnext \LTLglobally y = x' \}$, and $\Imp_\ga = \{ \LTLglobally (\top \to x' > x), \LTLglobally (y = 1 \to \LTLeventually y < 0)  \}$.
    Here \geninvp\ can derive for $\gamma := x = 1$, the most precise invariant $\alpha := x \geq 1$, such that we add $\LTLglobally (x = 1 \to \LTLglobally x \geq 1)$ to $\Imp_\ga$.  By \chainimpG\ we also add $\LTLglobally (\top \to x \geq 1)$ to $\Imp_\ga$. 
    After an expansion step and moving formulas from $\formF_\ga$ to $\formE_\ga$ we get a successor
    $\langle \emptyset, \emptyset, \emptyset, \formE_\ga', \emptyset, \Imp_\ga' \rangle$ with
    $\formE_\ga' = \{ y = 1, \LTLglobally y = x' \}$ and
    $\Imp_\ga' = \{ \LTLglobally (y = 1 \to \LTLeventually y < 0), \LTLglobally (\top \to x \geq 1),\LTLglobally (\top \to y = x'), \LTLglobally (\top \to x' > x), \dots \}$.
    This state is simplified by \substunsatF\ to $\langle \emptyset, \emptyset, \{ \bot \}, \{ \bot \}, \emptyset, \Imp_\ga \rangle$.
    This reasoning is not possible with our remaining rules without \geninvp\ .
\end{example}

\smallskip

\noindent
\textbf{Rule \genreach } generates implied reachability properties.
Those are useful not only for detecting invariants that must be violated, but also for ``discharging'' $\LTLeventually \beta$ sub-formulas that are implied by other components.
The rule requires us to provide a formula $\beta$ and then computes as a least fixpoint a formula $\gamma$ representing all the valuations from which all sequences that satisfy the invariants in $\DedInv_D(q)$ eventually satisfy $\beta$. 
This justifies the addition of  $\LTLglobally(\gamma \rightarrow \LTLeventually\beta)$ to $\Imp_D$.
A suitable application  is with formulas $\beta$ such that $\LTLglobally\neg \beta$ or $\LTLeventually \beta$ appears in $\formF_D$, which enables subsequently the detection of a contradiction or that the formula is satisfied, respectively.  

\begin{example}\label{ex:genreach}
    Consider a state $\langle \emptyset, \emptyset, \formF_\ga, \formE_\ga, \emptyset, \Imp_\ga \rangle$ where
    $\formE_\ga = \{ x = 0, \LTLglobally x ' > x\}$,
    $\formF_\ga = \{ \LTLeventually x > 1000 \}$, and
    $\Imp_\ga = \{ \LTLglobally (\top \to x ' > x)\}$.
    For $\beta := x > 1000$ and $\gamma := x = 0$, \genreach\ adds $\LTLglobally (x = 0 \to \LTLeventually x > 1000)$ to $\Imp_\ga$.
    Hence, \simplifynn\ replaces $\LTLeventually x > 1000$ in $\formF_\ga$ to $\top$
    which results in the state $\langle \emptyset, \emptyset, \emptyset, \{ x = 0, \LTLglobally x ' > x \}, \emptyset, \{  \LTLglobally (\top \to x ' > x), \LTLglobally (x = 0 \to \LTLeventually x > 1000) \} \rangle$.
\end{example}

\subsubsection{Application of State Transformation Rules}

In the following, we discuss how and in which order we apply the transformations in $\ruleset$,  in order to compute the function $\applyRules$.

To establish the correctness of function $\applyRules$,  we prove the following statement.
\begin{restatable}{theorem}
{restateRuleSoundness}\label{thm:rule-soundness}
Each rule in $\ruleset$ defines a sound monitor-state transformation.
\end{restatable}

\paragraph{Partitioning into $\formF_D$ and $\formE_D$}

We first discuss what parts of the assumptions and the guarantees we put in the sets $\formF_D$ and $\formE_D$ for $D = \as$ and $D = \ga$, respectively.
Since we use $\formE_D$ to derive the elements of $\Imp_D$, 
and use the elements of $\Imp_D$ to perform substitutions in $\formF_D$, we are only allowed to move formulas from $\formF_D$ to $\formE_D$,  and \emph{not vice versa}, in order to avoid circular reasoning.

As the rules  utilize from  $\formE_D$ only the current obligations in $\Curr_D(q)$ and formulas with top-level $\LTLglobally$ and $\LTLweakuntil$ operators with relatively simple arguments,  we place only such formulas in $\formE_D$.
Note that in that way formulas of the form $\LTLeventually \varphi$ remain in $\formF_D$, and hence,  can potentially be ``discharged''.

\paragraph{Rule Application Order}

We apply $\propagateassump$ and $\propagateG$ always when we modify $\formE_D$ and $\Imp_D$.
Furthermore, we apply formula rewriting to simplify formulas on the logical level whenever possible. 
As for the other rules, we apply the in the following order:
\begin{enumerate}
    \item[1.]
        We first apply $\substunsat$ and $\substunsatF$ to check early if the formula is contradictory.
    \item[2.]
        As expansion and rule $\propagateG$ might enlarge $\Imp_D$, we apply the chain rules and $\propagateW$ to saturate $\Imp_D$ with existing information to ease the next steps.
    \item[3.] 
        Next, we use the simplification rules $\substtrue$, $\substfalse$, $\simplifyimp$, $\simplifyand$ and $\simplifynn$ to simplify $F_D$ utilizing the saturated  set $\Imp_D$.
    \item[4.]
        After that, we apply the generation rules $\geninv$, $\genreach$, and $\geninvp$.
        We apply these rules last, as they are computationally costly, and thus we want to use them once the previous rules have possibly already simplified the state. 
    \item[5.] 
        As the previous step might have generated new information and added new derived properties to the sets $\Imp_D$ of implied formulas, we reiterate steps 1. -- 3. once, in order to take advantage of the new elements of $\Imp_D$ and use them e.g. to substitute further sub-formulas.
\end{enumerate}

\subsection{Discharging Implied Liveness Obligations}
\label{sec:liveness}
Note that the function $\applyRules$ can potentially discharge non-nested formulas of the form $\LTLeventually \beta$ (that is, substitute them with $\true$) by combining \genreach\ and \simplifynn, if they are implied by other parts of the specification.
This allows us to sometimes simplify specifications that contain $\LTLeventually$ to safety ones.
However, this does not work for the common $\LTLglobally\LTLeventually\beta$-pattern as it requires checking the global property that  $\LTLeventually\beta$ can be discharged again and again.

Therefore, after computing the monitor's transition function and reachable states, and before defining $\verdict$, we apply the following post-processing that reasons on a global level.
We focus on $\LTLglobally\LTLeventually\beta$ formulas for this global analysis, as those are fairly common.
To discharge some of those non-nested $\LTLglobally\LTLeventually \beta$ obligations, we proceed as follows.
Let
$Q_\mathit{Triv} := \{q \in Q \mid \stateformula(q) \in \{ \top, \bot \} \}$.

For every formula $\LTLglobally\LTLfinally \beta$ with $\beta\in\QF{\specvars}$ that appears non-nested in the set ${\formF_\ga}_0$
for some state $q_0 = \langle {\formF_\as}_0,{\formE_\as}_0,{\formF_\ga}_0,{\formE_\ga}_0,{\Imp_\as}_0,{\Imp_\ga}_0\rangle$ we do the following:
\begin{enumerate}
    \item[1.] We compute $Q_{\tiny\LTLeventually \beta} := \{q \in Q \mid \exists \LTLglobally(\gamma\rightarrow\LTLeventually\beta) \in \Imp_\ga(q).~\Curr_\ga(q) \wedge \DedInv_\ga(q)  \FOLentailsT{T} \gamma \}$, i.e.\ the states from which  $\LTLeventually \beta$ as obligation is already implied by $E_\ga$ and $\Imp_\ga$. 
    \item[2.] We compute $A_{\tiny\LTLeventually \beta} := \{q \in Q \mid \text{ every path from } q \text{ reaches } Q_{\tiny\LTLeventually \beta} \cup Q_\mathit{Triv} \}$,  i.e., the states from which $\LTLeventually \beta$ is always implied or the obligation becomes trivial.
    \item[3.] We compute $A_{\tiny\LTLglobally\LTLeventually \beta} := \{q \in Q \mid \text{ every path from } q \text{ contains only states from } A_{\tiny\LTLeventually \beta}\}$.
\end{enumerate}    
$A_{\tiny\LTLglobally\LTLeventually \beta}$ contains only states where we know that $\LTLeventually \beta$ always holds as long as the other requirements do.
Hence, we replace every non-nested occurrence of $\LTLglobally\LTLfinally \beta$ in ${\formF_\ga}_0$ of each $q_0 \in A_{\tiny\LTLglobally\LTLeventually \beta}$ by $\true$.

The construction of $A_{\tiny\LTLglobally\LTLeventually \beta}$ guarantees that if before the above transformation the monitor satisfied the monitor-state correctness property (\ref{eq:monitor-correctness}), then after the transformation it does so as well.

\subsection{Construction of the Verdict-Labelling Function}\label{sec:verdict-labelling}
If a state $q$ of the monitor is such that $\stateformula(q) = \false$,  no words that have a prefix reaching $q$ can satisfy $\Phi$.
Hence, the monitor can assign verdict $\UNSAT$ to state $q$. 
For states where $\stateformula(q)$ is not $\false$, we cannot always assign a verdict. 
In general, it is not always possible to determine satisfaction of a formula based on a finite prefix.
However,  if $\stateformula(q)$ represents a safety language,  the monitor can assign a $\SAFETY$ verdict, since expansion is sufficient for tracking safety requirements.
For our verdict-labelling function, we consider a syntactic criterion for safety temporal formulas.

Let
$Q_\false := \{q \in Q \mid \stateformula(q) = \bot\}$ and 
$Q_\mathit{safe}$ be the set of all states $q_s$ such that no $q_l \in Q$ where $\stateformula(q_l)$ is not a syntactic-safety formula is reachable.
We define $\verdict$ as follows.
\[
\verdict (q) : = \begin{cases}
\UNSAT & \text{if } q \in Q_\false,\\
\SAFETY & \text{if } q \in Q_\mathit{safe} \\
\OPEN & \text{otherwise}.\\
\end{cases}
\]

$\verdict$ is well-defined as for $q \in Q_\false$, $\delta(q,a) \in Q_\false$ for any $a \in \power{\preds}$.
Furthermore, any $q \in Q_\mathit{safe}$ is by definition syntactic-safety.
As expansion preserves syntactic-safety and the rules do not add any non-safety formulas to $E_D$ or $F_D$, any successors of $q$ is also either syntactic-safety or in $Q_\false$.

\begin{example}
In~\Cref{ex:geninv} and~\Cref{ex:genreach} the states get an \UNSAT\ and \SAFETY\ verdict, respectively, after applying the rules.
Without the rules they would both get an \OPEN\ verdict.
\end{example}

\begin{restatable}{theorem}
{restateMonitorCorrectness}\label{thm:monitor-correctness}
    Let $M$ be a monitor constructed from $\Phi \in \templogic(\specvars)$ as described in \Cref{sec:monitor}.
    Then,  $M$ satisfies the conditions in \Cref{def:monitor-formula} and is hence a monitor for $\Phi$.
\end{restatable}

\section{Experimental Evaluation}\label{sec:experiments}

\subsection{The Logic \tslmt\ and Reactive Program Games}
In our implementation, we use Temporal Stream Logic modulo Theories (TSL-MT)~\cite{FinkbeinerKPS19,FinkbeinerHP22} and reactive program games (RPGs) ~\cite{HeimD24}  instead of \templogic\ and the symbolic games, respectively. 
Similar to \templogic, TSL-MT allows us to specify temporal properties over inputs $\inputs$ and program variables $\progvars$ (often called cells,  and denoted by $\mathbb{C}$).
The only difference is that instead of generic predicates over the next-step program variables $\progvars'$, TSL-MT is restricted to so called \emph{updates}. 
Intuitively, updates behave like program variable assignments.
Denoted as $\upd{x}{t}$, an update assigns the value of the term $t$ (which might include program variables and inputs) to $x$ for the next state.
Different updates of the same variable are mutually exclusive at the same time.
Analogously,  RPGs allow the system player to only select updates.
This restriction makes extracting programs easier, as the system is restricted to a finite number of known assignments.
TSL-MT and RPGs can be encoded in the more general framework of \templogic\ and symbolic games such that our results apply. 
Furthermore, we believe that the game-solving acceleration technique from~\cite{HeimD24} could be lifted to our symbolic games.
Here we consider TSL-MT and RPGs, as for those tools and benchmarks exist.

\subsection{Prototype Tool Implementation}

We implemented our approach in our prototype tool \toolname\ \footnote{Available at \url{https://doi.org/10.5281/zenodo.13939202}.}.
The generation of the monitor is done on-the-fly while constructing the product.
This allows us to only consider predicate selections actually appearing in the game and to only explore states that matter for the game.
For $\propagationPreds$ we use the predicates that appear in the specification, the negations, and equalities for updates with constants.
We disable the \geninvp\ rule by default, as computing exact fixpoints is time-consuming, but often not needed.
For the automata constructions we use Spot~\cite{Duret-LutzRCRAS22}.
As SMT and CHC solver we use \texttt{z3}~\cite{Z3}.
For the rules \genreach\ and \geninvp\ we use the $\mu$CLP solver \texttt{MuVal}~\cite{UnnoTGK23} and the Optimal CHC solver \texttt{OptPCSat}~\cite{GuTU23}, respectively.
For solving the generated reactive program game we use \texttt{rpgsolve}~\cite{HeimD24}.

\subsection{Benchmarks}

In the following, we discuss the content, purpose, and results of our benchmark sets\footnote{All benchmarks are available at \url{https://doi.org/10.5281/zenodo.13939202}.}.
Note that we restrict our benchmarks to integers, since reals cannot be handled by some of the tools and solvers we compare to, although \toolname\ also supports reals.
We begin with the description of the benchmarks we used in the evaluation,  grouped into four categories listed below.

\paragraph{Benchmarks from the Literature}
We evaluated \toolname\ on the existing TSL-MT benchmarks from~\cite{NeiderT16} and~\cite{MaderbacherB22}.
We do not include benchmarks from~\cite{ChoiFPS22} as some contain uninterpreted functions and others are trivially realizable with a system violating the environment assumptions.

\paragraph{Basic Properties}
The second set contains benchmarks with simple reachability properties and invariants implied by or contradicted by some of the other invariants.
They do not need complicated strategic decisions and can be handled mainly on the language level.

\paragraph{Scenarios}
We created a set of more intriguing scenario benchmarks.
This includes cyber-physical-system controllers and robot mission planning tasks.
They are larger than the previous benchmarks and need more complex reasoning by the game solver.
They mostly require strategic decisions at some crucial points and are fairly deterministic in-between.
This is common for specifications of more complex systems.

\paragraph{Limitations}
With the last set of benchmarks we analyze the limitations of our monitor-based method.
These are benchmarks were a language-level analysis cannot effectively prune the game, the \geninvp\ rule is actually needed, or fixpoints in the different rules are too difficult to compute.

\subsection{Results and Analysis}

\begin{table}[t!]
\caption{%
Evaluation Results. 
$|\progvars|$, $|\inputs|$ are the number of respective variables.
R shows if the benchmark is expected to be realizable.
We show the wall-clock running time in seconds of \toolname\ with \textbf{mo}nitor and \textbf{wi}thout monitor, and of \textbf{Ra}boniel and \textbf{Te}MoS.
TO means timeout after 20 minutes, MO means out of memory (8GB), - means the tool is not applicable, and ER means error (by actual error or incorrect result).
The evaluation was performed on a Intel i7 (11thGen) processor.
}\label{tab:all}
\begin{minipage}{.48\linewidth}
\centering
\scalebox{0.65}{
\begin{tabular}[t]{l|ccc|rr|rr}
    Name & $|\progvars|$ & $|\inputs|$ & R & mo & wi & Ra & Te \\
\hline
    unsat   (\Cref{ex:unsat-motivating})    &  2 &  1 & n & \vbes{132} &  TO &  ER &   - \\  
    vacuous (\Cref{ex:vacuous-motivating})  &  2 &  1 & y &  \vbes{56} &  TO &  ER &  ER \\  
    discharge-GF (\Cref{ex:GF-motivating})  &  2 &  1 & y &  \vbes{15} &  TO &  TO &  ER \\
\hline
    Box Limited~\cite{NeiderT16}            &  2 &  2 & y &   6 &   \best{1} &   \best{1} &  MO \\
    Box                                     &  2 &  2 & y &  34 &   3 &   \best{1} &  TO \\ 
    Diagonal                                &  2 &  1 & y &  30 &   \best{1} &   5 &  MO \\
    Evasion                                 &  4 &  2 & y &  76 &   3 &   \best{2} &  TO \\
    Follow                                  &  4 &  2 & y &  TO &  \vbes{16} &  TO &  TO \\
    Solitrary                               &  2 &  0 & y &  12 &   \best{1} &  \best{1}        &  ER \\
    Square-5x5                              &  2 &  2 & y & 143 &   \best{9} &  43 &  TO \\
    Elevator Simple 3 \cite{MaderbacherB22} &  1 &  0 & y &  19 &   2 &   \best{1} &  MO \\
    Elevator Simple 4                       &  1 &  0 & y &  34 &   3 &   \best{1} &  MO \\
    Elevator Simple 5                       &  1 &  0 & y &  55 &   \best{4} &   \best{4}        &  TO \\
    Elevator Simple 8                       &  1 &  0 & y & 161 &   \best{6} &  23 &  TO \\
    Elevator Simple 10                      &  1 &  0 & y & 271 &  \best{10} &  98 &  TO \\
    Elevator Signal 3                       &  2 &  1 & y &  MO &  MO &  \vbes{17} &  MO \\
    Elevator Signal 4                       &  2 &  1 & y &  MO &  MO & \vbes{111} &  MO \\
    Elevator Signal 5                       &  2 &  1 & y &  MO &  MO & \vbes{735} &  MO \\
\hline
    G-real                                  &  3 &  1 & y & 308 &  TO &   \best{3} &  TO \\ 
    G-unreal-1                              &  2 &  1 & n &  \best{31} & 786 &  TO &   - \\ 
    G-unreal-2                              &  2 &  1 & n &  \vbes{71} &  TO &  ER &   - \\ 
    G-unreal-3                              &  1 &  0 & n &  \vbes{42} &  TO &  ER &   - \\ 
    F-real                                  &  3 &  1 & y &  \vbes{64} &  TO &  ER &  ER \\ 
    F-unreal                                &  2 &  1 & n & \vbes{102} &  TO &  TO &   - \\ 
    F-G-contradiction-1                     &  1 &  0 & n &  \vbes{32} &  TO &  ER &   - \\ 
    F-G-contradiction-2                     &  2 &  1 & n & 135 &  TO &   \best{1} &   - \\  
\hline
\end{tabular}}  
\end{minipage}%
\begin{minipage}{.48\linewidth}
\centering
\scalebox{0.65}{
\begin{tabular}[t]{l|ccc|rr|rr}
    Name & $|\progvars|$ & $|\inputs|$ & R & mo~ & wi~ & Ra & Te \\
\hline
    GF-real                                 &  1 &  1 & y &   \vbes{2} &  TO &  ER &  ER \\ 
    GF-unreal                               &  1 &  0 & n &   \vbes{3} &  TO &  TO &   - \\ 
    GF-G-contradiction                      &  1 &  0 & n &   \vbes{5} &  TO &  ER &   - \\ 

\hline
    thermostat-F                            &  2 &  1 & y &  \vbes{70} &  TO &  TO &  MO \\
    thermostat-F-unreal                     &  2 &  1 & n & \vbes{130} &  TO &  TO &   - \\
    thermostat-GF                           &  2 &  1 & y & \vbes{237} &  TO &  TO &  MO \\
    thermostat-GF-unreal                    &  2 &  1 & n &  \vbes{85} &  TO &  TO &   - \\
    ordered-visits                          &  2 &  1 & y &         TO &  TO &  TO &  TO \\
    patrolling-alarm                        &  2 &  2 & y & \vbes{104} &  TO &  TO &  MO \\
    patrolling                              &  2 &  0 & y & \vbes{288} &  TO &  ER &  TO \\
    robot-to-target-charging                &  3 &  1 & y & \vbes{257} &  TO &  TO &  TO \\
    robot-to-target-charging-unreal         &  3 &  1 & n &  \vbes{27} &  TO &  TO &   - \\
    robot-to-target                         &  3 &  1 & y & \vbes{412} &  TO &  TO &  MO \\
    robot-to-target-unreal                  &  3 &  1 & n & \vbes{341} &  TO &  TO &   - \\
    unordered-visits                        &  2 &  0 & y & \vbes{265} &  TO &  ER &  TO \\
    helipad                                 &  3 &  6 & y &  \vbes{89} &  MO &  TO &  MO \\
    package-delivery                        &  5 &  2 & y &  \vbes{77} &  MO &  TO &  MO \\
    tasks                                   &  3 &  0 & y & \vbes{791} &  TO &  ER &  MO \\
    tasks-unreal                            &  3 &  0 & n & \vbes{223} &  TO &  ER &   - \\
\hline
    buffer-storage                          &  3 &  1 & y &  TO &  TO &  \vbes{45} &  MO \\
    helipad-contradict                      &  3 &  6 & n & 158 &  \best{13} &  TO &   - \\
    ordered-visits-choice                   &  2 &  0 & y &  TO &  TO &  ER &  TO \\
    precise-reachability                    &  2 &  0 & y &  TO &  TO &  ER &  ER \\
    storage-GF-64                           &  2 &  0 & y &  TO &  TO &  ER &  MO \\
    unordered-visits-charging               &  3 &  0 & y &  TO &  TO &  TO &  TO \\
                                            &    &    &   &     &     &     &     \\
\hline

\hline
\end{tabular}}  
\end{minipage}
\end{table}

We compare our approach of pruning the symbolic game with a monitor to the naive generation and direct solving workflow (which we also implement in \toolname).
In addition, we compare to the TSL-MT synthesis tools \texttt{Raboniel}\cite{MaderbacherB22} and \texttt{TeMoS}\cite{ChoiFPS22} which are both based on abstraction refinement and LTL synthesis using \texttt{Strix}\cite{strix}.
\Cref{tab:all} shows the results of our evaluation, starting with the examples from \Cref{sec:examples}, followed by  the four categories of benchmarks described above.

\paragraph{Benchmarks from the Literature}
The benchmarks from the literature can already be handled by existing tools.
We can see that constructing the product with the monitor creates an expectable overhead during the computation.
However, solving them is still possible with the monitor.

\paragraph{Basic Properties}
The monitor can easily dismiss the basic properties described in this benchmark set.
However, the other approaches usually cannot handle those properties, and if they can, this is because some very local reasoning is possible.
This shows that our monitors are indeed a useful enhancement that can lead to improved performance.

\paragraph{Scenarios}
As shown by the empirical results, the monitor product outperforms the other approaches on our scenario benchmarks.
The reasons is that the rules greatly simplify the reasoning in the periods in-between the decision making points in those scenarios and make some bad decisions obvious.
This shows that the monitor product indeed has its merits for more complex specifications.

\paragraph{Limitations}
By design of these benchmarks, our approach does not perform well on them.
However, our experiments show that these benchmarks are challenging and cannot be handled by other techniques either.

\section{Related Work}\label{sec:related}
\paragraph{Synthesis of Infinite-State Reactive Systems}
The work on synthesis of reactive systems from temporal logic specifications extended with richer data domains has focused on the logics TSL~\cite{FinkbeinerKPS19} and \tslmt~\cite{FinkbeinerHP22},  and the logic LTL$_\mathcal T$~\cite{RodriguezS23}.
The synthesis techniques in~\cite{FinkbeinerKPS19,ChoiFPS22,MaderbacherB22} are based on propositional abstraction and iterative refinement of the specification by introducing assumptions. 
The logic LTL$_\mathcal T$ considered in~\cite{RodriguezS23} restricts the atomic propositions to be literals over \emph{current-time-step} variables only. This enables the method in~\cite{RodriguezS23} to encode the theory specification into an equirealizable Boolean LTL formula. Thus,  LTL$_\mathcal T$ realizability is decidable under decidability assumptions for the underlying first-order theory. In contrast,  $\templogic$ allows the specification of relationships between current and next-state variables,  and the realizability and synthesis problems are undecidable.
All of the above approaches rely on a procedure for finite-state synthesis to solve the resulting LTL synthesis task. Our method,  on the other hand,  treats the $\templogic$ specification directly by constructing an infinite-state two-player game encoding the synthesis task.

Another line of work on synthesis of reactive systems with unbounded data domains focuses on solving infinite-state games that directly encode the synthesis problem.
Abstraction-based methods~\cite{WalkerR14,VechevYY13, GrumbergLLS07,HenzingerJM03, FinkbeinerMPSS22,AzzopardiPSS23} extend techniques such as abstract interpretation and counterexample-guided refinelment to two-player games. 
Symbolic game-solving techniques work directly with the infinite-state game. 
This class of techniques includes constraint-based approaches~\cite{FaellaP23, FarzanK18,KatisFGGBGW18} for special classes of wining conditions and symbolic fixpoint computation methods~\cite{SamuelDK21, SamuelDK23,HeimD24, SchmuckHDN24}. 

The approach we propose in this paper is agnostic to the game solver and can enhance different methods for solving symbolic games.

\paragraph{Automata Constructions and Satisfiability Checking for LTL}
The construction of alternating B\"uchi word automata (ABW) from LTL~\cite{Vardi94,Vardi95} is the first where the transition function is defined by following the LTL expansion laws.
The algorithm in~\cite{GerthPVW95} uses tableaux, based on the LTL expansion laws, to translate LTL to generalized nondeterministic B\"uchi automata (GNBA). Both ABW and GNBA can be translated to nondeterministic B\"uchi automata, which can subsequently be determinized~\cite{Safra88,Piterman07} to obtain a DPA. In this two-step construction of DPA from LTL, the semantic structure of the states of the automaton and their correspondence to the LTL formula is lost. This limitation has motivated a line of work~\cite{KretinskyE12,EsparzaKS20,  EsparzaKRS22}, which developed methods for direct translation of LTL to deterministic $\omega$-automata.
Similarly to the tableaux-based constructions,  this translation uses the LTL expansion laws. This enables semantics-based reductions, such as merging states representing equivalent formulas, and heuristics for synthesis~\cite{LuttenbergerMS20}. 
 
Tableaux techniques for checking LTL satisfiability were first studied by Wolper~\cite{Wolper85} and subsequently developed by~\cite{Schwendimann98,Reynolds16a, BertelloGMR16}.
Several different approaches to LTL satisfiability have been proposed.
Methods based on model checking~\cite{RozierV10} reduce the problem to the verification of the negated formula against a universal model.
Some recent techniques employ SAT solvers~\cite{LiZPVH13, LiP0VH14,GeattiGMV24}, taking advantage of the progress in SAT solving.
Recently, ~\cite{HermoLS23} proposed a tableaux method for checking the realizability of the safety fragment of LTL.

None of these techniques handle the non-Boolean variable domains that our method tackles.

\paragraph{Satisfiability Checking and Automata Constructions for First-Order Temporal Logics}

The decidability of the satisfiability problem for fragments of first-order linear and branching-time temporal logics has been extensively studied~\cite{HodkinsonWZ01,HodkinsonWZ02,HodkinsonKKWZ03}.
\cite{Demri06,  DemriD07} study LTL with constraints, 
which is a fragment of first-order LTL. In general, the satisfiability problem for these logics is undecidable.
The model-checking problem for an extension of LTL with atomic propositions over variables with possibly infinite domains has been considered in~\cite{GrumbergKS12}. \cite{FaranK18} introduces linear temporal logic with arithmetic (LTLA),  which extends LTL with linear arithmetic over the integers.   
\cite{FaranK20} studies the synthesis problem for specifications in the form of variable automata with arithmetic.

Symbolic automata~\cite{DAntoniV15,DAntoniV21} extend finite-state automata with support for potentially infinite alphabets represented by effective Boolean algebras.  
They were recently generalized to $\omega$-regular languages of infinite words in~\cite{VeanesBES23}, which introduces alternating  and nondeterministic B\"uchi automata modulo $\mathcal A$ (\emph{ABW}$_\mathcal A$ and \emph{NBW}$_\mathcal A$, respectively),  where  $\mathcal A$ is an effective Boolean algebra. Further, ~\cite{VeanesBES23} provides algorithms for the construction of \emph{ABW}$_\mathcal A$ from specifications in LTL modulo $\mathcal A$,  as well as in the combination of LTL modulo $\mathcal A$ with extended regular expressions modulo $\mathcal A$. 
The key idea is the use of symbolic transition terms and symbolic derivatives, which enable symbolic automata constructions integrating theory reasoning and rewriting. The translation of \emph{ABW}$_\mathcal A$ into equivalent
\emph{NBW}$_\mathcal A$ is also performed using symbolic derivatives,  resulting in a symbolic generalization of the classical alternation elimination algorithm~\cite{MiyanoH84}. 
Our approach, on the other hand, targets the construction of deterministic monitors, as determinism is necessary for building the game encoding the synthesis problem. Investigating the use of symbolic derivatives to construct a deterministic transition relation in our context is a possible avenue for future work, with the potential to improve the efficiency of the successor-state computation.
Finally,  we focus on $\templogic$ specifications,  while extensions of temporal logics with regular expressions have proven useful for practical adoption in verification~\cite{PSL-standard}. The logic \emph{RLTL}$_\mathcal A$
introduced in~\cite{VeanesBES23}, combines LTL modulo $\mathcal A$ with
extended regular expressions modulo $\mathcal A$. The combination is facilitated by symbolic transition terms,  which provide a common semantics.

\cite{GeattiGG24} considers finite-word semantics of first-order LTL and proposes an automaton model and a corresponding notion of regular first-order languages. They study the closure of these languages under common operations and establish conditions under which non-emptiness is semi-decidable. Our work,  on the other hand,  focuses on infinite-word semantics.

\paragraph{Monitoring LTL and First-Order LTL}
Monitors for LTL are also essential in the context of runtime verification.
Their construction typically uses the LTL expansion laws~\cite{HavelundR01,SenRA03}.
Recent work~\cite{HavelundP21,HavelundP18} has developed theory and tools for monitoring first-order temporal logic, focusing on safety.

None of these constructions targets applications in synthesis or can perform the temporal and first-order reasoning crucial for our transformation rules.

\section{Conclusion}\label{sec:conclusion}
A common approach to synthesizing infinite-state reactive programs out of temporal logic specifications is to turn the specification into a symbolic game, which is then solved. 
However, a shortcoming of existing approaches is that the construction of the symbolic game loses semantic information from the high-level specification, making the game hard to solve.
For example, while it is easy to see that the formula $x = 0 \land (\LTLglobally x' > x) \land (\LTLeventually x < 0)$ is unrealizable, it is less evident from the constructed symbolic game.

In this paper, we introduce monitors, which provide the basis of a flexible framework for systematically adding some of the missing semantic information to the game. 
The monitors are constructed from the temporal logic formula via rules combining first-order and temporal reasoning. 
We use these monitors to prune the synthesis game's state space and winning condition by constructing a product game from the game and the monitor.
Our method is general because it is independent of the game construction and game solving.
It is extensible, allowing for the easy addition of new rules or the complete swap out of the monitor construction. 
We demonstrate the effectiveness of our approach empirically.

We believe our method and monitor are general enough that, for future work, they could be extended to more expressive formalisms, like regular expressions modulo theories. 
Furthermore, while the independence of the game construction is an advantage of our approach, a deeper integration of our method and the game construction -- like done in the construction of symbolic automata -- might yield even better-pruned games. 
Another possible extension is to find a way to extend the monitor construction with some form of inexpensive strategic reasoning.

\begin{acks}
    We thank the anonymous reviewers for their helpful suggestions.
\end{acks}

\bibliographystyle{ACM-Reference-Format}
\bibliography{references}


\begin{thebibliography}{69}

%%% the Web address.

\ifx \showCODEN    \undefined \def \showCODEN     #1{\unskip}     \fi
\ifx \showDOI      \undefined \def \showDOI       #1{#1}\fi
\ifx \showISBNx    \undefined \def \showISBNx     #1{\unskip}     \fi
\ifx \showISBNxiii \undefined \def \showISBNxiii  #1{\unskip}     \fi
\ifx \showISSN     \undefined \def \showISSN      #1{\unskip}     \fi
\ifx \showLCCN     \undefined \def \showLCCN      #1{\unskip}     \fi
\ifx \shownote     \undefined \def \shownote      #1{#1}          \fi
\ifx \showarticletitle \undefined \def \showarticletitle #1{#1}   \fi
\ifx \showURL      \undefined \def \showURL       {\relax}        \fi
% invisible to TeX
\providecommand\bibfield[2]{#2}
\providecommand\bibinfo[2]{#2}
\providecommand\natexlab[1]{#1}
\providecommand\showeprint[2][]{arXiv:#2}

\bibitem[Alur et~al\mbox{.}(2013)]%
        {AlurBJMRSSSTU13}
\bibfield{author}{\bibinfo{person}{Rajeev Alur}, \bibinfo{person}{Rastislav
  Bod{\'{\i}}k}, \bibinfo{person}{Garvit Juniwal}, \bibinfo{person}{Milo M.~K.
  Martin}, \bibinfo{person}{Mukund Raghothaman}, \bibinfo{person}{Sanjit~A.
  Seshia}, \bibinfo{person}{Rishabh Singh}, \bibinfo{person}{Armando
  Solar{-}Lezama}, \bibinfo{person}{Emina Torlak}, {and}
  \bibinfo{person}{Abhishek Udupa}.} \bibinfo{year}{2013}\natexlab{}.
\newblock \showarticletitle{Syntax-guided synthesis}. In
  \bibinfo{booktitle}{\emph{Formal Methods in Computer-Aided Design, {FMCAD}
  2013, Portland, OR, USA, October 20-23, 2013}}. \bibinfo{publisher}{{IEEE}},
  \bibinfo{pages}{1--8}.
\newblock
\urldef\tempurl%
\url{https://ieeexplore.ieee.org/document/6679385/}
\showURL{%
\tempurl}


\bibitem[Azzopardi et~al\mbox{.}(2023)]%
        {AzzopardiPSS23}
\bibfield{author}{\bibinfo{person}{Shaun Azzopardi}, \bibinfo{person}{Nir
  Piterman}, \bibinfo{person}{Gerardo Schneider}, {and}
  \bibinfo{person}{Luca~Di Stefano}.} \bibinfo{year}{2023}\natexlab{}.
\newblock \showarticletitle{{LTL} Synthesis on Infinite-State Arenas defined by
  Programs}.
\newblock \bibinfo{journal}{\emph{CoRR}}  \bibinfo{volume}{abs/2307.09776}
  (\bibinfo{year}{2023}).
\newblock
\urldef\tempurl%
\url{https://doi.org/10.48550/ARXIV.2307.09776}
\showDOI{\tempurl}
\showeprint[arXiv]{2307.09776}


\bibitem[Baier and Katoen(2008)]%
        {BaierK08}
\bibfield{author}{\bibinfo{person}{Christel Baier} {and}
  \bibinfo{person}{Joost{-}Pieter Katoen}.} \bibinfo{year}{2008}\natexlab{}.
\newblock \bibinfo{booktitle}{\emph{Principles of model checking}}.
\newblock \bibinfo{publisher}{{MIT} Press}.
\newblock
\showISBNx{978-0-262-02649-9}


\bibitem[Bertello et~al\mbox{.}(2016)]%
        {BertelloGMR16}
\bibfield{author}{\bibinfo{person}{Matteo Bertello}, \bibinfo{person}{Nicola
  Gigante}, \bibinfo{person}{Angelo Montanari}, {and} \bibinfo{person}{Mark
  Reynolds}.} \bibinfo{year}{2016}\natexlab{}.
\newblock \showarticletitle{Leviathan: {A} New {LTL} Satisfiability Checking
  Tool Based on a One-Pass Tree-Shaped Tableau}. In
  \bibinfo{booktitle}{\emph{Proceedings of the Twenty-Fifth International Joint
  Conference on Artificial Intelligence, {IJCAI} 2016, New York, NY, USA, 9-15
  July 2016}}, \bibfield{editor}{\bibinfo{person}{Subbarao Kambhampati}} (Ed.).
  \bibinfo{publisher}{{IJCAI/AAAI} Press}, \bibinfo{pages}{950--956}.
\newblock
\urldef\tempurl%
\url{http://www.ijcai.org/Abstract/16/139}
\showURL{%
\tempurl}


\bibitem[Bradley and Manna(2007)]%
        {BradleyM07}
\bibfield{author}{\bibinfo{person}{Aaron~R. Bradley} {and}
  \bibinfo{person}{Zohar Manna}.} \bibinfo{year}{2007}\natexlab{}.
\newblock \bibinfo{booktitle}{\emph{The calculus of computation - decision
  procedures with applications to verification}}.
\newblock \bibinfo{publisher}{Springer}.
\newblock
\urldef\tempurl%
\url{https://doi.org/10.1007/978-3-540-74113-8}
\showDOI{\tempurl}


\bibitem[Choi et~al\mbox{.}(2022)]%
        {ChoiFPS22}
\bibfield{author}{\bibinfo{person}{Wonhyuk Choi}, \bibinfo{person}{Bernd
  Finkbeiner}, \bibinfo{person}{Ruzica Piskac}, {and} \bibinfo{person}{Mark
  Santolucito}.} \bibinfo{year}{2022}\natexlab{}.
\newblock \showarticletitle{Can reactive synthesis and syntax-guided synthesis
  be friends?}. In \bibinfo{booktitle}{\emph{{PLDI} '22: 43rd {ACM} {SIGPLAN}
  International Conference on Programming Language Design and Implementation,
  San Diego, CA, USA, June 13 - 17, 2022}},
  \bibfield{editor}{\bibinfo{person}{Ranjit Jhala} {and} \bibinfo{person}{Isil
  Dillig}} (Eds.). \bibinfo{publisher}{{ACM}}, \bibinfo{pages}{229--243}.
\newblock
\urldef\tempurl%
\url{https://doi.org/10.1145/3519939.3523429}
\showDOI{\tempurl}


\bibitem[D'Antoni and Veanes(2015)]%
        {DAntoniV15}
\bibfield{author}{\bibinfo{person}{Loris D'Antoni} {and}
  \bibinfo{person}{Margus Veanes}.} \bibinfo{year}{2015}\natexlab{}.
\newblock \showarticletitle{Extended symbolic finite automata and transducers}.
\newblock \bibinfo{journal}{\emph{Formal Methods Syst. Des.}}
  \bibinfo{volume}{47}, \bibinfo{number}{1} (\bibinfo{year}{2015}),
  \bibinfo{pages}{93--119}.
\newblock
\urldef\tempurl%
\url{https://doi.org/10.1007/S10703-015-0233-4}
\showDOI{\tempurl}


\bibitem[D'Antoni and Veanes(2021)]%
        {DAntoniV21}
\bibfield{author}{\bibinfo{person}{Loris D'Antoni} {and}
  \bibinfo{person}{Margus Veanes}.} \bibinfo{year}{2021}\natexlab{}.
\newblock \showarticletitle{Automata modulo theories}.
\newblock \bibinfo{journal}{\emph{Commun. {ACM}}} \bibinfo{volume}{64},
  \bibinfo{number}{5} (\bibinfo{year}{2021}), \bibinfo{pages}{86--95}.
\newblock
\urldef\tempurl%
\url{https://doi.org/10.1145/3419404}
\showDOI{\tempurl}


\bibitem[de~Moura and Bj{\o}rner(2008)]%
        {Z3}
\bibfield{author}{\bibinfo{person}{Leonardo~Mendon{\c{c}}a de Moura} {and}
  \bibinfo{person}{Nikolaj~S. Bj{\o}rner}.} \bibinfo{year}{2008}\natexlab{}.
\newblock \showarticletitle{{Z3:} An Efficient {SMT} Solver}. In
  \bibinfo{booktitle}{\emph{Tools and Algorithms for the Construction and
  Analysis of Systems, 14th International Conference, {TACAS} 2008, Held as
  Part of the Joint European Conferences on Theory and Practice of Software,
  {ETAPS} 2008, Budapest, Hungary, March 29-April 6, 2008. Proceedings}}
  \emph{(\bibinfo{series}{Lecture Notes in Computer Science},
  Vol.~\bibinfo{volume}{4963})}, \bibfield{editor}{\bibinfo{person}{C.~R.
  Ramakrishnan} {and} \bibinfo{person}{Jakob Rehof}} (Eds.).
  \bibinfo{publisher}{Springer}, \bibinfo{pages}{337--340}.
\newblock
\urldef\tempurl%
\url{https://doi.org/10.1007/978-3-540-78800-3\_24}
\showDOI{\tempurl}


\bibitem[Demri(2006)]%
        {Demri06}
\bibfield{author}{\bibinfo{person}{St{\'{e}}phane Demri}.}
  \bibinfo{year}{2006}\natexlab{}.
\newblock \showarticletitle{Linear-time temporal logics with Presburger
  constraints: an overview}.
\newblock \bibinfo{journal}{\emph{J. Appl. Non Class. Logics}}
  \bibinfo{volume}{16}, \bibinfo{number}{3-4} (\bibinfo{year}{2006}),
  \bibinfo{pages}{311--348}.
\newblock
\urldef\tempurl%
\url{https://doi.org/10.3166/JANCL.16.311-347}
\showDOI{\tempurl}


\bibitem[Demri and D'Souza(2007)]%
        {DemriD07}
\bibfield{author}{\bibinfo{person}{St{\'{e}}phane Demri} {and}
  \bibinfo{person}{Deepak D'Souza}.} \bibinfo{year}{2007}\natexlab{}.
\newblock \showarticletitle{An automata-theoretic approach to constraint
  {LTL}}.
\newblock \bibinfo{journal}{\emph{Inf. Comput.}} \bibinfo{volume}{205},
  \bibinfo{number}{3} (\bibinfo{year}{2007}), \bibinfo{pages}{380--415}.
\newblock
\urldef\tempurl%
\url{https://doi.org/10.1016/J.IC.2006.09.006}
\showDOI{\tempurl}


\bibitem[Duret{-}Lutz et~al\mbox{.}(2016)]%
        {Duret-LutzLFMRX16}
\bibfield{author}{\bibinfo{person}{Alexandre Duret{-}Lutz},
  \bibinfo{person}{Alexandre Lewkowicz}, \bibinfo{person}{Amaury Fauchille},
  \bibinfo{person}{Thibaud Michaud}, \bibinfo{person}{Etienne Renault}, {and}
  \bibinfo{person}{Laurent Xu}.} \bibinfo{year}{2016}\natexlab{}.
\newblock \showarticletitle{Spot 2.0 - {A} Framework for {LTL} and
  {\textbackslash}omega -Automata Manipulation}. In
  \bibinfo{booktitle}{\emph{Automated Technology for Verification and Analysis
  - 14th International Symposium, {ATVA} 2016, Chiba, Japan, October 17-20,
  2016, Proceedings}} \emph{(\bibinfo{series}{Lecture Notes in Computer
  Science}, Vol.~\bibinfo{volume}{9938})},
  \bibfield{editor}{\bibinfo{person}{Cyrille Artho}, \bibinfo{person}{Axel
  Legay}, {and} \bibinfo{person}{Doron Peled}} (Eds.).
  \bibinfo{pages}{122--129}.
\newblock
\urldef\tempurl%
\url{https://doi.org/10.1007/978-3-319-46520-3\_8}
\showDOI{\tempurl}


\bibitem[Duret{-}Lutz et~al\mbox{.}(2022)]%
        {Duret-LutzRCRAS22}
\bibfield{author}{\bibinfo{person}{Alexandre Duret{-}Lutz},
  \bibinfo{person}{Etienne Renault}, \bibinfo{person}{Maximilien Colange},
  \bibinfo{person}{Florian Renkin}, \bibinfo{person}{Alexandre~Gbaguidi Aisse},
  \bibinfo{person}{Philipp Schlehuber{-}Caissier}, \bibinfo{person}{Thomas
  Medioni}, \bibinfo{person}{Antoine Martin},
  \bibinfo{person}{J{\'{e}}r{\^{o}}me Dubois}, \bibinfo{person}{Cl{\'{e}}ment
  Gillard}, {and} \bibinfo{person}{Henrich Lauko}.}
  \bibinfo{year}{2022}\natexlab{}.
\newblock \showarticletitle{From Spot 2.0 to Spot 2.10: What's New?}. In
  \bibinfo{booktitle}{\emph{Computer Aided Verification - 34th International
  Conference, {CAV} 2022, Haifa, Israel, August 7-10, 2022, Proceedings, Part
  {II}}} \emph{(\bibinfo{series}{Lecture Notes in Computer Science},
  Vol.~\bibinfo{volume}{13372})}, \bibfield{editor}{\bibinfo{person}{Sharon
  Shoham} {and} \bibinfo{person}{Yakir Vizel}} (Eds.).
  \bibinfo{publisher}{Springer}, \bibinfo{pages}{174--187}.
\newblock
\urldef\tempurl%
\url{https://doi.org/10.1007/978-3-031-13188-2\_9}
\showDOI{\tempurl}


\bibitem[Eisner and Fisman(2006)]%
        {EisnerF06}
\bibfield{author}{\bibinfo{person}{Cindy Eisner} {and} \bibinfo{person}{Dana
  Fisman}.} \bibinfo{year}{2006}\natexlab{}.
\newblock \bibinfo{booktitle}{\emph{A Practical Introduction to {PSL}}}.
\newblock \bibinfo{publisher}{Springer}.
\newblock
\showISBNx{978-0-387-35313-5}
\urldef\tempurl%
\url{https://doi.org/10.1007/978-0-387-36123-9}
\showDOI{\tempurl}


\bibitem[Esparza et~al\mbox{.}(2022)]%
        {EsparzaKRS22}
\bibfield{author}{\bibinfo{person}{Javier Esparza}, \bibinfo{person}{Jan
  Kret{\'{\i}}nsk{\'{y}}}, \bibinfo{person}{Jean{-}Fran{\c{c}}ois Raskin},
  {and} \bibinfo{person}{Salomon Sickert}.} \bibinfo{year}{2022}\natexlab{}.
\newblock \showarticletitle{From linear temporal logic and limit-deterministic
  B{\"{u}}chi automata to deterministic parity automata}.
\newblock \bibinfo{journal}{\emph{Int. J. Softw. Tools Technol. Transf.}}
  \bibinfo{volume}{24}, \bibinfo{number}{4} (\bibinfo{year}{2022}),
  \bibinfo{pages}{635--659}.
\newblock
\urldef\tempurl%
\url{https://doi.org/10.1007/S10009-022-00663-1}
\showDOI{\tempurl}


\bibitem[Esparza et~al\mbox{.}(2020)]%
        {EsparzaKS20}
\bibfield{author}{\bibinfo{person}{Javier Esparza}, \bibinfo{person}{Jan
  Kret{\'{\i}}nsk{\'{y}}}, {and} \bibinfo{person}{Salomon Sickert}.}
  \bibinfo{year}{2020}\natexlab{}.
\newblock \showarticletitle{A Unified Translation of Linear Temporal Logic to
  {\(\omega\)}-Automata}.
\newblock \bibinfo{journal}{\emph{J. {ACM}}} \bibinfo{volume}{67},
  \bibinfo{number}{6} (\bibinfo{year}{2020}), \bibinfo{pages}{33:1--33:61}.
\newblock
\urldef\tempurl%
\url{https://doi.org/10.1145/3417995}
\showDOI{\tempurl}


\bibitem[Faella and Parlato(2023)]%
        {FaellaP23}
\bibfield{author}{\bibinfo{person}{Marco Faella} {and} \bibinfo{person}{Gennaro
  Parlato}.} \bibinfo{year}{2023}\natexlab{}.
\newblock \showarticletitle{Reachability Games modulo Theories with a Bounded
  Safety Player}. In \bibinfo{booktitle}{\emph{Proceedings of the
  Thirty-Seventh AAAI Conference on Artificial Intelligence and Thirty-Fifth
  Conference on Innovative Applications of Artificial Intelligence and
  Thirteenth Symposium on Educational Advances in Artificial Intelligence}}
  \emph{(\bibinfo{series}{AAAI'23/IAAI'23/EAAI'23})}. \bibinfo{publisher}{AAAI
  Press}, Article \bibinfo{articleno}{710}, \bibinfo{numpages}{8}~pages.
\newblock
\showISBNx{978-1-57735-880-0}
\urldef\tempurl%
\url{https://doi.org/10.1609/aaai.v37i5.25779}
\showDOI{\tempurl}


\bibitem[Faran and Kupferman(2018)]%
        {FaranK18}
\bibfield{author}{\bibinfo{person}{Rachel Faran} {and} \bibinfo{person}{Orna
  Kupferman}.} \bibinfo{year}{2018}\natexlab{}.
\newblock \showarticletitle{{LTL} with Arithmetic and its Applications in
  Reasoning about Hierarchical Systems}. In
  \bibinfo{booktitle}{\emph{{LPAR-22.} 22nd International Conference on Logic
  for Programming, Artificial Intelligence and Reasoning, Awassa, Ethiopia,
  16-21 November 2018}} \emph{(\bibinfo{series}{EPiC Series in Computing},
  Vol.~\bibinfo{volume}{57})}, \bibfield{editor}{\bibinfo{person}{Gilles
  Barthe}, \bibinfo{person}{Geoff Sutcliffe}, {and} \bibinfo{person}{Margus
  Veanes}} (Eds.). \bibinfo{publisher}{EasyChair}, \bibinfo{pages}{343--362}.
\newblock
\urldef\tempurl%
\url{https://doi.org/10.29007/WPG3}
\showDOI{\tempurl}


\bibitem[Faran and Kupferman(2020)]%
        {FaranK20}
\bibfield{author}{\bibinfo{person}{Rachel Faran} {and} \bibinfo{person}{Orna
  Kupferman}.} \bibinfo{year}{2020}\natexlab{}.
\newblock \showarticletitle{On Synthesis of Specifications with Arithmetic}. In
  \bibinfo{booktitle}{\emph{{SOFSEM} 2020: Theory and Practice of Computer
  Science - 46th International Conference on Current Trends in Theory and
  Practice of Informatics, {SOFSEM} 2020, Limassol, Cyprus, January 20-24,
  2020, Proceedings}} \emph{(\bibinfo{series}{Lecture Notes in Computer
  Science}, Vol.~\bibinfo{volume}{12011})},
  \bibfield{editor}{\bibinfo{person}{Alexander Chatzigeorgiou},
  \bibinfo{person}{Riccardo Dondi}, \bibinfo{person}{Herodotos Herodotou},
  \bibinfo{person}{Christos~A. Kapoutsis}, \bibinfo{person}{Yannis
  Manolopoulos}, \bibinfo{person}{George~A. Papadopoulos}, {and}
  \bibinfo{person}{Florian Sikora}} (Eds.). \bibinfo{publisher}{Springer},
  \bibinfo{pages}{161--173}.
\newblock
\urldef\tempurl%
\url{https://doi.org/10.1007/978-3-030-38919-2\_14}
\showDOI{\tempurl}


\bibitem[Farzan and Kincaid(2018)]%
        {FarzanK18}
\bibfield{author}{\bibinfo{person}{Azadeh Farzan} {and}
  \bibinfo{person}{Zachary Kincaid}.} \bibinfo{year}{2018}\natexlab{}.
\newblock \showarticletitle{Strategy synthesis for linear arithmetic games}.
\newblock \bibinfo{journal}{\emph{Proc. {ACM} Program. Lang.}}
  \bibinfo{volume}{2}, \bibinfo{number}{{POPL}} (\bibinfo{year}{2018}),
  \bibinfo{pages}{61:1--61:30}.
\newblock
\urldef\tempurl%
\url{https://doi.org/10.1145/3158149}
\showDOI{\tempurl}


\bibitem[Finkbeiner et~al\mbox{.}(2022a)]%
        {FinkbeinerHP22}
\bibfield{author}{\bibinfo{person}{Bernd Finkbeiner}, \bibinfo{person}{Philippe
  Heim}, {and} \bibinfo{person}{Noemi Passing}.}
  \bibinfo{year}{2022}\natexlab{a}.
\newblock \showarticletitle{Temporal Stream Logic modulo Theories}. In
  \bibinfo{booktitle}{\emph{Foundations of Software Science and Computation
  Structures - 25th International Conference, {FOSSACS} 2022, Held as Part of
  the European Joint Conferences on Theory and Practice of Software, {ETAPS}
  2022, Munich, Germany, April 2-7, 2022, Proceedings}}
  \emph{(\bibinfo{series}{Lecture Notes in Computer Science},
  Vol.~\bibinfo{volume}{13242})}, \bibfield{editor}{\bibinfo{person}{Patricia
  Bouyer} {and} \bibinfo{person}{Lutz Schr{\"{o}}der}} (Eds.).
  \bibinfo{publisher}{Springer}, \bibinfo{pages}{325--346}.
\newblock
\urldef\tempurl%
\url{https://doi.org/10.1007/978-3-030-99253-8\_17}
\showDOI{\tempurl}


\bibitem[Finkbeiner et~al\mbox{.}(2019)]%
        {FinkbeinerKPS19}
\bibfield{author}{\bibinfo{person}{Bernd Finkbeiner}, \bibinfo{person}{Felix
  Klein}, \bibinfo{person}{Ruzica Piskac}, {and} \bibinfo{person}{Mark
  Santolucito}.} \bibinfo{year}{2019}\natexlab{}.
\newblock \showarticletitle{Temporal Stream Logic: Synthesis Beyond the Bools}.
  In \bibinfo{booktitle}{\emph{Computer Aided Verification - 31st International
  Conference, {CAV} 2019, New York City, NY, USA, July 15-18, 2019,
  Proceedings, Part {I}}} \emph{(\bibinfo{series}{Lecture Notes in Computer
  Science}, Vol.~\bibinfo{volume}{11561})},
  \bibfield{editor}{\bibinfo{person}{Isil Dillig} {and} \bibinfo{person}{Serdar
  Tasiran}} (Eds.). \bibinfo{publisher}{Springer}, \bibinfo{pages}{609--629}.
\newblock
\urldef\tempurl%
\url{https://doi.org/10.1007/978-3-030-25540-4\_35}
\showDOI{\tempurl}


\bibitem[Finkbeiner et~al\mbox{.}(2022b)]%
        {FinkbeinerMPSS22}
\bibfield{author}{\bibinfo{person}{Bernd Finkbeiner}, \bibinfo{person}{Kaushik
  Mallik}, \bibinfo{person}{Noemi Passing}, \bibinfo{person}{Malte
  Schledjewski}, {and} \bibinfo{person}{Anne{-}Kathrin Schmuck}.}
  \bibinfo{year}{2022}\natexlab{b}.
\newblock \showarticletitle{BOCoSy: Small but Powerful Symbolic Output-Feedback
  Control}. In \bibinfo{booktitle}{\emph{{HSCC} '22: 25th {ACM} International
  Conference on Hybrid Systems: Computation and Control, Milan, Italy, May 4 -
  6, 2022}}, \bibfield{editor}{\bibinfo{person}{Ezio Bartocci} {and}
  \bibinfo{person}{Sylvie Putot}} (Eds.). \bibinfo{publisher}{{ACM}},
  \bibinfo{pages}{24:1--24:11}.
\newblock
\urldef\tempurl%
\url{https://doi.org/10.1145/3501710.3519535}
\showDOI{\tempurl}


\bibitem[Geatti et~al\mbox{.}(2024a)]%
        {GeattiGG24}
\bibfield{author}{\bibinfo{person}{Luca Geatti}, \bibinfo{person}{Alessandro
  Gianola}, {and} \bibinfo{person}{Nicola Gigante}.}
  \bibinfo{year}{2024}\natexlab{a}.
\newblock \showarticletitle{A General Automata Model for First-Order Temporal
  Logics (Extended Version)}.
\newblock \bibinfo{journal}{\emph{CoRR}}  \bibinfo{volume}{abs/2405.20057}
  (\bibinfo{year}{2024}).
\newblock
\urldef\tempurl%
\url{https://doi.org/10.48550/ARXIV.2405.20057}
\showDOI{\tempurl}
\showeprint[arXiv]{2405.20057}


\bibitem[Geatti et~al\mbox{.}(2024b)]%
        {GeattiGMV24}
\bibfield{author}{\bibinfo{person}{Luca Geatti}, \bibinfo{person}{Nicola
  Gigante}, \bibinfo{person}{Angelo Montanari}, {and} \bibinfo{person}{Gabriele
  Venturato}.} \bibinfo{year}{2024}\natexlab{b}.
\newblock \showarticletitle{{SAT} Meets Tableaux for Linear Temporal Logic
  Satisfiability}.
\newblock \bibinfo{journal}{\emph{J. Autom. Reason.}} \bibinfo{volume}{68},
  \bibinfo{number}{2} (\bibinfo{year}{2024}), \bibinfo{pages}{6}.
\newblock
\urldef\tempurl%
\url{https://doi.org/10.1007/S10817-023-09691-1}
\showDOI{\tempurl}


\bibitem[Gerth et~al\mbox{.}(1995)]%
        {GerthPVW95}
\bibfield{author}{\bibinfo{person}{Rob Gerth}, \bibinfo{person}{Doron~A.
  Peled}, \bibinfo{person}{Moshe~Y. Vardi}, {and} \bibinfo{person}{Pierre
  Wolper}.} \bibinfo{year}{1995}\natexlab{}.
\newblock \showarticletitle{Simple on-the-fly automatic verification of linear
  temporal logic}. In \bibinfo{booktitle}{\emph{Protocol Specification, Testing
  and Verification XV, Proceedings of the Fifteenth {IFIP} {WG6.1}
  International Symposium on Protocol Specification, Testing and Verification,
  Warsaw, Poland, June 1995}} \emph{(\bibinfo{series}{{IFIP} Conference
  Proceedings}, Vol.~\bibinfo{volume}{38})},
  \bibfield{editor}{\bibinfo{person}{Piotr Dembinski} {and}
  \bibinfo{person}{Marek Sredniawa}} (Eds.). \bibinfo{publisher}{Chapman {\&}
  Hall}, \bibinfo{pages}{3--18}.
\newblock


\bibitem[Grumberg et~al\mbox{.}(2012)]%
        {GrumbergKS12}
\bibfield{author}{\bibinfo{person}{Orna Grumberg}, \bibinfo{person}{Orna
  Kupferman}, {and} \bibinfo{person}{Sarai Sheinvald}.}
  \bibinfo{year}{2012}\natexlab{}.
\newblock \showarticletitle{Model Checking Systems and Specifications with
  Parameterized Atomic Propositions}. In \bibinfo{booktitle}{\emph{Automated
  Technology for Verification and Analysis - 10th International Symposium,
  {ATVA} 2012, Thiruvananthapuram, India, October 3-6, 2012. Proceedings}}
  \emph{(\bibinfo{series}{Lecture Notes in Computer Science},
  Vol.~\bibinfo{volume}{7561})}, \bibfield{editor}{\bibinfo{person}{Supratik
  Chakraborty} {and} \bibinfo{person}{Madhavan Mukund}} (Eds.).
  \bibinfo{publisher}{Springer}, \bibinfo{pages}{122--136}.
\newblock
\urldef\tempurl%
\url{https://doi.org/10.1007/978-3-642-33386-6\_11}
\showDOI{\tempurl}


\bibitem[Grumberg et~al\mbox{.}(2007)]%
        {GrumbergLLS07}
\bibfield{author}{\bibinfo{person}{Orna Grumberg}, \bibinfo{person}{Martin
  Lange}, \bibinfo{person}{Martin Leucker}, {and} \bibinfo{person}{Sharon
  Shoham}.} \bibinfo{year}{2007}\natexlab{}.
\newblock \showarticletitle{When not losing is better than winning: Abstraction
  and refinement for the full mu-calculus}.
\newblock \bibinfo{journal}{\emph{Inf. Comput.}} \bibinfo{volume}{205},
  \bibinfo{number}{8} (\bibinfo{year}{2007}), \bibinfo{pages}{1130--1148}.
\newblock
\urldef\tempurl%
\url{https://doi.org/10.1016/j.ic.2006.10.009}
\showDOI{\tempurl}


\bibitem[Gu et~al\mbox{.}(2023)]%
        {GuTU23}
\bibfield{author}{\bibinfo{person}{Yu Gu}, \bibinfo{person}{Takeshi Tsukada},
  {and} \bibinfo{person}{Hiroshi Unno}.} \bibinfo{year}{2023}\natexlab{}.
\newblock \showarticletitle{Optimal {CHC} Solving via Termination Proofs}.
\newblock \bibinfo{journal}{\emph{Proc. {ACM} Program. Lang.}}
  \bibinfo{volume}{7}, \bibinfo{number}{{POPL}} (\bibinfo{year}{2023}),
  \bibinfo{pages}{604--631}.
\newblock
\urldef\tempurl%
\url{https://doi.org/10.1145/3571214}
\showDOI{\tempurl}


\bibitem[Havelund and Peled(2018)]%
        {HavelundP18}
\bibfield{author}{\bibinfo{person}{Klaus Havelund} {and} \bibinfo{person}{Doron
  Peled}.} \bibinfo{year}{2018}\natexlab{}.
\newblock \showarticletitle{Runtime Verification: From Propositional to
  First-Order Temporal Logic}. In \bibinfo{booktitle}{\emph{Runtime
  Verification - 18th International Conference, {RV} 2018, Limassol, Cyprus,
  November 10-13, 2018, Proceedings}} \emph{(\bibinfo{series}{Lecture Notes in
  Computer Science}, Vol.~\bibinfo{volume}{11237})},
  \bibfield{editor}{\bibinfo{person}{Christian Colombo} {and}
  \bibinfo{person}{Martin Leucker}} (Eds.). \bibinfo{publisher}{Springer},
  \bibinfo{pages}{90--112}.
\newblock
\urldef\tempurl%
\url{https://doi.org/10.1007/978-3-030-03769-7\_7}
\showDOI{\tempurl}


\bibitem[Havelund and Peled(2021)]%
        {HavelundP21}
\bibfield{author}{\bibinfo{person}{Klaus Havelund} {and} \bibinfo{person}{Doron
  Peled}.} \bibinfo{year}{2021}\natexlab{}.
\newblock \showarticletitle{An extension of first-order {LTL} with rules with
  application to runtime verification}.
\newblock \bibinfo{journal}{\emph{Int. J. Softw. Tools Technol. Transf.}}
  \bibinfo{volume}{23}, \bibinfo{number}{4} (\bibinfo{year}{2021}),
  \bibinfo{pages}{547--563}.
\newblock
\urldef\tempurl%
\url{https://doi.org/10.1007/S10009-021-00626-Y}
\showDOI{\tempurl}


\bibitem[Havelund and Rosu(2001)]%
        {HavelundR01}
\bibfield{author}{\bibinfo{person}{Klaus Havelund} {and}
  \bibinfo{person}{Grigore Rosu}.} \bibinfo{year}{2001}\natexlab{}.
\newblock \showarticletitle{Monitoring Programs Using Rewriting}. In
  \bibinfo{booktitle}{\emph{16th {IEEE} International Conference on Automated
  Software Engineering {(ASE} 2001), 26-29 November 2001, Coronado Island, San
  Diego, CA, {USA}}}. \bibinfo{publisher}{{IEEE} Computer Society},
  \bibinfo{pages}{135--143}.
\newblock
\urldef\tempurl%
\url{https://doi.org/10.1109/ASE.2001.989799}
\showDOI{\tempurl}


\bibitem[Heim and Dimitrova(2024)]%
        {HeimD24}
\bibfield{author}{\bibinfo{person}{Philippe Heim} {and} \bibinfo{person}{Rayna
  Dimitrova}.} \bibinfo{year}{2024}\natexlab{}.
\newblock \showarticletitle{Solving Infinite-State Games via Acceleration}.
\newblock \bibinfo{journal}{\emph{Proc. ACM Program. Lang.}}
  \bibinfo{volume}{8}, \bibinfo{number}{POPL}, Article \bibinfo{articleno}{57}
  (\bibinfo{date}{jan} \bibinfo{year}{2024}), \bibinfo{numpages}{31}~pages.
\newblock
\urldef\tempurl%
\url{https://doi.org/10.1145/3632899}
\showDOI{\tempurl}


\bibitem[Henzinger et~al\mbox{.}(2003)]%
        {HenzingerJM03}
\bibfield{author}{\bibinfo{person}{Thomas~A. Henzinger},
  \bibinfo{person}{Ranjit Jhala}, {and} \bibinfo{person}{Rupak Majumdar}.}
  \bibinfo{year}{2003}\natexlab{}.
\newblock \showarticletitle{Counterexample-Guided Control}. In
  \bibinfo{booktitle}{\emph{Automata, Languages and Programming, 30th
  International Colloquium, {ICALP} 2003, Eindhoven, The Netherlands, June 30 -
  July 4, 2003. Proceedings}} \emph{(\bibinfo{series}{Lecture Notes in Computer
  Science}, Vol.~\bibinfo{volume}{2719})},
  \bibfield{editor}{\bibinfo{person}{Jos C.~M. Baeten},
  \bibinfo{person}{Jan~Karel Lenstra}, \bibinfo{person}{Joachim Parrow}, {and}
  \bibinfo{person}{Gerhard~J. Woeginger}} (Eds.).
  \bibinfo{publisher}{Springer}, \bibinfo{pages}{886--902}.
\newblock
\urldef\tempurl%
\url{https://doi.org/10.1007/3-540-45061-0\_69}
\showDOI{\tempurl}


\bibitem[Hermo et~al\mbox{.}(2023)]%
        {HermoLS23}
\bibfield{author}{\bibinfo{person}{Montserrat Hermo}, \bibinfo{person}{Paqui
  Lucio}, {and} \bibinfo{person}{C{\'{e}}sar S{\'{a}}nchez}.}
  \bibinfo{year}{2023}\natexlab{}.
\newblock \showarticletitle{Tableaux for Realizability of Safety
  Specifications}. In \bibinfo{booktitle}{\emph{Formal Methods - 25th
  International Symposium, {FM} 2023, L{\"{u}}beck, Germany, March 6-10, 2023,
  Proceedings}} \emph{(\bibinfo{series}{Lecture Notes in Computer Science},
  Vol.~\bibinfo{volume}{14000})}, \bibfield{editor}{\bibinfo{person}{Marsha
  Chechik}, \bibinfo{person}{Joost{-}Pieter Katoen}, {and}
  \bibinfo{person}{Martin Leucker}} (Eds.). \bibinfo{publisher}{Springer},
  \bibinfo{pages}{495--513}.
\newblock
\urldef\tempurl%
\url{https://doi.org/10.1007/978-3-031-27481-7\_28}
\showDOI{\tempurl}


\bibitem[Hodkinson et~al\mbox{.}(2003)]%
        {HodkinsonKKWZ03}
\bibfield{author}{\bibinfo{person}{Ian~M. Hodkinson}, \bibinfo{person}{Roman
  Kontchakov}, \bibinfo{person}{Agi Kurucz}, \bibinfo{person}{Frank Wolter},
  {and} \bibinfo{person}{Michael Zakharyaschev}.}
  \bibinfo{year}{2003}\natexlab{}.
\newblock \showarticletitle{On the Computational Complexity of Decidable
  Fragments of First-Order Linear Temporal Logics}. In
  \bibinfo{booktitle}{\emph{10th International Symposium on Temporal
  Representation and Reasoning / 4th International Conference on Temporal Logic
  {(TIME-ICTL} 2003), 8-10 July 2003, Cairns, Queensland, Australia}}.
  \bibinfo{publisher}{{IEEE} Computer Society}, \bibinfo{pages}{91--98}.
\newblock
\urldef\tempurl%
\url{https://doi.org/10.1109/TIME.2003.1214884}
\showDOI{\tempurl}


\bibitem[Hodkinson et~al\mbox{.}(2001)]%
        {HodkinsonWZ01}
\bibfield{author}{\bibinfo{person}{Ian~M. Hodkinson}, \bibinfo{person}{Frank
  Wolter}, {and} \bibinfo{person}{Michael Zakharyaschev}.}
  \bibinfo{year}{2001}\natexlab{}.
\newblock \showarticletitle{Monodic fragments of first-order temporal logics:
  2000-2001 {A.D}}. In \bibinfo{booktitle}{\emph{Logic for Programming,
  Artificial Intelligence, and Reasoning, 8th International Conference, {LPAR}
  2001, Havana, Cuba, December 3-7, 2001, Proceedings}}
  \emph{(\bibinfo{series}{Lecture Notes in Computer Science},
  Vol.~\bibinfo{volume}{2250})}, \bibfield{editor}{\bibinfo{person}{Robert
  Nieuwenhuis} {and} \bibinfo{person}{Andrei Voronkov}} (Eds.).
  \bibinfo{publisher}{Springer}, \bibinfo{pages}{1--23}.
\newblock
\urldef\tempurl%
\url{https://doi.org/10.1007/3-540-45653-8\_1}
\showDOI{\tempurl}


\bibitem[Hodkinson et~al\mbox{.}(2002)]%
        {HodkinsonWZ02}
\bibfield{author}{\bibinfo{person}{Ian~M. Hodkinson}, \bibinfo{person}{Frank
  Wolter}, {and} \bibinfo{person}{Michael Zakharyaschev}.}
  \bibinfo{year}{2002}\natexlab{}.
\newblock \showarticletitle{Decidable and Undecidable Fragments of First-Order
  Branching Temporal Logics}. In \bibinfo{booktitle}{\emph{17th {IEEE}
  Symposium on Logic in Computer Science {(LICS} 2002), 22-25 July 2002,
  Copenhagen, Denmark, Proceedings}}. \bibinfo{publisher}{{IEEE} Computer
  Society}, \bibinfo{pages}{393--402}.
\newblock
\urldef\tempurl%
\url{https://doi.org/10.1109/LICS.2002.1029847}
\showDOI{\tempurl}


\bibitem[Jacobs et~al\mbox{.}(2023)]%
        {TLSF}
\bibfield{author}{\bibinfo{person}{Swen Jacobs}, \bibinfo{person}{Guillermo~A.
  P{\'{e}}rez}, {and} \bibinfo{person}{Philipp Schlehuber{-}Caissier}.}
  \bibinfo{year}{2023}\natexlab{}.
\newblock \showarticletitle{The Temporal Logic Synthesis Format {TLSF} v1.2}.
\newblock \bibinfo{journal}{\emph{CoRR}}  \bibinfo{volume}{abs/2303.03839}
  (\bibinfo{year}{2023}).
\newblock
\urldef\tempurl%
\url{https://doi.org/10.48550/ARXIV.2303.03839}
\showDOI{\tempurl}
\showeprint[arXiv]{2303.03839}


\bibitem[Katis et~al\mbox{.}(2018)]%
        {KatisFGGBGW18}
\bibfield{author}{\bibinfo{person}{Andreas Katis}, \bibinfo{person}{Grigory
  Fedyukovich}, \bibinfo{person}{Huajun Guo}, \bibinfo{person}{Andrew Gacek},
  \bibinfo{person}{John Backes}, \bibinfo{person}{Arie Gurfinkel}, {and}
  \bibinfo{person}{Michael~W. Whalen}.} \bibinfo{year}{2018}\natexlab{}.
\newblock \showarticletitle{Validity-Guided Synthesis of Reactive Systems from
  Assume-Guarantee Contracts}. In \bibinfo{booktitle}{\emph{Tools and
  Algorithms for the Construction and Analysis of Systems - 24th International
  Conference, {TACAS} 2018, Held as Part of the European Joint Conferences on
  Theory and Practice of Software, {ETAPS} 2018, Thessaloniki, Greece, April
  14-20, 2018, Proceedings, Part {II}}} \emph{(\bibinfo{series}{Lecture Notes
  in Computer Science}, Vol.~\bibinfo{volume}{10806})},
  \bibfield{editor}{\bibinfo{person}{Dirk Beyer} {and} \bibinfo{person}{Marieke
  Huisman}} (Eds.). \bibinfo{publisher}{Springer}, \bibinfo{pages}{176--193}.
\newblock
\urldef\tempurl%
\url{https://doi.org/10.1007/978-3-319-89963-3\_10}
\showDOI{\tempurl}


\bibitem[Kret{\'{\i}}nsk{\'{y}} and Esparza(2012)]%
        {KretinskyE12}
\bibfield{author}{\bibinfo{person}{Jan Kret{\'{\i}}nsk{\'{y}}} {and}
  \bibinfo{person}{Javier Esparza}.} \bibinfo{year}{2012}\natexlab{}.
\newblock \showarticletitle{Deterministic Automata for the (F, G)-Fragment of
  {LTL}}. In \bibinfo{booktitle}{\emph{Computer Aided Verification - 24th
  International Conference, {CAV} 2012, Berkeley, CA, USA, July 7-13, 2012
  Proceedings}} \emph{(\bibinfo{series}{Lecture Notes in Computer Science},
  Vol.~\bibinfo{volume}{7358})},
  \bibfield{editor}{\bibinfo{person}{P.~Madhusudan} {and}
  \bibinfo{person}{Sanjit~A. Seshia}} (Eds.). \bibinfo{publisher}{Springer},
  \bibinfo{pages}{7--22}.
\newblock
\urldef\tempurl%
\url{https://doi.org/10.1007/978-3-642-31424-7\_7}
\showDOI{\tempurl}


\bibitem[Li et~al\mbox{.}(2014)]%
        {LiP0VH14}
\bibfield{author}{\bibinfo{person}{Jianwen Li}, \bibinfo{person}{Geguang Pu},
  \bibinfo{person}{Lijun Zhang}, \bibinfo{person}{Moshe~Y. Vardi}, {and}
  \bibinfo{person}{Jifeng He}.} \bibinfo{year}{2014}\natexlab{}.
\newblock \showarticletitle{Fast {LTL} Satisfiability Checking by {SAT}
  Solvers}.
\newblock \bibinfo{journal}{\emph{CoRR}}  \bibinfo{volume}{abs/1401.5677}
  (\bibinfo{year}{2014}).
\newblock
\showeprint[arXiv]{1401.5677}
\urldef\tempurl%
\url{http://arxiv.org/abs/1401.5677}
\showURL{%
\tempurl}


\bibitem[Li et~al\mbox{.}(2013)]%
        {LiZPVH13}
\bibfield{author}{\bibinfo{person}{Jianwen Li}, \bibinfo{person}{Lijun Zhang},
  \bibinfo{person}{Geguang Pu}, \bibinfo{person}{Moshe~Y. Vardi}, {and}
  \bibinfo{person}{Jifeng He}.} \bibinfo{year}{2013}\natexlab{}.
\newblock \showarticletitle{{LTL} Satisfiability Checking Revisited}. In
  \bibinfo{booktitle}{\emph{2013 20th International Symposium on Temporal
  Representation and Reasoning, Pensacola, FL, USA, September 26-28, 2013}},
  \bibfield{editor}{\bibinfo{person}{C{\'{e}}sar S{\'{a}}nchez},
  \bibinfo{person}{Kristen~Brent Venable}, {and} \bibinfo{person}{Esteban
  Zim{\'{a}}nyi}} (Eds.). \bibinfo{publisher}{{IEEE} Computer Society},
  \bibinfo{pages}{91--98}.
\newblock
\urldef\tempurl%
\url{https://doi.org/10.1109/TIME.2013.19}
\showDOI{\tempurl}


\bibitem[Luttenberger et~al\mbox{.}(2020)]%
        {LuttenbergerMS20}
\bibfield{author}{\bibinfo{person}{Michael Luttenberger},
  \bibinfo{person}{Philipp~J. Meyer}, {and} \bibinfo{person}{Salomon Sickert}.}
  \bibinfo{year}{2020}\natexlab{}.
\newblock \showarticletitle{Practical synthesis of reactive systems from {LTL}
  specifications via parity games}.
\newblock \bibinfo{journal}{\emph{Acta Informatica}} \bibinfo{volume}{57},
  \bibinfo{number}{1-2} (\bibinfo{year}{2020}), \bibinfo{pages}{3--36}.
\newblock
\urldef\tempurl%
\url{https://doi.org/10.1007/S00236-019-00349-3}
\showDOI{\tempurl}


\bibitem[Maderbacher and Bloem(2022)]%
        {MaderbacherB22}
\bibfield{author}{\bibinfo{person}{Benedikt Maderbacher} {and}
  \bibinfo{person}{Roderick Bloem}.} \bibinfo{year}{2022}\natexlab{}.
\newblock \showarticletitle{Reactive Synthesis Modulo Theories using
  Abstraction Refinement}. In \bibinfo{booktitle}{\emph{22nd Formal Methods in
  Computer-Aided Design, {FMCAD} 2022, Trento, Italy, October 17-21, 2022}},
  \bibfield{editor}{\bibinfo{person}{Alberto Griggio} {and}
  \bibinfo{person}{Neha Rungta}} (Eds.). \bibinfo{publisher}{{IEEE}},
  \bibinfo{pages}{315--324}.
\newblock
\urldef\tempurl%
\url{https://doi.org/10.34727/2022/ISBN.978-3-85448-053-2\_38}
\showDOI{\tempurl}


\bibitem[Meyer et~al\mbox{.}(2018)]%
        {strix}
\bibfield{author}{\bibinfo{person}{Philipp~J. Meyer}, \bibinfo{person}{Salomon
  Sickert}, {and} \bibinfo{person}{Michael Luttenberger}.}
  \bibinfo{year}{2018}\natexlab{}.
\newblock \showarticletitle{Strix: Explicit Reactive Synthesis Strikes Back!}.
  In \bibinfo{booktitle}{\emph{Computer Aided Verification - 30th International
  Conference, {CAV} 2018, Held as Part of the Federated Logic Conference, FloC
  2018, Oxford, UK, July 14-17, 2018, Proceedings, Part {I}}}
  \emph{(\bibinfo{series}{Lecture Notes in Computer Science},
  Vol.~\bibinfo{volume}{10981})}, \bibfield{editor}{\bibinfo{person}{Hana
  Chockler} {and} \bibinfo{person}{Georg Weissenbacher}} (Eds.).
  \bibinfo{publisher}{Springer}, \bibinfo{pages}{578--586}.
\newblock
\urldef\tempurl%
\url{https://doi.org/10.1007/978-3-319-96145-3\_31}
\showDOI{\tempurl}


\bibitem[Miyano and Hayashi(1984)]%
        {MiyanoH84}
\bibfield{author}{\bibinfo{person}{Satoru Miyano} {and}
  \bibinfo{person}{Takeshi Hayashi}.} \bibinfo{year}{1984}\natexlab{}.
\newblock \showarticletitle{Alternating Finite Automata on omega-Words}.
\newblock \bibinfo{journal}{\emph{Theor. Comput. Sci.}}  \bibinfo{volume}{32}
  (\bibinfo{year}{1984}), \bibinfo{pages}{321--330}.
\newblock
\urldef\tempurl%
\url{https://doi.org/10.1016/0304-3975(84)90049-5}
\showDOI{\tempurl}


\bibitem[Neider and Topcu(2016)]%
        {NeiderT16}
\bibfield{author}{\bibinfo{person}{Daniel Neider} {and} \bibinfo{person}{Ufuk
  Topcu}.} \bibinfo{year}{2016}\natexlab{}.
\newblock \showarticletitle{An Automaton Learning Approach to Solving Safety
  Games over Infinite Graphs}. In \bibinfo{booktitle}{\emph{Tools and
  Algorithms for the Construction and Analysis of Systems - 22nd International
  Conference, {TACAS} 2016, Held as Part of the European Joint Conferences on
  Theory and Practice of Software, {ETAPS 2016, Eindhoven, The Netherlands,
  April 2-8, 2016, Proceedings}}} \emph{(\bibinfo{series}{Lecture Notes in
  Computer Science}, Vol.~\bibinfo{volume}{9636})},
  \bibfield{editor}{\bibinfo{person}{Marsha Chechik} {and}
  \bibinfo{person}{Jean{-}Fran{\c{c}}ois Raskin}} (Eds.).
  \bibinfo{publisher}{Springer}, \bibinfo{pages}{204--221}.
\newblock
\urldef\tempurl%
\url{https://doi.org/10.1007/978-3-662-49674-9\_12}
\showDOI{\tempurl}


\bibitem[Piterman(2007)]%
        {Piterman07}
\bibfield{author}{\bibinfo{person}{Nir Piterman}.}
  \bibinfo{year}{2007}\natexlab{}.
\newblock \showarticletitle{From Nondeterministic B{\"{u}}chi and Streett
  Automata to Deterministic Parity Automata}.
\newblock \bibinfo{journal}{\emph{Log. Methods Comput. Sci.}}
  \bibinfo{volume}{3}, \bibinfo{number}{3} (\bibinfo{year}{2007}).
\newblock
\urldef\tempurl%
\url{https://doi.org/10.2168/LMCS-3(3:5)2007}
\showDOI{\tempurl}


\bibitem[Pnueli(1977)]%
        {Pnueli77}
\bibfield{author}{\bibinfo{person}{Amir Pnueli}.}
  \bibinfo{year}{1977}\natexlab{}.
\newblock \showarticletitle{The Temporal Logic of Programs}. In
  \bibinfo{booktitle}{\emph{18th Annual Symposium on Foundations of Computer
  Science, Providence, Rhode Island, USA, 31 October - 1 November 1977}}.
  \bibinfo{publisher}{{IEEE} Computer Society}, \bibinfo{pages}{46--57}.
\newblock
\urldef\tempurl%
\url{https://doi.org/10.1109/SFCS.1977.32}
\showDOI{\tempurl}


\bibitem[Pnueli and Rosner(1989)]%
        {PnueliR89}
\bibfield{author}{\bibinfo{person}{Amir Pnueli} {and} \bibinfo{person}{Roni
  Rosner}.} \bibinfo{year}{1989}\natexlab{}.
\newblock \showarticletitle{On the Synthesis of a Reactive Module}. In
  \bibinfo{booktitle}{\emph{Conference Record of the Sixteenth Annual {ACM}
  Symposium on Principles of Programming Languages, Austin, Texas, USA, January
  11-13, 1989}}. \bibinfo{publisher}{{ACM} Press}, \bibinfo{pages}{179--190}.
\newblock
\urldef\tempurl%
\url{https://doi.org/10.1145/75277.75293}
\showDOI{\tempurl}


\bibitem[Reynolds(2016)]%
        {Reynolds16a}
\bibfield{author}{\bibinfo{person}{Mark Reynolds}.}
  \bibinfo{year}{2016}\natexlab{}.
\newblock \showarticletitle{A New Rule for {LTL} Tableaux}. In
  \bibinfo{booktitle}{\emph{Proceedings of the Seventh International Symposium
  on Games, Automata, Logics and Formal Verification, GandALF 2016, Catania,
  Italy, 14-16 September 2016}} \emph{(\bibinfo{series}{{EPTCS}},
  Vol.~\bibinfo{volume}{226})}, \bibfield{editor}{\bibinfo{person}{Domenico
  Cantone} {and} \bibinfo{person}{Giorgio Delzanno}} (Eds.).
  \bibinfo{pages}{287--301}.
\newblock
\urldef\tempurl%
\url{https://doi.org/10.4204/EPTCS.226.20}
\showDOI{\tempurl}


\bibitem[Rodr{\'{\i}}guez and S{\'{a}}nchez(2023)]%
        {RodriguezS23}
\bibfield{author}{\bibinfo{person}{Andoni Rodr{\'{\i}}guez} {and}
  \bibinfo{person}{C{\'{e}}sar S{\'{a}}nchez}.}
  \bibinfo{year}{2023}\natexlab{}.
\newblock \showarticletitle{Boolean Abstractions for Realizability Modulo
  Theories}. In \bibinfo{booktitle}{\emph{Computer Aided Verification - 35th
  International Conference, {CAV} 2023, Paris, France, July 17-22, 2023,
  Proceedings, Part {III}}} \emph{(\bibinfo{series}{Lecture Notes in Computer
  Science}, Vol.~\bibinfo{volume}{13966})},
  \bibfield{editor}{\bibinfo{person}{Constantin Enea} {and}
  \bibinfo{person}{Akash Lal}} (Eds.). \bibinfo{publisher}{Springer},
  \bibinfo{pages}{305--328}.
\newblock
\urldef\tempurl%
\url{https://doi.org/10.1007/978-3-031-37709-9\_15}
\showDOI{\tempurl}


\bibitem[Rozier and Vardi(2010)]%
        {RozierV10}
\bibfield{author}{\bibinfo{person}{Kristin~Y. Rozier} {and}
  \bibinfo{person}{Moshe~Y. Vardi}.} \bibinfo{year}{2010}\natexlab{}.
\newblock \showarticletitle{{LTL} satisfiability checking}.
\newblock \bibinfo{journal}{\emph{Int. J. Softw. Tools Technol. Transf.}}
  \bibinfo{volume}{12}, \bibinfo{number}{2} (\bibinfo{year}{2010}),
  \bibinfo{pages}{123--137}.
\newblock
\urldef\tempurl%
\url{https://doi.org/10.1007/S10009-010-0140-3}
\showDOI{\tempurl}


\bibitem[Safra(1988)]%
        {Safra88}
\bibfield{author}{\bibinfo{person}{Shmuel Safra}.}
  \bibinfo{year}{1988}\natexlab{}.
\newblock \showarticletitle{On the Complexity of omega-Automata}. In
  \bibinfo{booktitle}{\emph{29th Annual Symposium on Foundations of Computer
  Science, White Plains, New York, USA, 24-26 October 1988}}.
  \bibinfo{publisher}{{IEEE} Computer Society}, \bibinfo{pages}{319--327}.
\newblock
\urldef\tempurl%
\url{https://doi.org/10.1109/SFCS.1988.21948}
\showDOI{\tempurl}


\bibitem[Samuel et~al\mbox{.}(2021)]%
        {SamuelDK21}
\bibfield{author}{\bibinfo{person}{Stanly Samuel}, \bibinfo{person}{Deepak
  D'Souza}, {and} \bibinfo{person}{Raghavan Komondoor}.}
  \bibinfo{year}{2021}\natexlab{}.
\newblock \showarticletitle{GenSys: a scalable fixed-point engine for maximal
  controller synthesis over infinite state spaces}. In
  \bibinfo{booktitle}{\emph{{ESEC/FSE} '21: 29th {ACM} Joint European Software
  Engineering Conference and Symposium on the Foundations of Software
  Engineering, Athens, Greece, August 23-28, 2021}},
  \bibfield{editor}{\bibinfo{person}{Diomidis Spinellis},
  \bibinfo{person}{Georgios Gousios}, \bibinfo{person}{Marsha Chechik}, {and}
  \bibinfo{person}{Massimiliano~Di Penta}} (Eds.). \bibinfo{publisher}{{ACM}},
  \bibinfo{pages}{1585--1589}.
\newblock
\urldef\tempurl%
\url{https://doi.org/10.1145/3468264.3473126}
\showDOI{\tempurl}


\bibitem[Samuel et~al\mbox{.}(2023)]%
        {SamuelDK23}
\bibfield{author}{\bibinfo{person}{Stanly Samuel}, \bibinfo{person}{Deepak
  D'Souza}, {and} \bibinfo{person}{Raghavan Komondoor}.}
  \bibinfo{year}{2023}\natexlab{}.
\newblock \showarticletitle{Symbolic Fixpoint Algorithms for Logical {LTL}
  Games}. In \bibinfo{booktitle}{\emph{38th {IEEE/ACM} International Conference
  on Automated Software Engineering, {ASE} 2023, Luxembourg, September 11-15,
  2023}}. \bibinfo{publisher}{{IEEE}}, \bibinfo{pages}{698--709}.
\newblock
\urldef\tempurl%
\url{https://doi.org/10.1109/ASE56229.2023.00212}
\showDOI{\tempurl}


\bibitem[Schmuck et~al\mbox{.}(2024)]%
        {SchmuckHDN24}
\bibfield{author}{\bibinfo{person}{Anne{-}Kathrin Schmuck},
  \bibinfo{person}{Philippe Heim}, \bibinfo{person}{Rayna Dimitrova}, {and}
  \bibinfo{person}{Satya~Prakash Nayak}.} \bibinfo{year}{2024}\natexlab{}.
\newblock \showarticletitle{Localized Attractor Computations for Infinite-State
  Games}. In \bibinfo{booktitle}{\emph{Computer Aided Verification - 36th
  International Conference, {CAV} 2024, Montreal, QC, Canada, July 24-27, 2024,
  Proceedings, Part {III}}} \emph{(\bibinfo{series}{Lecture Notes in Computer
  Science}, Vol.~\bibinfo{volume}{14683})},
  \bibfield{editor}{\bibinfo{person}{Arie Gurfinkel} {and}
  \bibinfo{person}{Vijay Ganesh}} (Eds.). \bibinfo{publisher}{Springer},
  \bibinfo{pages}{135--158}.
\newblock
\urldef\tempurl%
\url{https://doi.org/10.1007/978-3-031-65633-0\_7}
\showDOI{\tempurl}


\bibitem[Schwendimann(1998)]%
        {Schwendimann98}
\bibfield{author}{\bibinfo{person}{Stefan Schwendimann}.}
  \bibinfo{year}{1998}\natexlab{}.
\newblock \showarticletitle{A New One-Pass Tableau Calculus for {PLTL}}. In
  \bibinfo{booktitle}{\emph{Automated Reasoning with Analytic Tableaux and
  Related Methods, International Conference, {TABLEAUX} '98, Oisterwijk, The
  Netherlands, May 5-8, 1998, Proceedings}} \emph{(\bibinfo{series}{Lecture
  Notes in Computer Science}, Vol.~\bibinfo{volume}{1397})},
  \bibfield{editor}{\bibinfo{person}{Harrie C.~M. de~Swart}} (Ed.).
  \bibinfo{publisher}{Springer}, \bibinfo{pages}{277--292}.
\newblock
\urldef\tempurl%
\url{https://doi.org/10.1007/3-540-69778-0\_28}
\showDOI{\tempurl}


\bibitem[Sen et~al\mbox{.}(2003)]%
        {SenRA03}
\bibfield{author}{\bibinfo{person}{Koushik Sen}, \bibinfo{person}{Grigore
  Rosu}, {and} \bibinfo{person}{Gul Agha}.} \bibinfo{year}{2003}\natexlab{}.
\newblock \showarticletitle{Generating Optimal Linear Temporal Logic Monitors
  by Coinduction}. In \bibinfo{booktitle}{\emph{Advances in Computing Science -
  {ASIAN} 2003 Programming Languages and Distributed Computation, 8th Asian
  Computing Science Conference, Mumbai, India, December 10-14, 2003,
  Proceedings}} \emph{(\bibinfo{series}{Lecture Notes in Computer Science},
  Vol.~\bibinfo{volume}{2896})}, \bibfield{editor}{\bibinfo{person}{Vijay~A.
  Saraswat}} (Ed.). \bibinfo{publisher}{Springer}, \bibinfo{pages}{260--275}.
\newblock
\urldef\tempurl%
\url{https://doi.org/10.1007/978-3-540-40965-6\_17}
\showDOI{\tempurl}


\bibitem[Sistla(1994)]%
        {Sistla94}
\bibfield{author}{\bibinfo{person}{A.~Prasad Sistla}.}
  \bibinfo{year}{1994}\natexlab{}.
\newblock \showarticletitle{Safety, Liveness and Fairness in Temporal Logic}.
\newblock \bibinfo{journal}{\emph{Formal Aspects Comput.}} \bibinfo{volume}{6},
  \bibinfo{number}{5} (\bibinfo{year}{1994}), \bibinfo{pages}{495--512}.
\newblock
\urldef\tempurl%
\url{https://doi.org/10.1007/BF01211865}
\showDOI{\tempurl}


\bibitem[Standards(2010)]%
        {PSL-standard}
\bibfield{author}{\bibinfo{person}{IEEE Standards}.}
  \bibinfo{year}{2010}\natexlab{}.
\newblock \showarticletitle{IEEE Standard for Property Specification Language
  (PSL)}.
\newblock \bibinfo{journal}{\emph{IEEE Std 1850-2010 (Revision of IEEE Std
  1850-2005)}} (\bibinfo{year}{2010}), \bibinfo{pages}{1--182}.
\newblock
\urldef\tempurl%
\url{https://doi.org/10.1109/IEEESTD.2010.5446004}
\showDOI{\tempurl}


\bibitem[Unno et~al\mbox{.}(2023)]%
        {UnnoTGK23}
\bibfield{author}{\bibinfo{person}{Hiroshi Unno}, \bibinfo{person}{Tachio
  Terauchi}, \bibinfo{person}{Yu Gu}, {and} \bibinfo{person}{Eric Koskinen}.}
  \bibinfo{year}{2023}\natexlab{}.
\newblock \showarticletitle{Modular Primal-Dual Fixpoint Logic Solving for
  Temporal Verification}.
\newblock \bibinfo{journal}{\emph{Proc. {ACM} Program. Lang.}}
  \bibinfo{volume}{7}, \bibinfo{number}{{POPL}} (\bibinfo{year}{2023}),
  \bibinfo{pages}{2111--2140}.
\newblock
\urldef\tempurl%
\url{https://doi.org/10.1145/3571265}
\showDOI{\tempurl}


\bibitem[Vardi(1994)]%
        {Vardi94}
\bibfield{author}{\bibinfo{person}{Moshe~Y. Vardi}.}
  \bibinfo{year}{1994}\natexlab{}.
\newblock \showarticletitle{Nontraditional Applications of Automata Theory}. In
  \bibinfo{booktitle}{\emph{Theoretical Aspects of Computer Software,
  International Conference {TACS} '94, Sendai, Japan, April 19-22, 1994,
  Proceedings}} \emph{(\bibinfo{series}{Lecture Notes in Computer Science},
  Vol.~\bibinfo{volume}{789})}, \bibfield{editor}{\bibinfo{person}{Masami
  Hagiya} {and} \bibinfo{person}{John~C. Mitchell}} (Eds.).
  \bibinfo{publisher}{Springer}, \bibinfo{pages}{575--597}.
\newblock
\urldef\tempurl%
\url{https://doi.org/10.1007/3-540-57887-0\_116}
\showDOI{\tempurl}


\bibitem[Vardi(1995)]%
        {Vardi95}
\bibfield{author}{\bibinfo{person}{Moshe~Y. Vardi}.}
  \bibinfo{year}{1995}\natexlab{}.
\newblock \showarticletitle{An Automata-Theoretic Approach to Linear Temporal
  Logic}. In \bibinfo{booktitle}{\emph{Logics for Concurrency - Structure
  versus Automata (8th Banff Higher Order Workshop, Banff, Canada, August 27 -
  September 3, 1995, Proceedings)}} \emph{(\bibinfo{series}{Lecture Notes in
  Computer Science}, Vol.~\bibinfo{volume}{1043})},
  \bibfield{editor}{\bibinfo{person}{Faron Moller} {and}
  \bibinfo{person}{Graham~M. Birtwistle}} (Eds.).
  \bibinfo{publisher}{Springer}, \bibinfo{pages}{238--266}.
\newblock
\urldef\tempurl%
\url{https://doi.org/10.1007/3-540-60915-6\_6}
\showDOI{\tempurl}


\bibitem[Veanes et~al\mbox{.}(2023)]%
        {VeanesBES23}
\bibfield{author}{\bibinfo{person}{Margus Veanes}, \bibinfo{person}{Thomas
  Ball}, \bibinfo{person}{Gabriel Ebner}, {and} \bibinfo{person}{Olli
  Saarikivi}.} \bibinfo{year}{2023}\natexlab{}.
\newblock \showarticletitle{Symbolic Automata: {\(\omega\)}-Regularity Modulo
  Theories}.
\newblock \bibinfo{journal}{\emph{CoRR}}  \bibinfo{volume}{abs/2310.02393}
  (\bibinfo{year}{2023}).
\newblock
\urldef\tempurl%
\url{https://doi.org/10.48550/ARXIV.2310.02393}
\showDOI{\tempurl}
\showeprint[arXiv]{2310.02393}


\bibitem[Vechev et~al\mbox{.}(2013)]%
        {VechevYY13}
\bibfield{author}{\bibinfo{person}{Martin~T. Vechev}, \bibinfo{person}{Eran
  Yahav}, {and} \bibinfo{person}{Greta Yorsh}.}
  \bibinfo{year}{2013}\natexlab{}.
\newblock \showarticletitle{Abstraction-guided synthesis of synchronization}.
\newblock \bibinfo{journal}{\emph{Int. J. Softw. Tools Technol. Transf.}}
  \bibinfo{volume}{15}, \bibinfo{number}{5-6} (\bibinfo{year}{2013}),
  \bibinfo{pages}{413--431}.
\newblock
\urldef\tempurl%
\url{https://doi.org/10.1007/S10009-012-0232-3}
\showDOI{\tempurl}


\bibitem[Walker and Ryzhyk(2014)]%
        {WalkerR14}
\bibfield{author}{\bibinfo{person}{Adam Walker} {and} \bibinfo{person}{Leonid
  Ryzhyk}.} \bibinfo{year}{2014}\natexlab{}.
\newblock \showarticletitle{Predicate abstraction for reactive synthesis}. In
  \bibinfo{booktitle}{\emph{Formal Methods in Computer-Aided Design, {FMCAD}
  2014, Lausanne, Switzerland, October 21-24, 2014}}.
  \bibinfo{publisher}{{IEEE}}, \bibinfo{pages}{219--226}.
\newblock
\urldef\tempurl%
\url{https://doi.org/10.1109/FMCAD.2014.6987617}
\showDOI{\tempurl}


\bibitem[Wolper(1985)]%
        {Wolper85}
\bibfield{author}{\bibinfo{person}{Pierre Wolper}.}
  \bibinfo{year}{1985}\natexlab{}.
\newblock \showarticletitle{The tableau method for temporal logic: an
  overview}.
\newblock \bibinfo{journal}{\emph{Logique Et Analyse}}  \bibinfo{volume}{28}
  (\bibinfo{year}{1985}), \bibinfo{pages}{119--136}.
\newblock
\urldef\tempurl%
\url{https://api.semanticscholar.org/CorpusID:118632087}
\showURL{%
\tempurl}


\end{thebibliography}

\newpage
\appendix
\section{Proofs}\label{sec:approofs}
\subsection{Proofs from \Cref*{sec:logic_and_realizability} and \Cref*{sec:monitor_definition}}

Let $(\symgame,\Lambda)$ be the symbolic game with $\symgame = (L, \linit,\inputs, \progvars, \dom, \delta)$.
For a $\rho \in \xiruns$ we define $\uniquerun(\rho,\symgame) \in \states^\omega$ as the unique run in $\sema{\symgame}$ that starts in $\linit$ and conforms to all assignments to $\inputs$ and $\progvars$. 
The run is unique by the first condition \Cref*{def:sym-games}.
We define the \emph{pseudo-language} of $(\symgame,\Lambda)$ as $\plang{\symgame,\Lambda} := \{ \rho \in \xiruns \mid \uniquerun(\rho,\symgame) \in \sema{\Lambda}\}$.

Let $\varphi$ be an $\templogic(specvars)$ formula,  and $\tobool{\varphi}$ be its propositional version.

For each $\rho \in \xiruns$,  there exists a unique sequence $\tobool{\rho} \in \apruns$ such that for all $i \in \Nat$, and all $p \in \AP(\tobool{\varphi})$,
$\rho[i] \models p$ if and only if $p \in \tobool{\rho}[i]$.

By structural induction on the definition of $\templogic$,  it is easy to show that
\[\rho \models \varphi \Longleftrightarrow \tobool{\rho} \models \tobool{\varphi}. \]

For each $\eta \in \apruns$, we define 
$\concretize(\eta) := \{\rho \in \xiruns \mid \tobool{\rho} = \eta\}.$

By structural induction on the definition of $\templogic$,  it is easy to show that
\[\eta \models \tobool{\varphi} \Longrightarrow \concretize(\eta) \subseteq \lang{\varphi}. \]

\begin{lemma}\label{lem:langeq}
    Let $\varphi \in \templogic(\specvars)$ and $(\symgame,\Lambda)$ be \textbf{a symbolic game for $\varphi$}, then
    \[\lang{\varphi} = \plang{\symgame, \Lambda}.\]
\end{lemma}
\begin{proof}
    Since $(\symgame,\Lambda)$ is a symbolic game for $\varphi$,  there exists a DPA $\mathcal{A}_{\tobool{\varphi}}$ from which $(\symgame,\Lambda)$ was constructed. We prove the two inclusions separately.
    
    $(\subseteq)\quad$ 
    Let $\rho \in \lang{\varphi}$. 
    Then,  $\tobool{\rho} \in \lang{\mathcal{A}_{\tobool{\varphi}}} $.
    Since  $(\symgame,\Lambda)$ is constructed from $\mathcal{A}_{\tobool{\varphi}}$,  we have that $\concretize(\tobool{\rho}) \subseteq \plang{\symgame,\Lambda} $. 
    Thus,  $\rho \in \plang{\symgame,\Lambda}$. 
    
     $(\supseteq)\quad$ 
     Let $\rho \in \plang{\symgame,\Lambda}$.
     Then,  since $(\symgame,\Lambda)$ is constructed from $\mathcal{A}_{\tobool{\varphi}}$,  we have that $\tobool{\rho} \in \lang{\mathcal{A}_{\tobool{\varphi}}} $.
     Therefore,  $\concretize(\tobool{\rho}) \subseteq \lang{\varphi}$.
     Thus, $\rho \in  \lang{\varphi}$.
\end{proof}

\begin{lemma}\label{lem:langeqeqreal}
    Let 
    $\varphi \in \templogic(\specvars)$ and 
    $(\symgame,\Lambda)$ be \textbf{some symbolic game} where $\symgame = (L, \linit,\inputs, \progvars, \dom, \delta)$
    \textbf{such that} $\lang{\varphi} = \mathcal{L}_P(\symgame, \Lambda)$. 
    Then $\varphi$ is realizable if and only if $(\linit,\assmt{x}) \in \win_\sys(\sema{\symgame},\sema{\Lambda})$ for every $\assmt{x} \in \assignments{\progvars}$.
\end{lemma}
\begin{proof} We prove separately the two directions.

 $(\Longrightarrow)\quad$
 Let  $\varphi$ be realizable.  
 Then, there exists a function 
$\sigma: \assignments{\progvars} \times \assignments{\inputs}^+ \to \assignments{\progvars}$ such that for every infinite sequence $\mathit{input} \in \assignments{\inputs}^\omega$ of assignments to $\inputs$ and initial assignment $\mathit{init} \in \assignments{\progvars}$, for  $\rho \in {(\assignments{\progvars \cup \inputs})}^\omega$ defined as 
$\rho[0] := \mathit{init} \cup \mathit{input}[0] $ and $\rho[n] := \sigma(\mathit{init}, \mathit{input}[0,n-1]) \cup \mathit{input}[n]$ for $n > 0$,  $\rho \models \varphi$ holds.
   
   We can define a strategy $\sigma_\sys$ for Player $\sys$ in $\sema{\symgame}$ that emulates the function $\sigma$.  
   Then, every $\xi \in \plays_{\sema{\symgame}}((\linit,\assmt{x}),\sigma_\sys)$ corresponds  to a sequence $\rho \in {(\assignments{\progvars \cup \inputs})}^\omega$ consistent with the function $\sigma$.
   Thus,  since $\lang{\varphi} = \mathcal{L}_P(\symgame, \Lambda)$, 
   we have that $\xi \in \sema{\Lambda}$.
   Therefore,  we conclude that $\sigma_\sys$ is a winning strategy for 
 Player $\sys$ from $(\linit,\assmt{x})$, and hence
$(\linit,\assmt{x}) \in  \win_\sys(\sema{\symgame},\sema{\Lambda})$.
 
 $(\Longleftarrow)\quad$
  Suppose that $(\linit,\assmt{x}) \in \win_\sys(\sema{\symgame},\sema{\Lambda})$ for every $\assmt{x} \in \assignments{\progvars}$.
  Then, for every $\assmt{x} \in \assignments{\progvars}$,  there exists a winning strategy for Player $\sys$ in $\sema{\symgame}$ from $(\linit,\assmt{x})$.
  We can define a function $\sigma: \assignments{\progvars} \times \assignments{\inputs}^+ \to \assignments{\progvars}$ that for every $\assmt{x} \in \assignments{\progvars}$ mimics the respective winning strategy for player $\sys$. 
  Thus,   for every infinite sequence $\mathit{input} \in \assignments{\inputs}^\omega$ of assignments to $\inputs$ and initial assignment $\mathit{init} \in \assignments{\progvars}$, for  $\rho \in {(\assignments{\progvars \cup \inputs})}^\omega$ defined as 
$\rho[0] := \mathit{init} \cup \mathit{input}[0] $ and $\rho[n] := \sigma(\mathit{init}, \mathit{input}[0,n-1]) \cup \mathit{input}[n]$ for $n > 0$,  we have that $\rho$ corresponds to a play consistent with the respective winning strategy for player $\sys$ from $\assmt{x}$.
Since $\lang{\varphi} = \mathcal{L}_P(\symgame, \Lambda)$,  we have $\rho \models \varphi$.
\end{proof}

\restateGameCorrectness*
\begin{proof}
This is a direct consequence of \Cref*{lem:langeq} and \Cref*{lem:langeqeqreal}.
The condition $\assmt{x} \FOLentailsT{T} \dom(\linit)$ is not relevant as our construction sets $\dom$ always to true.
\end{proof}

\restateProductCorrectness*
\begin{proof}
Since $(\symgame,\Lambda)$ is a symbolic game for $\varphi$, we have that $\dom(l) = \true$ of all $l \in L$. 
By the definition of $(\symgame_\times,\Lambda_\times)$,
$\dom_\times(l,q) := \dom(l) = \true$ for all $(l,q)$.
This, the condition $\assmt{x} \FOLentailsT{T} \dom((\linit,\qinit))$ is not relevant.  
We will show that 
$\plang{\symgame_\times,\Lambda_\times} =
 \plang{\symgame,\Lambda}$,  
 and the desired claim will directly follow from \Cref*{lem:langeq} and \Cref*{lem:langeqeqreal}.
We show the two inclusions separately.

$(\subseteq)\quad$ 
 Let $\rho  \in \plang{\symgame_\times,\Lambda_\times}$.
 Thus,  $\uniquerun(\rho,\symgame_\times) \in \sema{\Lambda_\times}$.
 The projection of  $\uniquerun(\rho,\symgame_\times)$ on $L$ is equal to $\uniquerun(\rho,\symgame)$. 
 If $\uniquerun(\rho,\symgame) \in \sema{\Lambda}$, then we are done.
 Suppose that $\uniquerun(\rho,\symgame) \not\in \sema{\Lambda}$.
 Let  $\uniquerun(\rho,\symgame_\times) = (l_0,q_0)(l_1,q_1)\ldots\in \Lambda_\times$.
 By definition of $\Lambda_\times$,
 $\verdict(q_i) \neq \UNSAT$ for all $i \in \Nat$,  and
there exists $i \in \Nat$ such that $\verdict(q_i)  = \SAFETY$.
By \Cref*{def:monitor-formula},  the second condition, we have that $\rho \in \lang{\varphi}$.  
Since  $(\symgame,\Lambda)$ is a symbolic game for $\varphi$,  we have by \Cref*{lem:langeq} that 
$\lang{\varphi} = \plang{\symgame, \Lambda}$.
Thus,  $\rho \in  \plang{\symgame, \Lambda}$, which contradicts our supposition that $\uniquerun(\rho,\symgame) \not\in \sema{\Lambda}$. 
This concludes the proof in this direction.

$(\supseteq)\quad$  Let $\rho  \in \plang{\symgame,\Lambda}$.
 Thus,  $\uniquerun(\rho,\symgame) \in \sema{\Lambda}$.
 Consider $\uniquerun(\rho,\symgame_\times)$.  
 The projection of  $\uniquerun(\rho,\symgame_\times)$ on $L$ is equal to $\uniquerun(\rho,\symgame)$. 
 Thus, by the definition of $\Lambda_\times$, since $\uniquerun(\rho,\symgame) \in \sema{\Lambda}$, 
 we also have $\uniquerun(\rho,\symgame_\times) \in \Lambda_\times$.
\end{proof}

\subsection{Proofs from \Cref*{sec:monitor}}
\restateLemaExpansion*
\begin{proof}
The proofs follows a standard argument via induction on $\varphi$~\cite{EsparzaKS20}.
\begin{itemize}
\item $\varphi = \alpha$:
\[
\rho \models \varphi \Longleftrightarrow 
\langle\rho[0], \rho[1] \rangle \FOLentailsT{T} \alpha\Longleftrightarrow 
 \expand(\alpha, a) = \top \Longleftrightarrow 
 \rho_{+1} \models \expand(\varphi, a) 
\]
\item $\varphi = \neg\psi$:
\[
\rho \models \neg\psi \Longleftrightarrow 
\rho \not\models \psi\Longleftrightarrow 
 \rho_{+1} \models \neg \expand(\psi, a) \Longleftrightarrow
 \rho_{+1} \models \expand(\varphi, a)
 \]
\item $\varphi = \varphi_1 \land \varphi_2$:
\[\begin{array}{l}
\rho \models \varphi_1 \land \varphi_2 \Longleftrightarrow \\
 \rho \models \varphi_1 \text{ and } \rho \models \varphi_2 \Longleftrightarrow \\
 \rho_{+1} \models \expand(\varphi_1, a) \text{ and }
  \rho_{+1} \models \expand(\varphi_2, a) \Longleftrightarrow \\
 \rho_{+1} \models \expand(\varphi, a)
 \end{array}
 \]
\item $\varphi = \LTLnext \psi$:
\[
\rho \models \LTLnext\psi \Longleftrightarrow 
\rho_{+1} \models \psi\Longleftrightarrow 
 \rho_{+1} \models \expand(\varphi, a)
 \]
\item $\varphi = \varphi_1 \LTLuntil \varphi_2$:
\[\begin{array}{l}
\rho \models \varphi_1\LTLuntil \varphi_2 \Longleftrightarrow \\
 \rho \models \varphi_2 \vee (\varphi_1 \land\LTLnext \varphi) \Longleftrightarrow \\
 \rho_{+1} \models \expand(\varphi_2,a) \vee \expand(\varphi_1 ,a) \land (\varphi_1 \LTLuntil \varphi_2)\Longleftrightarrow \\
 \rho_{+1} \models \expand(\varphi, a)
 \end{array}
 \]
\end{itemize}
\end{proof}

\begin{lemma}[Lift~\Cref*{lem:expansion}]\label{lem:liftexpansion}
For all $\rho \in {(\assignments{\progvars \cup \inputs})}^\omega$ and $a \subseteq \preds$ where 
$\langle \rho[0], \rho[1] \rangle \FOLentailsT{T} \left( \bigwedge_{\alpha \in \preds \cap a} \alpha \right) \land \left( \bigwedge_{\alpha \in \preds \setminus a} \lnot \alpha \right)$, 
it holds that
$\rho \models \stateformula(q)$ if and only if $\rho_{+1} \models \stateformula(\nextState(q, a))$.
\end{lemma}
\begin{proof}
This follows from~\Cref*{lem:expansion} and the fact that $\nextState$ applies $\expand$ point-wise to the elements of $q$.
This corresponds to applying $\expand$ to the overall $\stateformula(q)$.

The only difference is the use of $\propagate$. 
However,  $\propagate$ only contains predicates that are trivially implied by $a$,  and we only consider $\rho$ where $a$ holds initially. 
Therefore,  the elements added to $\nextState$ by $\propagate$ also hold for $\rho_{+1}$.
\end{proof}

\restateRuleSoundness*
\begin{proof}

We begin with a useful observation and a lemma.

By definition, the formulas $\Curr_D(q)$ and $\DedInv_D(q)$  have the following properties
\begin{itemize}
\item The implications 
$\toform{\formE_\as} \rightarrow \Curr_\as(q)$ and 
$\toform{\formE_\as \cup \formE_\ga} \rightarrow \Curr_\ga(q)$
are valid.
\item The implications 
$\toform{\Imp_\as} \rightarrow \LTLglobally(\DedInv_\as(q))$ and
$\toform{\Imp_\ga} \rightarrow \LTLglobally(\DedInv_\ga(q))$
are valid.
\end{itemize}

The next lemma states some simple relationships between the  conditions in \Cref*{def:sound-transform} which will be helpful in establishing that a given transformation meets these conditions.

\begin{lemma}\label{lemma:soundness-conditions}
Let $\transform : Q \to Q$ be a monitor-state transformation function and consider
$q_1 = \langle \formF_\as^1,\formE_\as^1,\formF_\ga^1,\formE_\ga^1,\Imp_\as^1,\Imp_\ga^1\rangle$ and  
$q_2 = \langle \formF_\as^2,\formE_\as^2,\formF_\ga^2,\formE_\ga^2,\Imp_\as^2,\Imp_\ga^2\rangle$ where
$\transform(q_1) = q_2$.
\begin{itemize}
\item If $\Imp_\as^2 = \Imp_\as^1$ and $\Imp_\ga^2 = \Imp_\ga^1$, then condition~(b) in \Cref*{def:sound-transform} is  satisfied.
\item If $\formF_D^2 = \formF_D^1$,  $\formE_D^2 = \formE_D^1$ and $\Imp_D^2 \supseteq \Imp_D^1$, for all $D \in \agset$,  and if condition~(b) in \Cref*{def:sound-transform} is satisfied, then conditions~(a) and (c) in \Cref*{def:sound-transform} are also satisfied. 
\item If $\formE_D^2 = \formE_D^1$, then condition~(c) in \Cref*{def:sound-transform} is satisfied. 
\end{itemize}
\end{lemma}

We show for each rule in $\ruleset$ that the respective monitor-state transformation
$\transform$ satisfies the conditions of \Cref*{def:sound-transform} for all
$q = \langle \formF_\as,\formE_\as,\formF_\ga,\formE_\ga,\Imp_\as,\Imp_\ga\rangle$ and  
$q' = \langle \formF_\as',\formE_\as',\formF_\ga',\formE_\ga',\Imp_\as',\Imp_\ga'\rangle$ with 
$\transform(q)=q'$.

\medskip

\noindent
\textbf{Rule \substunsat.}
The premise $ \Curr_D(q) \land \DedInv_D(q) \FOLentailsT{T} \false$ 
entails that $\toform{\formF_D \cup \formE_D \cup \Imp_D} \equiv \false$.
The rule sets $\formE_D' =\{\bot\}$, and hence condition (c) is satisfied.

\medskip

\noindent
\textbf{Rule \substunsatF.}
The  premises 
$\LTLglobally(\gamma \to \LTLeventually \beta) \in \Imp_D$ and 
$\Curr_D(q) \FOLentailsT{T}  \gamma$
ensure that $\Curr_D(q) \land \toform{\Imp_D}$ implies $\LTLeventually \beta$.  
Together with the premise $\beta \land  \DedInv_D(q)\FOLentailsT{T}  \false$ 
this entails that $\Curr_D(q) \land \toform{\Imp_D}$ is unsatsifiable.
The rule sets $\formE_D' =\{\bot\}$, and hence condition (c) is satisfied.

\medskip

\noindent
\textbf{Rule \substtrue\ and rule \substfalse.}
The premise of the rule guarantees that 
$\toform{ \Imp_D} \rightarrow \gamma$ 
(respectively $\toform{\Imp_D} \rightarrow \neg\gamma$) is valid.
Thus,  $\toform{\formF_\ga \cup \formE_\ga \cup \Imp_\ga} \land \gamma \equiv \toform{\formF_\ga \cup \formE_\ga \cup \Imp_\ga}$ 
(respectively $\toform{\formF_\ga \cup \formE_\ga \cup \Imp_\ga} \land \neg \gamma \equiv \toform{\formF_\ga \cup \formE_\ga \cup \Imp_\ga}$). 
Since the substitution is performed in $F_D$,  we obtain $\toform{\formF_\ga \cup \formE_\ga \cup \Imp_\ga} \equiv \toform{\formF_\ga' \cup \formE_\ga' \cup \Imp_\ga'}$.
The rule does not modify the sets $\Imp_\as$ and $\Imp_\ga$, 
thus condition~(b) is also satisfied.

The proof for the next two rules make use of the following lemma, which follows from the definition of the substitution $\mapsto$ and the semantics of $\LTLglobally$.

\begin{lemma}\label{lemma:subst-globally}
For all \templogic formulas $\psi,\psi_1$ and $\psi_2$,  it holds that
\[\psi \land \LTLglobally(\psi_1 \leftrightarrow \psi_2) \equiv 
\psi[\psi_1 \mapsto \psi_2] \land \LTLglobally(\psi_1 \leftrightarrow \psi_2).\]
\end{lemma}

\medskip

\noindent
\textbf{Rule \simplifyimp.}
The premise of the rule guarantees that 
$\toform{ \Imp_D} \rightarrow  \LTLglobally((\gamma \rightarrow \varphi) \leftrightarrow \true)$.
Furthermore,   we have that 
$\toform{\formF_D}\land \LTLglobally((\gamma \rightarrow \varphi) \leftrightarrow \true) \equiv
 \toform{ \formF_D}[(\gamma \rightarrow \varphi) \mapsto \true]
 \land \LTLglobally((\gamma \rightarrow \varphi) \leftrightarrow \true)$.
Since the substitution is performed only in $F_D$,  we obtain $\toform{\formF_\ga \cup \formE_\ga \cup \Imp_\ga} \equiv \toform{\formF_\ga' \cup \formE_\ga' \cup \Imp_\ga'}$.
The rule does not modify the sets $\Imp_\as$ and $\Imp_\ga$, 
thus condition~(b) is also satisfied.

\medskip

\noindent
\textbf{Rule \simplifyand.}
The soundness argument is similar to that for rule \simplifyimp, 
by establishing that the implication 
$\toform{ \Imp_D} \rightarrow  \LTLglobally((\gamma \land \varphi) \leftrightarrow \gamma)$ 
is valid.

The proof for the next rule makes use of the following lemma, which follows from the definition of the substitution $\mapstonn$ and the semantics of $\LTLglobally$.
\begin{lemma}\label{lemma:subst-current}
For all \templogic formulas $\psi,\psi_1$ and $\psi_2$ such that
$\psi_1$ $\psi_2$ contain no temporal operators, 
\[\psi \land \LTLglobally(\psi_1 \leftrightarrow \psi_2) \equiv 
\psi[\psi_1 \mapstonn \psi_2] \land \LTLglobally(\psi_1 \leftrightarrow \psi_2).\]
\end{lemma}

\medskip

\noindent
\textbf{Rule \simplifynn.}
The premise of the rule guarantees that 
$\toform{\formE_D \cup \Imp_D} \rightarrow  (\gamma \leftrightarrow \true)$.
Furthermore,  since $\gamma \in \QF{\specvars}$, we have that 
$\toform{\formF_D} \land \LTLglobally(\gamma \leftrightarrow \true)
\equiv
 \toform{\formF_D}[\gamma  \mapsto \true] \land \LTLglobally(\gamma \leftrightarrow \true)$.
Since the substitution is performed only in $F_D$,  we obtain 
$\toform{\formF_\ga \cup \formE_\ga \cup \Imp_\ga} \equiv \toform{\formF_\ga' \cup \formE_\ga' \cup \Imp_\ga'}$.
The rule does not modify the sets $\Imp_\as$ and $\Imp_\ga$, 
thus condition~(b) is also satisfied.

\medskip

\noindent
\textbf{Rule \propagateassump.}
Since $\Imp_\ga' = \Imp_\as \cup \Imp_\ga$,  
we have that $\toform{\formF_\ga \cup \formE_\ga \cup \Imp_\ga} \equiv 
\toform{\formF_\as \cup \formE_\as \cup \Imp_\as} 
\rightarrow \toform{\formF_\ga \cup \formE_\ga \cup \Imp_\ga'}
\equiv  \toform{\formF_\ga' \cup \formE_\ga' \cup \Imp_\ga'}$.
Additionally, $\Imp_\ga' = \Imp_\as \cup \Imp_\ga$ entails condition~(b).

\medskip

\noindent
\textbf{Rule \propagateG.}
Since $\toform{ \formE_D \land \bigwedge \Imp_D } \rightarrow \toform{ \Imp_D' }$ 
and $\Imp_D \subseteq \Imp_D'$,    
it directly follows that $\toform{\formF_\ga \cup \formE_\ga \cup \Imp_\ga} \equiv \toform{\formF_\ga' \cup \formE_\ga' \cup \Imp_\ga'}$ and condition~(b) is satisfied.

\medskip

\noindent
\textbf{Rule \propagateW.}
The argument is similar to that for rule \propagateG\ 
by establishing that the premises of the rule entail that that
$\toform{ \formE_D}  \land \toform{ \Imp_D} \rightarrow \LTLglobally(\alpha_1 \land \alpha_2)$.
 
 \medskip
 
\noindent
\textbf{Rule \geninv.}
The premises of the rule guarantee 
for every sequence $\rho \in (\progvars\cup\inputs)^\omega$ that 
if $\langle\rho[0], \rho[1]\rangle \models \gamma$ and for every $i \in \Nat$ it holds that $\langle \rho[i],\rho[i+1]\rangle \models \DedInv_D(q) $, 
then $\rho \models \LTLglobally \alpha$. 
Thus,  the implication $\LTLglobally(\DedInv_D(q)) \rightarrow \LTLglobally(\gamma \rightarrow\LTLglobally\alpha)$ is valid.
Therefore, the implication  $\toform{\Imp_D} \rightarrow \toform{\Imp_D'}$ is valid, and hence the rule is sound.

\medskip

\noindent
\textbf{Rule \geninvp.}
The premise of the rule guarantees that 
$\assignment \models \alpha$ if and only if 
there exists a sequence $\rho \in (\progvars\cup\inputs)^*$ such that 
$\langle\rho[0],\rho[1]\rangle \models \gamma$, 
for every $i < |\rho|$ it holds that $\langle \rho[i],\rho[i+1]\rangle \models \DedInv_D(q) $,  
and for some $\assignment = \rho[|\rho|-1]$.
Intuitively,  $\alpha$ characterizes precisely the set of assignments ``reachable'' from $\gamma$ via a sequence of assignments in which every consecutive pair of assignments satisfies  $\DedInv_D(q)$.
This means that $\alpha$ is invariant on such sequences.
Thus,  the implication 
$\LTLglobally(\DedInv_D(q)) \rightarrow \LTLglobally(\gamma \rightarrow\LTLglobally\alpha)$ 
is valid.
Therefore, the implication $\toform{\Imp_D} \rightarrow \toform{\Imp_D'}$ is valid, 
and hence the rule is sound.

\medskip

\noindent
\textbf{Rule \genreach.}
The premise of the rule guarantees that 
$\assignment \models \gamma$ if and only if 
there exists a sequence $\rho \in (\progvars\cup\inputs)^*$ such that 
$\rho[|\rho|] \models \beta$, 
for every $i < |\rho|$ it holds that $\langle \rho[i],\rho[i+1]\rangle \models \DedInv_D(q)$,  
and $\assignment = \rho[0]$.
Intuitively,  $\gamma$ characterizes a set of assignments from which every sequence of assignments in which every consecutive pair satisfies  $\DedInv_D(q) $ eventually ``reaches' $\beta$.
Thus,  
we have that the implication 
$\LTLglobally(\DedInv_D(q))  \rightarrow \LTLglobally(\gamma \rightarrow\LTLeventually\beta)$ is valid.
Therefore, the implication  
$\toform{\formE_D} \wedge \toform{\Imp_D} \rightarrow \toform{\Imp_D'}$ 
is valid, and hence the rule is sound.

\medskip

\noindent
\textbf{Rules \chainimp,\chainimpG,\chainimpF,\chainimpN,\joinimp.}
All of these rules ensure that
$\toform{\formF_D} \equiv \toform{\formF_D'}$ and 
$\toform{\formE_D} \equiv \toform{\formE_D'}$ for $D \in \agset$, and that 
$\toform{\formE_\as} \land \toform{\Imp_\as} \equiv \toform{\Imp_\as'}$ and 
$\toform{\formE_\as \cup \formE_\ga} \land \toform{\Imp_\ga} \equiv \toform{\Imp_\ga'}$.
This implies $\toform{\formF_\ga \cup \formE_\ga \cup \Imp_\ga^1} \equiv \toform{\formF_\ga' \cup \formE_\ga' \cup \Imp_\ga'}$ as well as condition~(b).
\end{proof}

\begin{lemma}\label{lem:proprules}
    Let $M$ be a monitor, $\Phi \in \templogic(\specvars)$, $q \in Q$ and $\pi\cdot\nu\cdot\rho \in \assignments{\progvars\cup\inputs}^\omega$ for $\nu \in \assignments{\progvars\cup\inputs}$.
    If \Cref*{eq:monitor-correctness} and \Cref*{eq:imp-correctness} hold for $q \in Q$ 
    then they also hold on $\applyRules(q) \in Q$. 
\end{lemma}
\begin{proof}
    First note that~\Cref*{def:sound-transform} is closed under transitivity. 
    Hence,~\Cref*{def:sound-transform} also holds on $\applyRules$.
    To this end let $q = q_1$ and $\applyRules(q) = q_2$ with the symbols of \Cref*{def:sound-transform}.
    To prove \Cref*{eq:imp-correctness} it suffices to prove the stronger statement that $\toform{\formE_\as^2} \to \toform{\Imp_\as^2}$ and $\toform{\formE_\as^2 \cup \formE_\ga^2} \to \toform{\Imp_\ga^2}$ are valid. 
    This follows from the precondition and \Cref*{def:sound-transform}.
\end{proof}

\begin{lemma}\label{lem:propnextstate}
    Let $M$ be a monitor, $\Phi \in \templogic(\specvars)$, $q \in Q$ and $\pi\cdot\nu_1\cdot\rho \in \assignments{\progvars\cup\inputs}^\omega$ for $\rho = \nu_2 \cdot \rho'$ and $\nu_1, \nu_2 \in \assignments{\progvars\cup\inputs}$
    If \Cref*{eq:monitor-correctness} and \Cref*{eq:imp-correctness} hold for $q \in Q$ 
    then they also hold on $\nextState(q, a)$ for $a = \{p \in \preds\mid \langle \assignment_1,\assignment_2\rangle \models p\}$ 
    (with $\pi \cdot \nu_1$, $\nu_2$, and $\rho'$ as the respective trace elements).
\end{lemma}
\begin{proof}
    \Cref*{eq:monitor-correctness} follows directly from \Cref*{lem:liftexpansion} and the precondition that the language equivalence holds on $q$:
        \[ \pi \cdot \nu_1 \cdot (\nu_2 \cdot \rho') \models \Phi  \iff \nu_1 \cdot \nu_2 \cdot \rho' \models \stateformula(q) \iff  \nu_2 \cdot \rho' \models \stateformula(\nextState(q, a)) \]
    For \Cref*{eq:imp-correctness} the expansion results shift in a similar fashion.
\end{proof}

\begin{lemma}\label{lem:propertyovermonitor}
Let $M$ be a monitor constructed from $\Phi \in \templogic(\specvars)$ as described in \Cref*{sec:monitor-states}, \Cref*{sec:expansion}, and \Cref*{sec:rules}.
Then $M$ satisfies the monitor-state correctness property~(\ref*{eq:monitor-correctness}) and the implied-state correctness property~(\ref*{eq:imp-correctness}).
\end{lemma}
\begin{proof}
We prove this statement by induction over the length of $\pi \in \assignments{\progvars \cup \inputs}^*$ in both statements.
Note that we keep $\rho$ and $q$ in the inductive statement quantified. 

\textbf{Case} $|\pi| = 1$ where $\pi = \epsilon \cdot \nu$ for $\nu \in  \assignments{\progvars \cup \inputs}$: 
By definition, $\delta_M^*(\epsilon \cdot \nu) = \qinit$.
As $\stateformula(\qinit) = \Phi$ the monitor-state correctness property trivially holds. 
Furthermore, as initially for $\qinit$, $\Imp_\as$ and $\Imp_\ga$ are empty, the implied-state correctness property holds.

\textbf{Case} $|\pi| > 1$ where $\pi = \pi' \cdot \nu_1 \cdot \nu_2$ for $\pi' \in \assignments{\progvars \cup \inputs}^*$ and $\nu_1, \nu_2 \in \assignments{\progvars \cup \inputs}$ :
Let $a := \{p \in \preds\mid \langle \assignment_1,\assignment_2\rangle \models p\}$ be then, by definition
$\delta_M^*(\pi) = \delta(q', a) = \applyRules(\nextState(q', a))$ 
where $q' = \delta^*_M(\pi'\cdot\assignment_1)$.
By induction hypothesis the monitor-state correctness property and implied-state correctness property hold from $q'$.
By \Cref*{lem:propnextstate} and \Cref*{lem:proprules} the also hold for $q$.
\end{proof}

\begin{lemma}\label{lem:propertyaftertransform}
Let $M$ be a monitor, such that the monitor-state correctness property~(\ref*{eq:monitor-correctness}) and the implied-state correctness property~(\ref*{eq:imp-correctness}) hold for all $q \in Q$.
Then after applying the processing in \Cref*{sec:liveness}, the monitor-state correctness property~(\ref*{eq:monitor-correctness}) sill holds for all $q \in Q$.
\end{lemma}
\begin{proof}
In the first computation step, by \Cref*{eq:imp-correctness} we get for $Q_{\tiny\LTLeventually \beta}$ indeed only state where $\LTLeventually \beta$ holds on all traces where the assumptions hold and $\formE_\ga$ hold.
By \Cref*{eq:monitor-correctness} those hold if and only if the traces hold on $\Phi$
As $\LTLeventually \beta$ holds on all trace suffixes, by the second and third step $\LTLglobally\LTLeventually \beta$ is indeed true on all traces where $\Phi$ holds.
Hence, removing it maintains \Cref*{eq:monitor-correctness}.
\end{proof}

\restateMonitorCorrectness*
\begin{proof}
We already argue in~\Cref*{sec:verdict-labelling} why the constructed monitor is well-formed.
\Cref*{lem:propertyovermonitor} and \Cref*{lem:propertyaftertransform} imply that the monitor-state correctness property holds for all states $q \in Q$ of the monitor.
Hence, it remains to show that we assign the correct verdict according to \Cref*{def:monitor-formula}.

Let $\pi \in {(\assignments{\progvars\cup\inputs})}^*$, $\nu \in \assignments{\progvars\cup\inputs}$, and let $q=\delta^*_M(\pi \cdot \nu)$.
Furthermore, let $\rho \in (\assignments{\progvars\cup\inputs})^\omega$.

If $\verdict(q) = \UNSAT$ then $\stateformula(q) = \bot$, and hence, $\nu \cdot \rho \not\models \stateformula(q)$.
As the monitor-state correctness holds, we can conclude $\pi \cdot \nu \cdot \rho \not\models \Phi$, which proves the condition.

If $\verdict(q) = \SAFETY$ and for $q_i := \delta^*_M(\pi \cdot \nu \cdot \rho[0,i])$ such that $q_0 = q$ and for all $i \in \Nat$, $\verdict(q_i) \neq \UNSAT$ then by \Cref*{def:monitor} and the definition of $\verdict$, $\stateformula(q_i)$ is syntactic safety and $\stateformula(q_i) \neq \bot$.
Hence, expansion by \Cref*{lem:liftexpansion} implies that $\nu \cdot \rho[0,i]$ is not a bad prefix of the safety language $\lang{\stateformula(q)}$. 
Hence, by the property of a safety language, $\nu \cdot \rho \in \lang{\stateformula(q)}$.
As the monitor-state correctness holds, we can conclude $\pi \cdot \nu \cdot \rho \in \lang{\phi}$ which proves this case.

\end{proof}

\end{document}